\documentclass[%
reprint,
amsmath,amssymb,
aps,
]{revtex4-1}
\usepackage{amsthm}

\usepackage{graphicx}
\usepackage{dcolumn}
\usepackage{bm}
\usepackage{hyperref}
\usepackage{pgfplots}

\usepackage[caption=false]{subfig}
\newtheorem{theorem}{Theorem}
\newtheorem{lemma}[theorem]{Lemma}

\newtheorem{proposition}[theorem]{Proposition}

\newtheorem{remark}[theorem]{Remark}

\DeclareMathOperator{\Tr}{Tr}
\DeclareMathOperator{\diag}{diag}

\newcommand{\bw}{\textbf{w}}

\begin{document}
	
	\preprint{APS/123-QED}

		\title{ Designing Optimal   Multiplex Networks for Certain Laplacian Spectral Properties		}
\author{Heman Shakeri$^{1}$}
\author{Ali Tavasoli$^{2}$}
\author{Ehsan Ardjmand$^{3}$}
\author{Pietro Poggi-Corradini$^4$}
\affiliation{$^1$School of Data Science, University of Virginia, Charlottesville, Virginia, USA}
\affiliation{$^2$Department of Mechanical Engineering, Payame Noor University, Tehran, Iran}
\affiliation{$^3$Department of Analytics and Information Systems, College of Business, Ohio University, OH, USA}
\affiliation{$^4$Mathematics Department, Kansas State University, Manhattan, Kansas, USA}

\date{\today}

	\begin{abstract} 
 We discuss  the design of interlayer edges in a multiplex network, under a limited budget, with the goal of  improving its overall performance. We analyze the following three problems separately; first, we maximize the smallest nonzero eigenvalue, also known as the algebraic connectivity; secondly, we minimize the largest eigenvalue, also known as the spectral radius; and finally, we minimize the spectral width.
	Maximizing the algebraic connectivity requires identical weights on the interlayer edges  for budgets less than a threshold value. However, for larger budgets, the optimal weights are generally non-uniform. The dual formulation transforms the problem into a graph realization (embedding) problem that allows us to give a fuller picture. Namely, before the threshold budget, the optimal realization is one-dimensional with nodes in the same layer embedded to a single point; while, beyond the threshold, the optimal embeddings generally unfold into spaces with dimension bounded by the multiplicity of the algebraic connectivity. Finally, for extremely large budgets the embeddings revert again to lower dimensions.
Minimizing the largest eigenvalue is driven by the spectral radius of the individual networks and its corresponding eigenvector. Before a threshold, the total budget is distributed among interlayer edges corresponding to the nodal lines of this eigenvector, and the optimal largest eigenvalue of the Laplacian remains constant. For larger budgets, the weight distribution tends to be almost uniform. In the dual picture, the optimal graph embedding is one-dimensional and non-homogeneous at first, with the nodes corresponding to the layer with largest spectral radius distributed on a line according to its eigenvector, while the other layer is embedded at the origin. Beyond this threshold, the optimal embedding expands to be multi-dimensional, and for larger values of the budget, the two layers fill the embedding space.
Finally, we show how these two problems are connected to minimizing the spectral width.
	\end{abstract}
	
%
	\maketitle

\section{Introduction}

Multiplex networks consist of distinct layers interacting together in diverse social, economical, transportation, and biological networks \citep{Kim2013Coevolution,Arenas2014multilayer}. Coupling structure of network layers is shown to affect  connectivity and robustness properties of the entire system \citep{Kim2012Correlated}. With the existing literature mostly limited to single networks, further systematic work is needed to discover the underlying  mechanism in multiplex networks \citep{Arenas2014multilayer}. 

Let $G=\left(V,E\right) $
represent an undirected network and by $V=\left\{ 1,\ldots,n\right\} $
and $E\subset{V\choose 2}$, we denote the set of nodes and links. 
For a link $e$ between nodes $i$ and $j$, i.e., $e:\{i,j\}\in E$, we define a nonnegative value $w_{ij}$ as the weight of the link.
Given $G$  a multiplex network,  let $G_1 = \left\{ V_1,E_1\right\}$ and $G_2=\left\{ V_2,E_2\right\}$, $|V_1| =| V_2|$ represent the layers, and a bipartite graph $G_3=\left\{ V,E_3\right\}$ with $E_3\subseteq \{\{v_1,v_2\}: v_i\in V_i\}$ are connecting the layers (Fig. \ref{fig:Multilayer}). Throughout the paper, we use the term \textit{intralayer} links for $E_1$ and $E_2$, and \textit{interlayer} links for $E_3$.

The links in $G_3$ bridge $G_1$ and $G_2$ and should be chosen strategically, for instance in a way that minimizes the disruption of the flow of information, electric power or goods, or to avoid failures against attackers and possible errors that can  fragment the system or cause  cascading phenomena \citep{Buldyrev2010}.
The edge weights of $G_3$ are the sole design parameters.

The Laplacian matrix is defined as 
\begin{equation}
\label{eq:Laplacian}
L(w):=
\sum_{\{i,j\}\in E_1\cup E_2}L_{ij}+\sum_{\{i,j\}\in E_3}w_{ij}L_{ij},
\end{equation}
where $L_{ij}:= (\delta_i-\delta_j)(\delta_i-\delta_j)^T$, for each link $\{i,j\}$, and $\delta_i$ is the delta function at vertex $i$. In particular, we think of the Laplacian matrix of $G$ as a function of the interlayer weights $w$.
Enumerating the vertices in $V_1$ followed by the vertices in $V_2$, we can write $L(w)$ in block form in terms of the Laplacian matrices of the layers, $L_1$ and $L_2$, as follows:
\begin{equation} \label{eq:LaplacMatrix}
\begin{aligned}
L(w)=
\begin{bmatrix}
L_1+W & -W \\
-W & L_2+W
\end{bmatrix}
\end{aligned}
\end{equation}
In particular, we will assume from now on that  $W=\text{diag}(w)$, meaning that $E_3$ consists of a perfect matching, see Fig. \ref{fig:multiplex}.
\begin{remark}\label{remark:posPert}
	All eigenvalues of $L(w)$ are nondecreasing functions of $w$,
	because \eqref{eq:LaplacMatrix} can be thought of as perturbing $L(w)$ using a positive semidefinite matrix.
\end{remark}
\begin{figure}[!htb]\label{fig:multiplex}
	\centering
	\includegraphics[clip,width=.75\columnwidth]{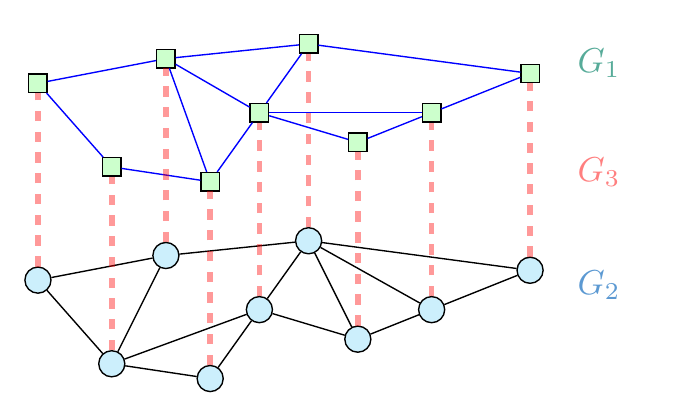}%
	\caption{A multiplex network with two layers $G_1$ and $G_2$ and interlayer link structure in $G_3$.}
	\label{fig:Multilayer}
\end{figure}

%

Our goal is to  allocate weights on the interlayer links, subject to a budget such that $\sum w_{ij} = c$,
 to obtain extremal spectral properties. More specifically, we are interested in maximizing the smallest positive eigenvalue of the Laplacian,  minimizing the largest eigenvalue, and minimizing the difference between the two. After formulating the primal-dual program and deriving its properties, the equivalent graph realization problem in each case is extracted and its features with respect to the multiplex network's structure are identified. 


Recall that the Laplacian matrix $L(\omega) $ is positive semidefinite and has (at least for connected networks) one zero eigenvalue with eigenvector $\boldsymbol{e}=[1,\dots,1]^{T}$, the vector of all ones of appropriate length. The eigenvalues of $ L(\omega) $ are ordered as $ 0=\lambda_1(\omega)\leq\lambda_2(\omega)\leq\lambda_3(\omega)\leq\dots\leq\lambda_{n}(\omega) $.

The main contributions of the paper are traceable under several circumstances. While the literature is limited to uniformly weighted interconnections, we study optimal distributions (not necessarily uniform) of interlayer weights and define three types of optimization criteria based on our inter-structure design objectives. In addition to maximizing  algebraic connectivity $\lambda_2(\omega)$, we consider the problem of minimizing   spectral radius $\lambda_{n}(\omega) $ and finally minimizing the spectral width $\lambda_{n}(\omega) -\lambda_2(\omega)$ with motivating their  numerous applications. 
Most of our results for all three problems are  analytical so that there is no need for resorting to numerical solutions in such conditions. However,   formulating the considered  problems in a convex framework enables  efficient numerical approaches  in regions where there is no analytical solution yet. Moreover, we provide a primal-dual setting that reflects a multi-sided view of each problem; i.e. the geometric dual formulation reduces the primal problems into  graph embedding problems  related to several physical aspects that also enables easier interpretations.

\subsection{Algebraic connectivity $\lambda_2$}
From now on we assume the network to be connected, thus the second smallest eigenvalue is positive in this case, and its magnitude reflects the level of connectedness of the graph. In fact, $\lambda_2$ is called the algebraic connectivity of $G$; \citet{Fiedler1973} showed that algebraic connectivity increases monotonically by adding links.  The algebraic connectivity can also be considered as a measure of network robustness \citep{Jamakovic2007}. Moreover, various bounds in graph partitioning, optimal graph labeling, min-sum problems, or bandwidth optimization can be obtained using $\lambda_2$ as a determining factor \citep{Juvan1993Bandwidth,Helmberg1995Bounds}. The convergence speed of various processes such as mixing Markov chains on graphs \citep{bremaud2013markov}, reaching consensus in multi-agent systems \citep{Jadbabaie2003NearestNeighbor,Olfati-Saber2004Consensus,Olfati-SaberFlock,Jadbabaie2007Flocking}, or diffusion dynamics on networks \citep{GomezDiffusion2013} are controlled by the second smallest eigenvalue of the Laplacian. 

In our multiplex model with varying weights $w$, the algebraic connectivity for constant $w$ grows linearly with $w$ up to a critical $w^*$, and then has a nonlinear behavior afterwards \citet{GomezDiffusion2013}. Bounds for $w^*$ are found in \citet{Radicchi2013} and its exact value is found in \citet{sahneh2014exact}.  In \citep{Martin-Hernandez2014}, the structure of the interlayer links $G_3$ is studied by comparing two configurations: diagonal (one-to-one) vs. random.

In a single layer  graph $G$, with variable edge-weights subject to a total budget, \citet{boyd2004fastestChain} and \citet{GoringShadowSeperator} discuss how maximizing $\lambda_2$  corresponds to a dual semidefinite optimization problem  and show that the optimal solutions of the dual are related to the eigenvectors of the optimal algebraic connectivity. It turns out that the dual may be interpreted as an embedding of the single-layer graph in $\mathbb{R}^n$ (optimal realization of the graph in Euclidean space), and the optimal embedding has structural properties tightly connected to the separators of the graph. 

In this paper we consider the setting of multi-layer networks with a one-to-one interconnected structure and maximize the algebraic connectivity of the whole given a limited total budget. We address this by formulating and studying the properties of primal and dual problems. \citet{shakeri2015PRL} show that strong duality holds and revisit the fact  that, up to a threshold budget, the solution to the primal problem is the uniform distribution with identical weights. For larger budgets, the optimal weights are generally not uniform. 
The dual formulation transforms the problem into a graph realization (embedding) problem.  We show that before the threshold budget, the optimal realization is one-dimensional and consists of nodes in the same layer clumped together; while, beyond the threshold, the eigenvalues of the Laplacian coalesce and optimal embeddings generally take place in spaces with more than one dimension. Finally, for very large budgets the embedding is again one-dimensional.

\subsection{Laplacian spectral radius $\lambda_n$}
The largest eigenvalue $\lambda_n$ provides useful information about graph structural properties (see \citet{SusannaThesis} and references therein), various bounds for algebraic properties in graphs \citep{Juvan1993Bandwidth,Helmberg1995Bounds}, and applications in game theory \citep{li2009largest}.
The largest eigenvalue must be sufficiently small for stability of formation control algorithms when the
agent dynamics are prone to high-gain instability or unmodeled dynamics \citep{Bai2010InstabilityMechanisms}. Optimization of the largest eigenvalue is related to finding a weighted tree with the largest spectral radius of the Laplacian matrix \citep{Tan2010WeightedTrees,Li2011WeightedTrees}.  \citet{Fiedler1990LambdaMax} considered the problem of minimizing the maximum eigenvalue of the weighted Laplacian for trees and bipartite graphs and investigated connections with doubly stochastic matrices. \citet{goring2012MaxLambda} studied the problem of minimizing the maximum eigenvalue of the Laplacian of single-layer graphs, and by transforming the corresponding  dual problem studied the graph realization problem. 

\subsection{Spectral width $\lambda_n-\lambda_2$}
The difference between the second smallest and largest eigenvalue, or the \textit{spectral width}, provides important information about several bounds on different graph metrics. For instance, the value of a uniform sparsest cut falls in the interval bounded by the second smallest and the largest eigenvalue of the Laplacian divided by the number of nodes (see \citet[Lemma 4.1]{beineke2004topics}).
Moreover, a small spectral width implies that the graph is Hamiltonian, i.e. a graph possessing a cycle visiting each vertex exactly once \citep{Butler2010Hamiltonian}. Also instability can occur in cooperative motion control when the frequency of the parametric perturbation is near an eigenvalue of the Laplacian \citep{Bai2010InstabilityMechanisms}. Thus, minimizing the spectral width will shrink the frequency band of parametric resonance. \citet{goring2013Gap} studied the problem of minimizing the spectral width and discovered connections between this hybrid problem and the separate problems of maximizing the algebraic connectivity and minimizing the spectral radius. 
In this paper, the equivalent  problem for multiplex networks is considered by deriving and studying the primal and dual embedding problems. In particular, we show that the primal and embedding problems corresponding to spectral width minimization can be related to the individual problems of maximizing $\lambda_2$ and minimizing $\lambda_n$.


\subsection{Graph realization problem}
Inspired by \citet{Boyd2006FastestMP}, 
we reformulate the dual problem  as a graph realization problem and study the connections to the graph's structure. Since the Laplacian is a positive semidefinite matrix, the solution of the dual program can be represented by a Gram matrix. The dual problem is seen to be equivalent to a realization of the graph in $\mathbb{R}^n$, by assigning the graph nodes to the vectors of the Gram representation. It is of particular interest to find the smallest dimension for which a graph is realizable. Graph realizations related to extremal Laplacian eigenvalues have close connections with problems in other areas. For example, graph realization problems arise in the determination of molecular conformation \citep{Hendrickson1995MoleculeProblem}. An interesting problem also appears in manifold learning where the structure of a low dimensional manifold is constructed from sampled high dimensional data  \citep{Weinberger2004Kernel, Weinberger2006UnsupervisedLearning}.\\


In Sections \ref{sec:lambda2}, \ref{sec:lambdan}, and \ref{app:gap}, we formulate maximizing the second smallest eigenvalue, minimizing the maximum eigenvalue, and minimizing the difference between the second smallest and the largest eigenvalue in multiplex networks, respectively. We inspect and demonstrate the properties of primal-dual and embedding problems in each case. Section \ref{sec:conclusion} is devoted to concluding remarks.

\section{Maximizing $\lambda_2$}\label{sec:lambda2}
This section examines properties of the graph resulting from maximizing the second smallest eigenvalue under a budget constraint on the interlayer edge weights. The original work for single-layer graphs was proposed by \citet{FiedlerLaplacianOfGraphs} who maximized the graph connectivity under the budget constraint that the total of the edge weights equals the number of edges. \citet{shakeri2015PRL}  follow \citet{GoringShadowSeperator,goring2011rotational} and propose the following formulation for  multi-layer networks 
\begin{equation}\label{lambda_2Primal}
\begin{aligned}
& \underset{w_{ij}, \lambda_2, \mu}{\text{maximize}}
& & \lambda_2 \\
& \text{subject to}
& & \sum_{ij\in E_3}w_{ij}L_{ij}+L_0+\mu \boldsymbol{e} \boldsymbol{e}^T - \lambda_2 I\succeq 0 \\
& & & \sum_{ij\in E_3}w_{ij}= c \\
& & & w_{ij}\geq 0 
\end{aligned}
\end{equation}
where $L_0 = \sum_{ij\in E_1\cup E_2}L_{ij}$ is the Laplacian for the disjoint union of the layers. In the semidefinite constraint, the free variable $\mu$ serves to shift the zero eigenvalue, and ensure that $\lambda_2$ becomes the smallest eigenvalue of $\sum_{ij\in E_3}w_{ij}L_{ij}+L_0+\mu \boldsymbol{e} \boldsymbol{e}^T $. In this regard, when the optimal solution $ (w_{ij}^*, \lambda_2^*, \mu^*) $ is attained, $\lambda_2^*$ is the smallest eigenvalue of $\sum_{ij\in E_3}w_{ij}^*L_{ij}+L_0+\mu^* \boldsymbol{e} \boldsymbol{e}^T $ or equivalently the second smallest eigenvalue of $\sum_{ij\in E_3}w_{ij}^*L_{ij}+L_0 $. \\

\citet{shakeri2015PRL} show that, for the special case of multiplex networks with identical layers the uniform weight distribution is always optimal for \eqref{lambda_2Primal}. Moreover, the corresponding optimal algebraic connectivity increases linearly with the budget $c$ up to a threshold $c^*>0$ and then remains constant after that. When the layers are not identical, the 
uniform distribution is optimal only for budgets $c$ up to the threshold $ c^*>0 $, 
and increasing the budget past the threshold, yields nonuniform optimal weight distributions. Therefore, in this case, it is still possible to improve the algebraic connectivity by increasing the budget beyond $ c^* $. Figure \ref{fig:Primal} shows uniform and nonuniform optimal weight distributions in two budget regimes  for a small multiplex.

\begin{figure}[!htb]
	\centering
	\subfloat[]{%
		\includegraphics[clip,width=.8\columnwidth]{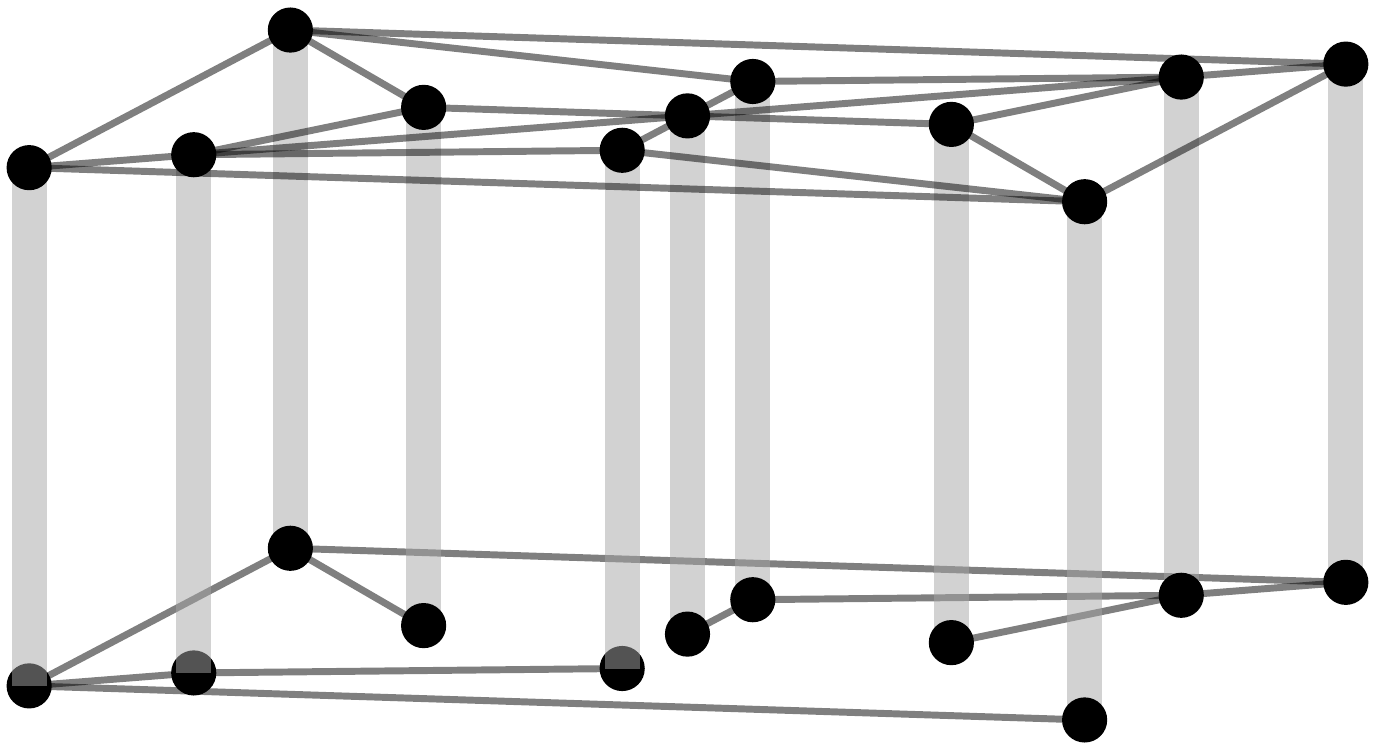}%
	}\\
	\subfloat[]{%
		\includegraphics[clip,width=.8\columnwidth]{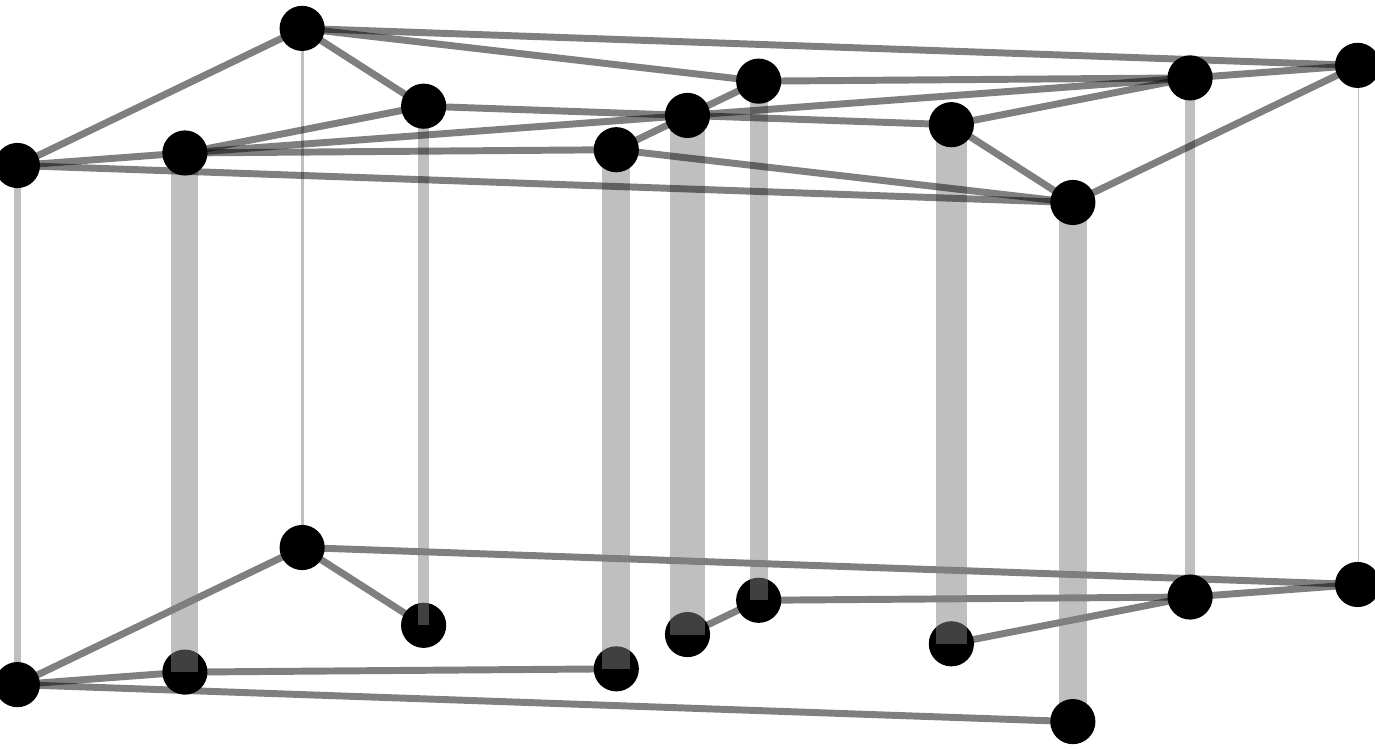}%
	}%
	\caption{Weight distribution for optimal algebraic connectivity: (a) uniform weights for $c\leq c^*$, (b) nonuniform weights for $c>c^*$.}
	\label{fig:Primal}
\end{figure}

More specifically, before the threshold, $\underset{\sum{w_{ij}}= c}{\text{max}}\lambda_2(w_{ij})=\frac{4c}{n}$, $\forall c\leq c^*$,  and beyond the threshold it remains under $4c/n$ for all $c>c^*$. Lacking the knowledge of these regimes casts serious challenges on proposing heauristics similar to \citet{ghosh2006growing}.

The threshold budget is calculated in \citep{sahneh2014exact} (and in \citet{shakeri2015PRL} with a different proof) to be
\begin{equation}
\begin{aligned}
c^*=\frac{n}{2}\lambda_2\left[\left(L_1^\dagger+L_2^\dagger\right)^\dagger\right]
\end {aligned}
\end{equation}
where $ ^{\dagger}$ denotes the Moore-Penrose pseudoinverse. An upper-bound for the algebraic connectivity is also given in \citet{Radicchi2013}
\begin{equation}\label{ConnecBound} \lambda_2\left[L\left(w^*\right)\right] \le \lambda_2(L_{ave}) 
\end{equation}
where $L_{ave} = \frac{L_1+L_2}{2}$.
%
Figure \ref{fig:lambda2} compares the optimal value of algebraic connectivity to the one obtained by the uniform distribution as the budget $c$ varies past the threshold, for four different
random layer structures (see SM \ref{section:NetModels} for network models). In all cases, the optimal distribution gives a higher algebraic connectivity after the threshold. We also observe that for $c>c^*$, $\lambda_2$ increases at a slower rate than before the threshold. \\

In Figure \ref{fig:lambda2ER} , we consider two Erd\H{o}s-R\'enyi layers and vary the difference of the algebraic connectivity of the layers, i.e. we look at small and large values of $\left|\lambda_2(L_1)-\lambda_2(L_2)\right|$. In the case when the layers have similar algebraic connectivities, the threshold $c^*$ has a larger value. This enlarges the linear part of the diagram, thus postponing the nonlinear region and its slow growth. As a result, nodes in multiplex graphs with individual layers sharing near connectivity properties are synchronized more easily (with less budget) and superdiffusion occurs \citep{GomezDiffusion2013}. To achieve the same degree of interdependent connectivity when layers have much different connectivity properties, a larger budget is required. In the extreme case of identical layers, the upperbound \eqref{ConnecBound} is achieved at the threshold. Similar behaviors are observed for the other network models. 
\begin{figure}[!htb]
	\subfloat[]{\includegraphics[clip,width=.5\columnwidth]{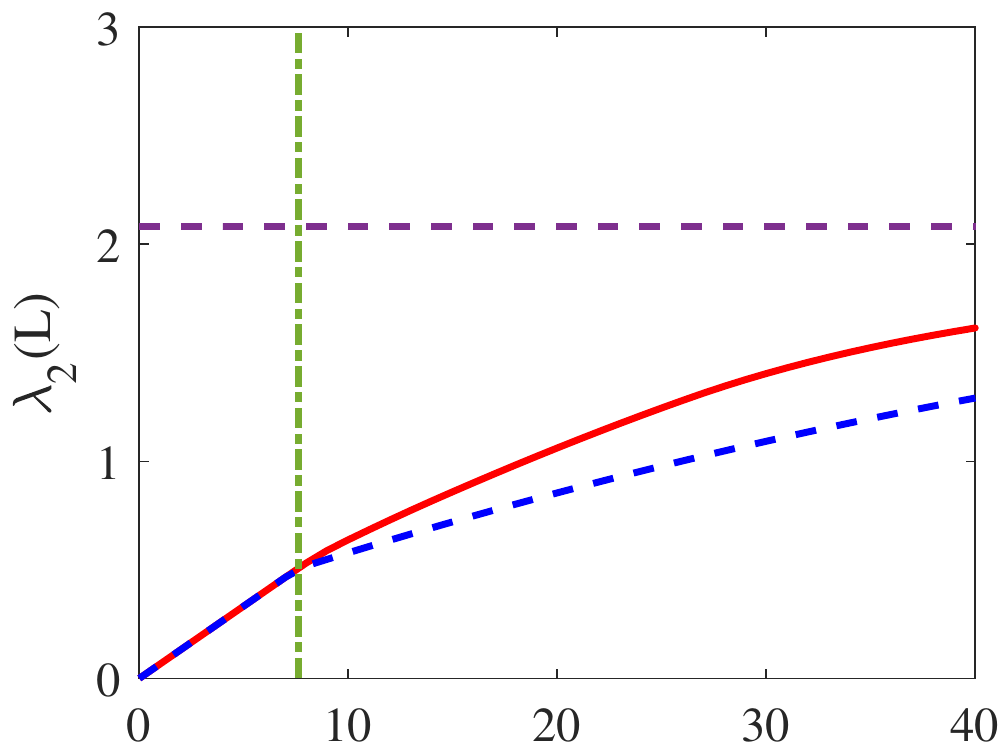}          	
	} 
	\subfloat[]{\includegraphics[clip,width=.5\columnwidth]{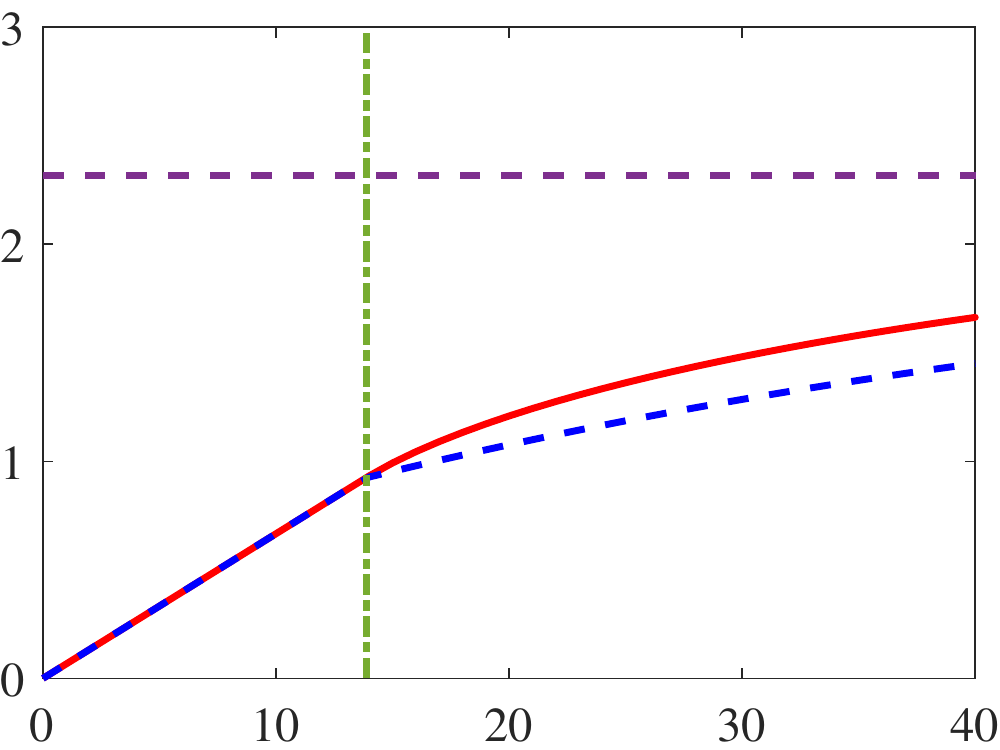}          	
	}
	\\
	\subfloat[]{\includegraphics[clip,width=.5\columnwidth]{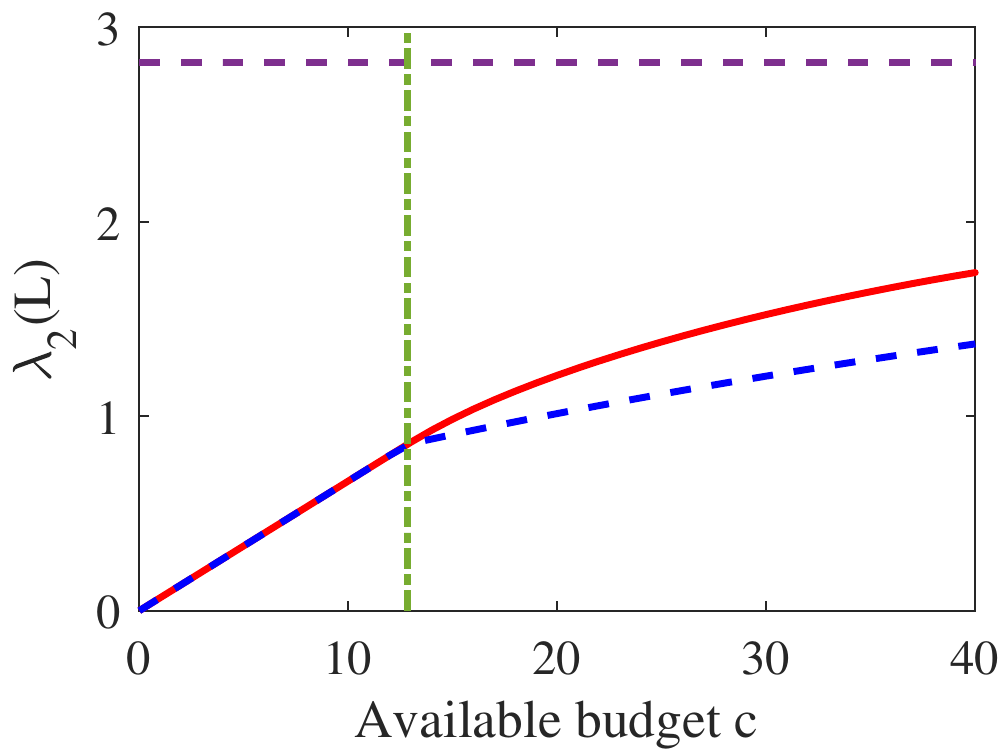}          	
	} 
	\subfloat[]{\includegraphics[clip,width=.5\columnwidth]{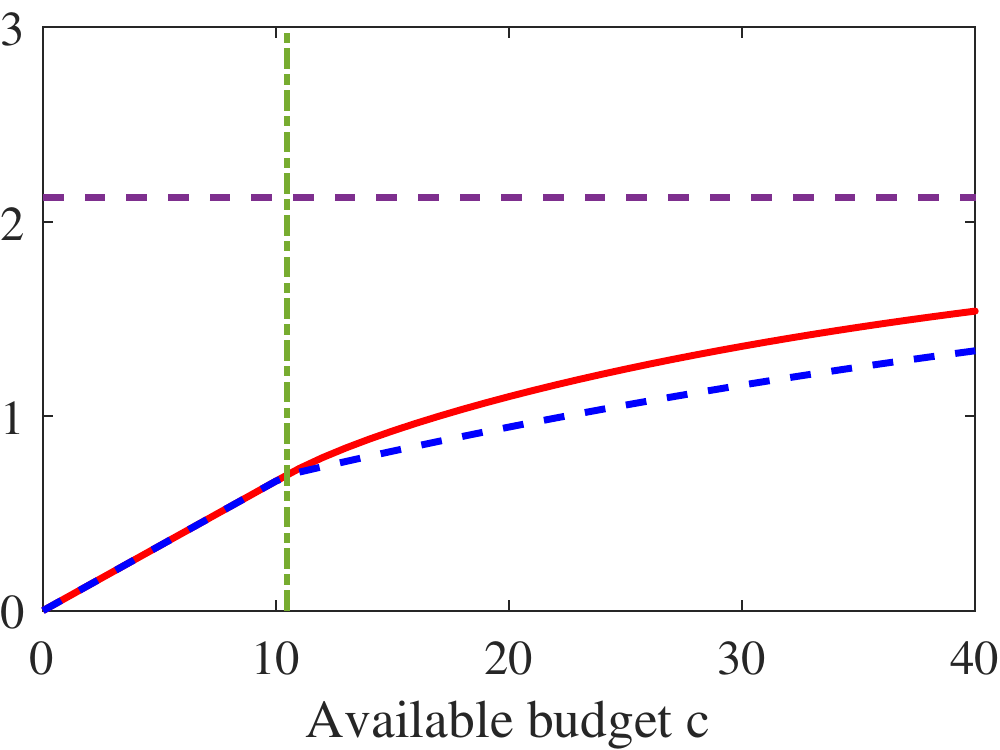}          	
	} 
    \\
    \subfloat[]{\includegraphics[clip,width=.6\columnwidth]{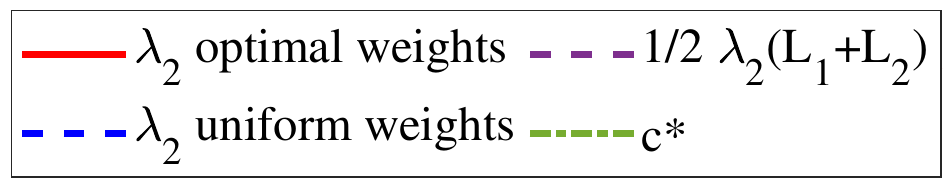}          	
    } 
	\caption{Algebraic connectivity for two different (a) Barab\'asi-Albert scale free, (b) Erd\H{o}s-R\'enyi, (c) random geometric, and (d) Watts-Strogatz networks with individual networks of 30 nodes (see SM \ref{section:NetModels} for network models).}
	\label{fig:lambda2}
\end{figure}

\begin{figure}[!htb]
	\subfloat[]{\includegraphics[clip,width=.5\columnwidth]{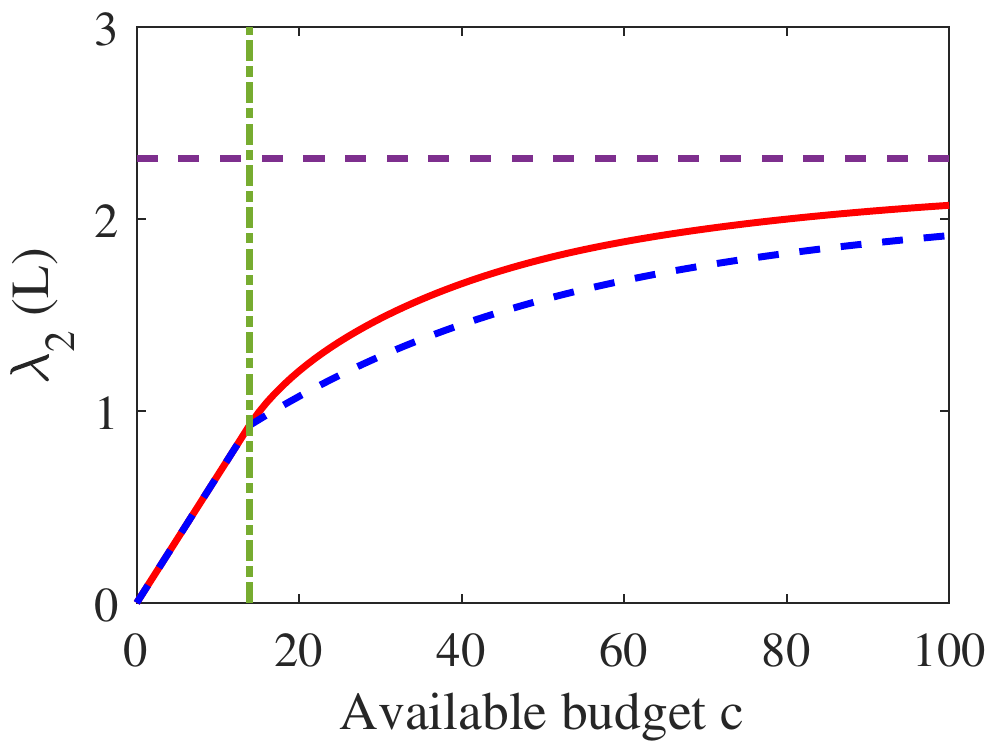}          	
	} 
	\subfloat[]{\includegraphics[clip,width=.5\columnwidth]{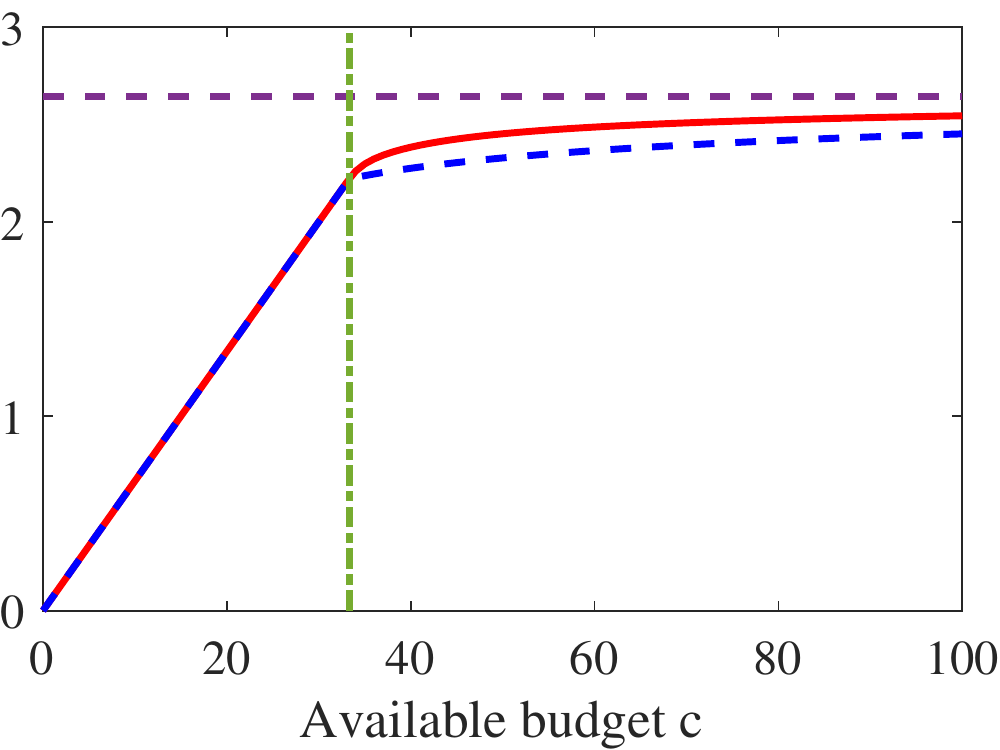}          	
	}\\
    \subfloat[]{\includegraphics[clip,width=.6\columnwidth]{Legend1.pdf}          	
    }	
	\caption{Algebraic connectivity for network with interconnected Erd\H{o}s-R\'enyi layers with (a) large  difference between the algebraic connectivities of individual networks ($\lambda_2(L_1)=0.5105$ and $\lambda_2(L_2)=2.3485$), and (b) small difference between the algebraic connectivities of individual networks ($\lambda_2(L_1)=1.4752$ and $\lambda_2(L_2)=1.6723$).}
	\label{fig:lambda2ER}
\end{figure}
\subsection{Dual and embedding problems for $\lambda_2$}\label{subsec:Lam2Dual}
The dual problem for \eqref{lambda_2Primal} is obtained by the Lagrangian approach \citep{boyd2004convex},  (see SM \ref{app_lagrangian}), and can be written as follow:
\begin{equation}\label{lambda_2Dual}
\begin{aligned}
& \underset{\xi\in\mathbb{R}, X\in\mathbb{R}^{n\times n}}{\text{maximize}}
& & c\xi -\langle X, L_0\rangle  \\
& \text{subject to}
& & \langle X,I\rangle = 1 \\
& & & \langle X,ee^T\rangle = 0\\
& & & \langle X,L_{ij}\rangle \leq -\xi ~~\quad\forall\lbrace i,j\rbrace\in E_3\\
& & & X\succeq 0 
\end{aligned}
\end{equation}
where  $\langle X, L_0\rangle =\Tr(L_0^TX) = \sum_{\lbrace i,j\rbrace \in E_1\cup E_2}x_{ii}+x_{jj}-2x_{ij}$.

\begin{proposition}\label{lem:FeasibleDualSolution}
	The feasible set of the dual problem is not empty.
\end{proposition}
\begin{proof}
	See SM \ref{proof:FeasibleDualSolution}.
\end{proof}
\begin{proposition}\label{lem:StrongDuality}
	Strong duality holds for the primal and dual problems \eqref{lambda_2Primal} and \eqref{lambda_2Dual}, and the dual problem attains its optimal solution.
\end{proposition}
\begin{proof}
	See SM \ref{proof:StrongDuality}.
\end{proof}

Since $X$ is a positive semi-definite matrix,  we can use its Gram representation $X = U^TU$, where $U\in \mathbb{R}^{n\times n}$ \citep{Boyd2006FastestMP}, and rewrite \eqref{lambda_2Dual} as:
\begin{equation}\label{Lambda2Embedding}
\begin{aligned}
& \underset{\xi\in \mathbb{R}, u_i\in \mathbb{R}^n}{\text{maximize}}
& & c\xi -\sum_{\lbrace i,j\rbrace\in E_1\cup E_2}\| u_i-u_j \|^2 \\
& \text{subject to}
& & \sum_{i\in V}\|u_i\|^2 = 1\\
& & & \sum_{i\in V}u_i = 0\\
& & &  \| u_i - u_j\|^2\leq -\xi \quad \ \forall \lbrace i,j\rbrace\in E_3\\
\end{aligned}
\end{equation}
%

\begin{proposition}\label{lem:ProjectionOfEmbedding}
	The projections of optimal embedding onto one-dimensional subspaces yield eigenvectors for the algebraic connectivity.
\end{proposition} 
\begin{proof}
	See SM \ref{proof:ProjectionOfEmbedding}.
\end{proof}

%
%

For connected single-layer networks, \citet{GoringShadowSeperator}   scale the weights by $c\lambda_2\neq0$, i.e., they consider $\frac{w_{ij}}{c\lambda_2} = \hat{w}_{ij}$, where $c$ is the budget.  We can use the same procedure here, 
and obtain a scaled version of the primal-dual problem for multiplex networks. The scaled primal problem of \eqref{lambda_2Dual} is 
\begin{equation}\label{eq:PrimalScaled}
\begin{aligned}
& \underset{\hat{w}_{i,j}\in\mathbb{R}^{E_3}}{\text{minimize}}
& & \sum_{\lbrace i,j\rbrace\in E_3}\hat{w}_{ij} \\
& \text{subject to}
& & c\sum_{i,j\in E_3}\hat{w}_{ij}L_{ij}+(\sum_{\lbrace i,j\rbrace\in E_3}\hat{w}_{ij})L_0\\
& & & \ \  \  \ +\hat{\mu} \boldsymbol{e} \boldsymbol{e}^T - I\succeq 0 \\
& & & \hat{w}_{ij}\geq 0 \ \ \forall\lbrace i,j\rbrace\in E_3
\end{aligned}
\end{equation}
and the scaled dual (embedding) problem \eqref{Lambda2Embedding} is written as (see SM \ref{app_lagrangian})
\begin{equation}\label{eq:EmbeddingScaled}
\begin{aligned}
& \underset{\hat{u}_i\in\mathbb{R}^n}{\text{maximize}}
& & \sum_{i\in V}\| \hat{u}_i \|^2 \\
& \text{subject to}
& & c\| \hat{u}_i - \hat{u}_j\|^2+\sum_{\lbrace k,l\rbrace \in E_1}\| \hat{u}_k - \hat{u}_l\|^2+\\
& & &  \sum_{\lbrace k,l\rbrace \in E_2}\| \hat{u}_k - \hat{u}_l\|^2\leq 1 \ \ \forall\lbrace i,j\rbrace\in E_3\\
& & & \sum_{i\in V}\hat{u}_i= 0
\end{aligned}
\end{equation}
It is known that there are transformations that map optimal solutions of the primal and dual problems to the scaled ones \citep{SusannaThesis}. Together with Proposition \ref{lem:StrongDuality}, this shows that strong duality holds also for the scaled problem. 

\begin{proposition}\label{lem:clump}
	For budget values up to the threshold $c^*$, the optimal solution of the embedding problem \eqref{eq:EmbeddingScaled} is given as
	\begin{equation}\label{eq:UniformEmbed}
	\begin{aligned}
	\hat{u}_i^*=\begin{cases}
	h & \text{if}\ i \in V_1 \\ 
	-h, & \text{if}\ i \in V_2
	\end{cases}
	\end{aligned}
	\end{equation} 
	where $\boldsymbol{h}=\langle h\rangle$ is a one-dimensional subspace.
\end{proposition}

\begin{proof}
	See SM \ref{Proof:clump}.
\end{proof}
The embedding \eqref{eq:UniformEmbed} implies each layer clumps together at opposite ends, in the optimal realization (Figure \ref{fig:clumped}). In this case, the Fiedler cut distinguishes the individual layers \citep{van2010graph}. 


%

\begin{remark} \label{rem:EmbDim}
	The algebraic multiplicity of $\lambda_2(L)$ sets an upper-bound on the dimension of realization \citep{Helmberg2010}. This can be understood from Proposition \ref{lem:ProjectionOfEmbedding} by recalling that the dimension of the eigenspace corresponding to $\lambda_2$ is at most the multiplicity of $\lambda_2$. 
\end{remark}

%
%

\subsection{Interpretation of the embedding problem}

Here we expand on the interpretation of problem \eqref{eq:EmbeddingScaled}. The goal is to maximize the spread (variance) of the vectors $\hat{u}_i$ subject to a constraint involving the neighbor relations (and a fixed barycenter).

Note that the budget $c$ appears uniquely in the inequality constraint.  For small values of $c$, this inequality puts less limitations on the embedding distances between connected nodes via the interlayer links than the total embedding distances between nodes within the layers.
As the budget increases, the interlayer distances pay a higher and higher toll and the intralayer ones gain more freedom (we say that they ``unfold''). In particular, the layer with lower algebraic connectivity enjoys  more flexibility and will unfold faster.


Another interpretation of problem \eqref{eq:EmbeddingScaled} consists of thinking of the vectors $\hat{u}_i$ as the positions of some unit masses subject to repelling and attracting forces.



Here we follow the mechanical interpretation given in  \citet[Sec. 4.6]{Boyd2006FastestMP} in the case of one layer. Similarly, we define a potential energy  $\mathcal{U}$ for the $n$-point,  
$$\mathcal{U} = -\frac{1}{2}\sum_{i<j\in V_1\cup V_2}\Vert \hat{u}_i-\hat{u}_j\Vert^2.$$
By rotation and translation invariance,  minimizing $\mathcal{U}$ is seen to be equivalent to maximizing the objective function in \eqref{eq:EmbeddingScaled}. 
The constraints in \eqref{eq:EmbeddingScaled} can be seen as bounds on the elastic potential energy along the edges.


The total force on each node can be found by differentiating
this energy function and the resultant force is toward the origin with a magnitude proportional to the distance of the point. Consequently, we can write the static equilibrium condition as
\begin{equation}\label{eq_equilibrium}
\sum_{j\sim i}T_{ij}\frac{\hat{u}_i -\hat{u}_j}{\Vert \hat{u}_i -\hat{u}_j\Vert} = \hat{u}_i \quad \forall i\in V
\end{equation}
where $T_{ij}$ is the tension in the imaginary link (similar to a spring) between the two nodes $i$ and $j$. Since the interlayer connections are one to one, assume without loss of generality that $i\in G_1$, and  $k\in G_2$ is the matched neighbor of $i$ in the other layer. Then \eqref{eq_equilibrium} becomes
\begin{equation}
T_{ik}\frac{\hat{u}_i -\hat{u}_k}{\Vert \hat{u}_i -\hat{u}_j\Vert} +\sum_{\substack{j\sim i\\j\in G_1}}T_{ij}\frac{\hat{u}_i -\hat{u}_j}{\Vert \hat{u}_i -\hat{u}_j\Vert}=\hat{u}_i,
\end{equation}
Therefore, 
at each node, the sum of forces due to interlayer and intralayer springs is balanced by a repulsive force  that is proportional to the distance of the node to the origin.
We can use this to explain the embedding configurations, by combining the inequality constraints with the orientation of the force field.
For example, by Proposition \ref{lem:clump}, for budgets $c<c^*$ below the threshold, 
we have 
$$\sum_{j\in G_1}a_{ij}T_{ij}\frac{\hat{u}_i -\hat{u}_j}{\Vert \hat{u}_i -\hat{u}_j\Vert}=0$$ 
hence $T_{ik}=w_{ik}$, in this case.

In Figure \ref{fig:Emb1}, we consider two layers (of Watts-Strogatz type)  whose algebraic connectivities are very close. In Figure \ref{fig:Emb1a}, we plot the behavior of $\lambda_2$ for the resulting multiplex, and observe three different regimes as the  budget increases. For $c<c^*$, $\lambda_2$ is simple and grows linearly, after the threshold, the multiplicity increases to two. Finally, there is another phase shift $c^{**}>c^*$, so that for large budgets $\lambda_2$ becomes simple once again.

Recall that  by Remark \ref{rem:EmbDim}, the multiplicity of $\lambda_2$ provides an upperbound for the 
 the embedding dimension in each regime. This can be seen is in  Figures \ref{fig:Emb1b}-\ref{fig:Emb1f} shows. For low budgets the embedding looks like Figure \ref{fig:clumped}, i.e., the nodes are clumped in each layer. In the second regime,  when $c^*<c<c^{**}$, due to the nonuniform interlayer forces interacting with intralayer forces, the embedding unfolds two-dimensionally, in Figures \ref{fig:Emb1b} and \ref{fig:Emb1c}. Notice that the different layers are still distinguishable in two-dimensional embeddings for budgets just above the threshold $c^*$. 
By further increasing $c$ towards $c^{**}$, the stiffened interlayer springs draw the two layers towards each other, thus decreasing the distance between them (see \ref{fig:Emb1d} and Figure \ref{fig:Emb1e}); and eventually, after $c^{**}$, the two layers combine and operate as a whole one-dimensional embedding in Figure \ref{fig:Emb1f}. 

The difference between the one-dimensional embeddings in Figures \ref{fig:Emb1a} and \ref{fig:Emb1f} is that, for small budgets each layer collapses to a pair of distinct points, while for large budgets, each matched pair from different layers collapses to a point.
%
\begin{figure}[tbh!]
	\centering
	\includegraphics[clip,width=.45\columnwidth]{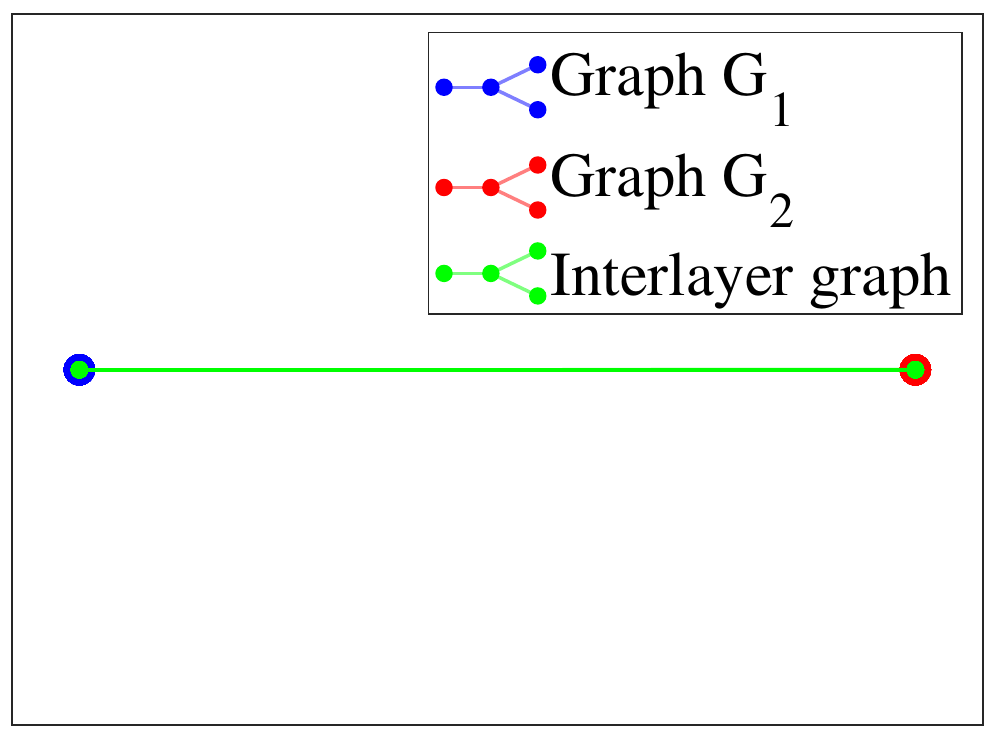}
	\caption{Optimal realization for $c<c^*$: nodes in each layer clumped together}
	\label{fig:clumped}
\end{figure}
\begin{figure}[!htb]
	\centering
	\subfloat[\label{fig:Emb1a}]{\includegraphics[clip,width=.6\columnwidth]{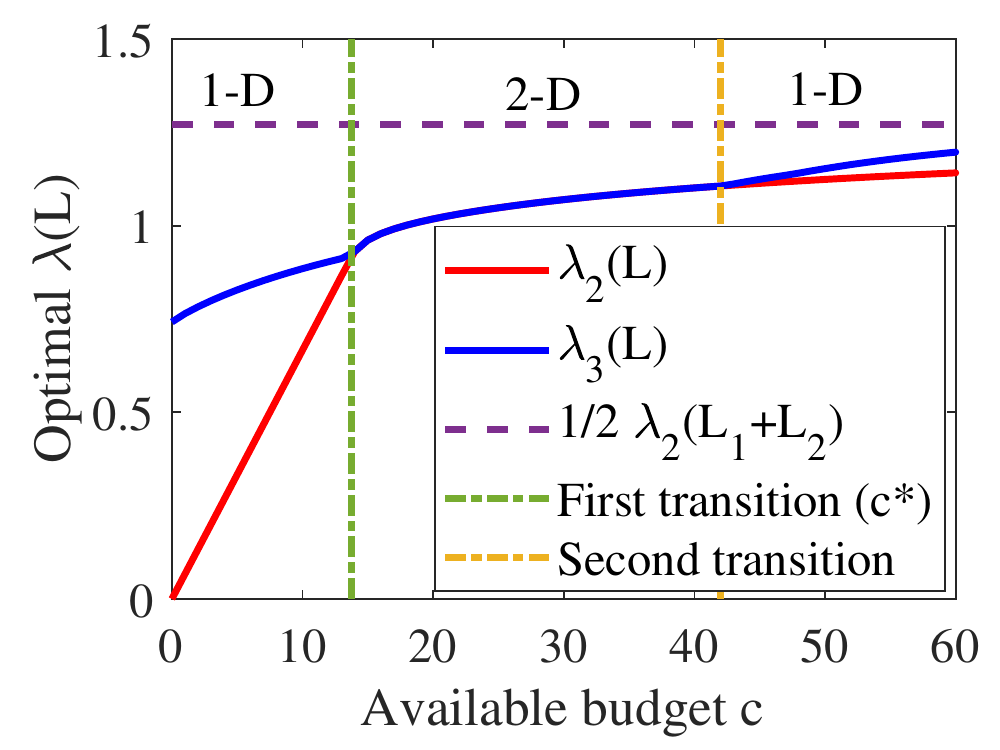}  } \\
	%
	\subfloat[\label{fig:Emb1b}]{\includegraphics[clip,width=.3\columnwidth]{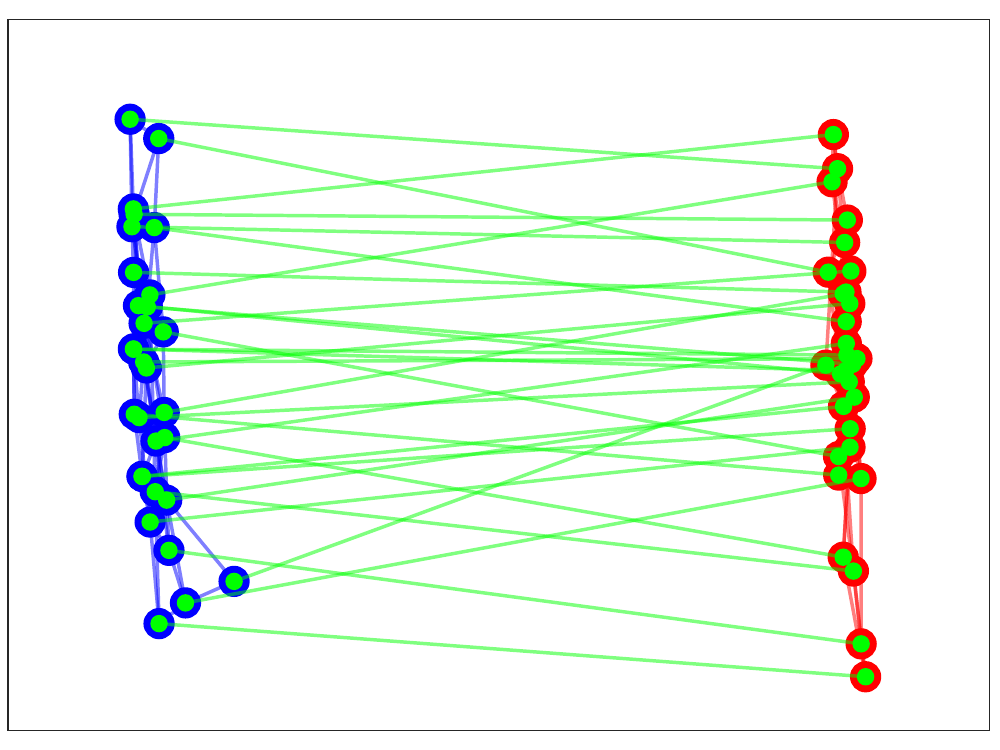}
	}
	\subfloat[\label{fig:Emb1c}]{\includegraphics[clip,width=.3\columnwidth]{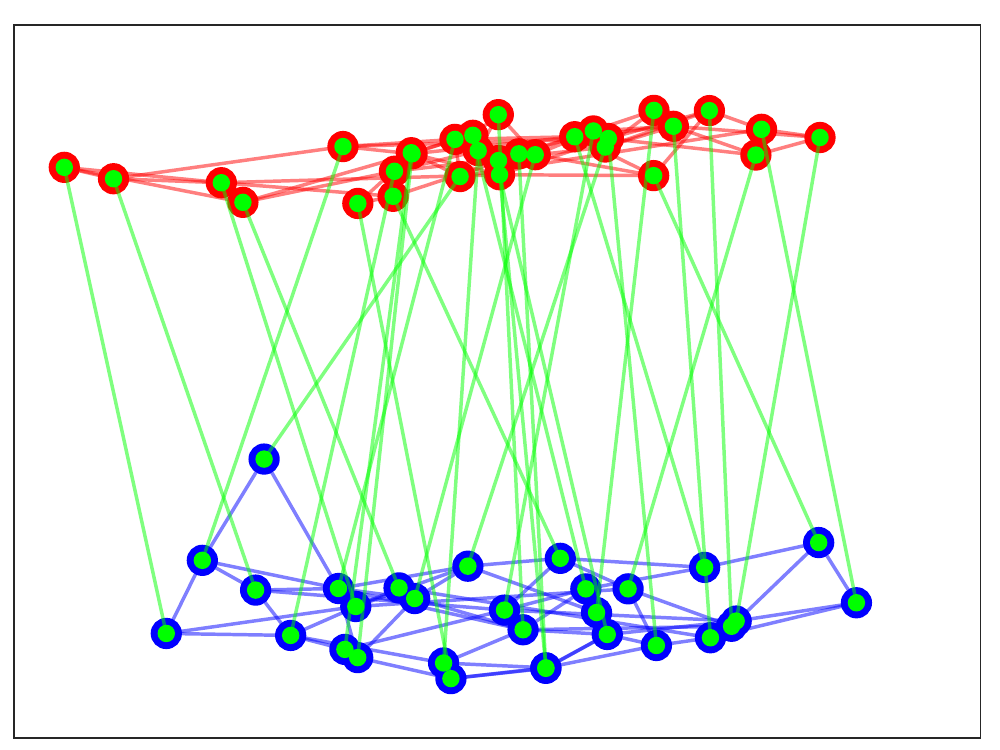}
	}
	\subfloat[\label{fig:Emb1d}]{\includegraphics[clip,width=.3\columnwidth]{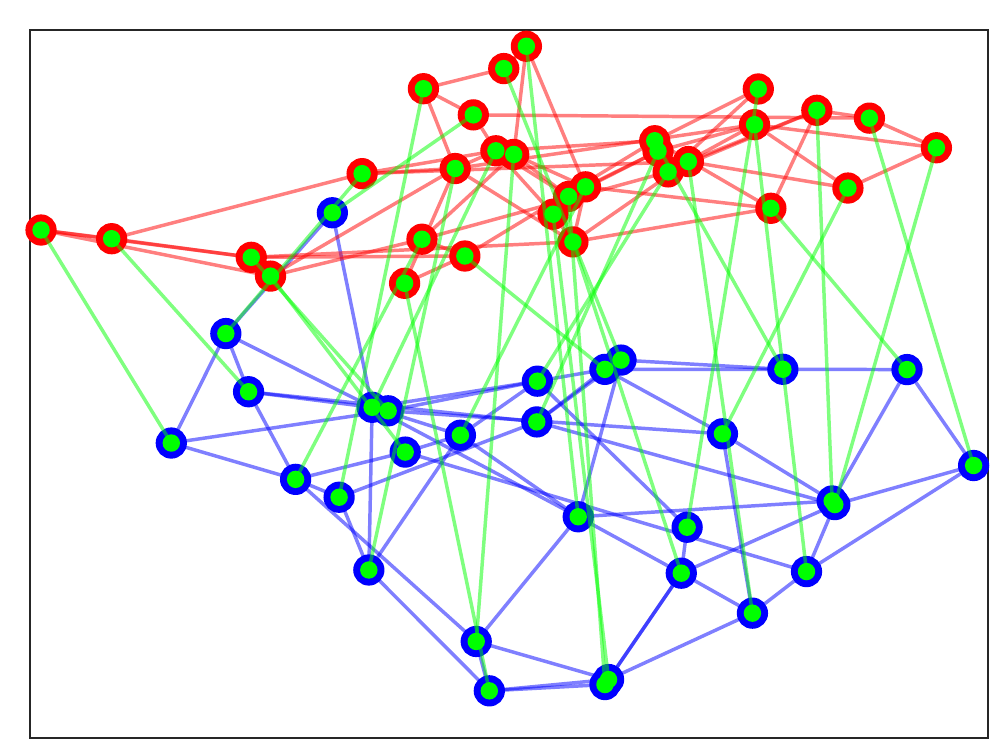}
	}\\ 
	\subfloat[\label{fig:Emb1e}]{\includegraphics[clip,width=.3\columnwidth]{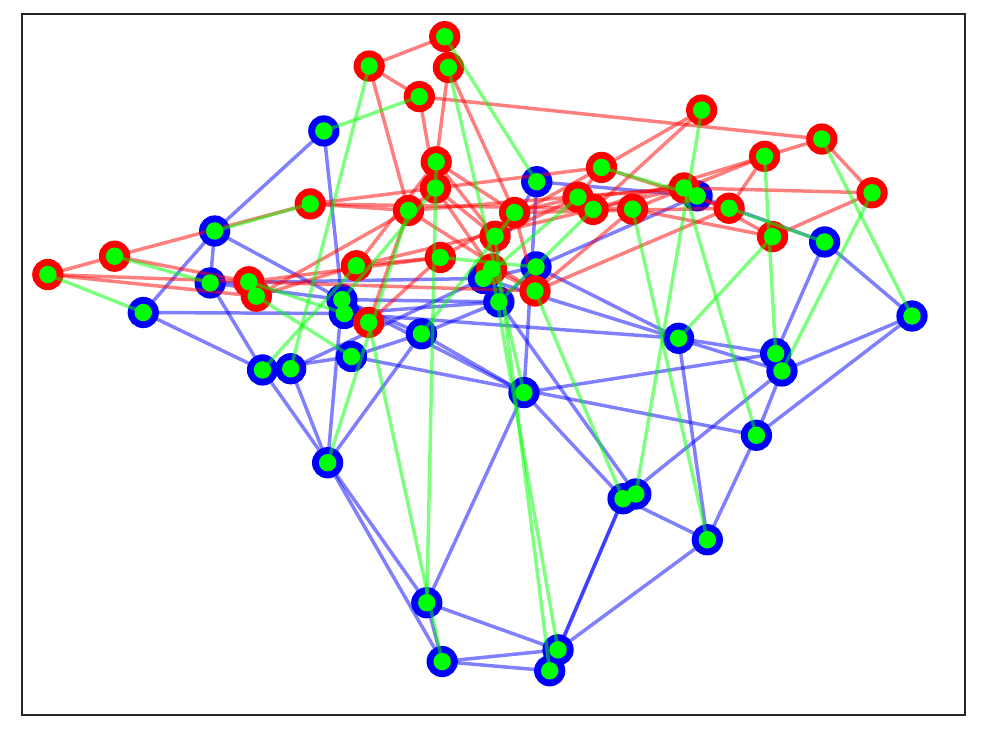}
	}
	\subfloat[\label{fig:Emb1f}]{\includegraphics[clip,width=.3\columnwidth]{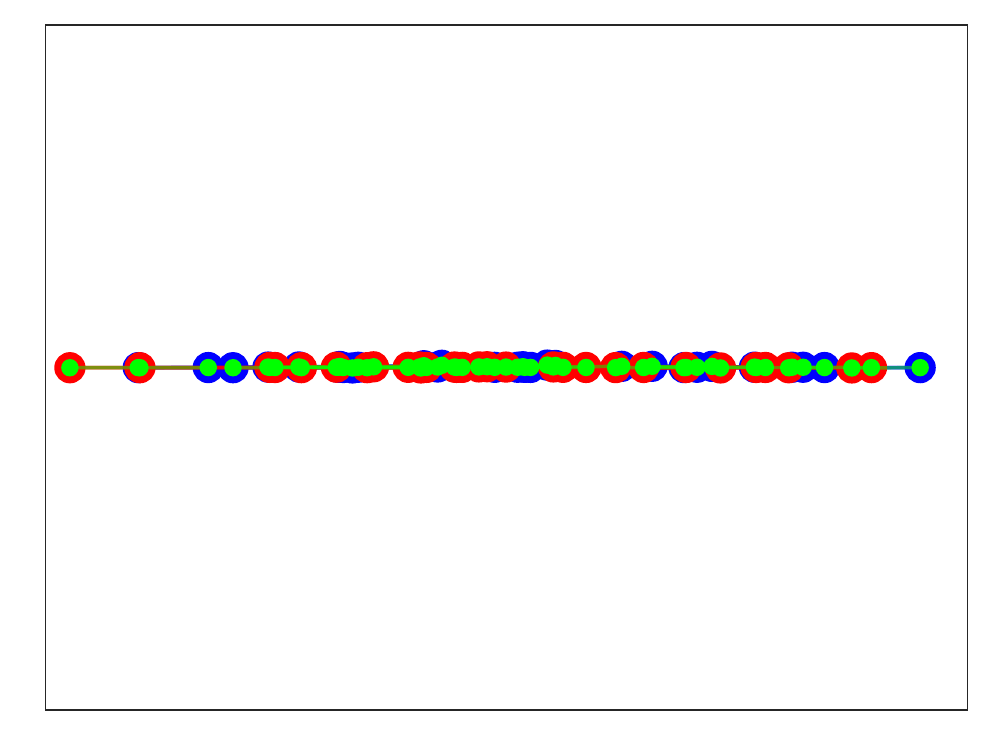}
	} \\
    \subfloat[]{\includegraphics[clip,width=.7\columnwidth]{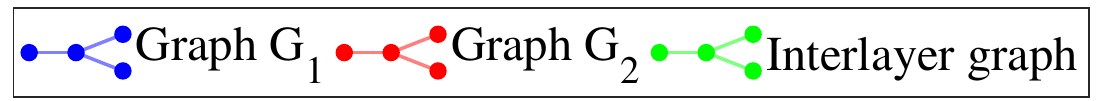}
    } 
	\caption{Embeddings for different budgets of two 30-node Watts-Strogatz networks with algebraic connectivities $\lambda_2\left(L_1\right)=0.7462$ and $\lambda_2\left(L_2\right)=0.7419$ ($c^*=13.7763$).  (a) algebraic connectivity versus budget $c$ with three regimes: $c<c^*$, $c^*<c<c^{**}$ and $c>c^{**}$ ($c^{**}$ is the next threshold budget that multiplicty changes) corresponding to 1-D, 2-D, and 3-D embedding respectively.  Embedding  are plotted  for (b) $c=14$, (c) $c=15$, (d) $c=20$, (e) $c=30$, (f) $c=50$.}
	\label{fig:Emb1}
\end{figure}
In Figure \ref{fig:Emb3}, we consider two layers with very different algebraic connectivity. We observe that the layer with weaker connectivity unfolds sooner. Notice also, in this case the embedding becomes three dimensional.
\begin{figure}[!htb]
	
	\subfloat[\label{fig:Emb3a}]{\includegraphics[clip,width=.7\columnwidth]{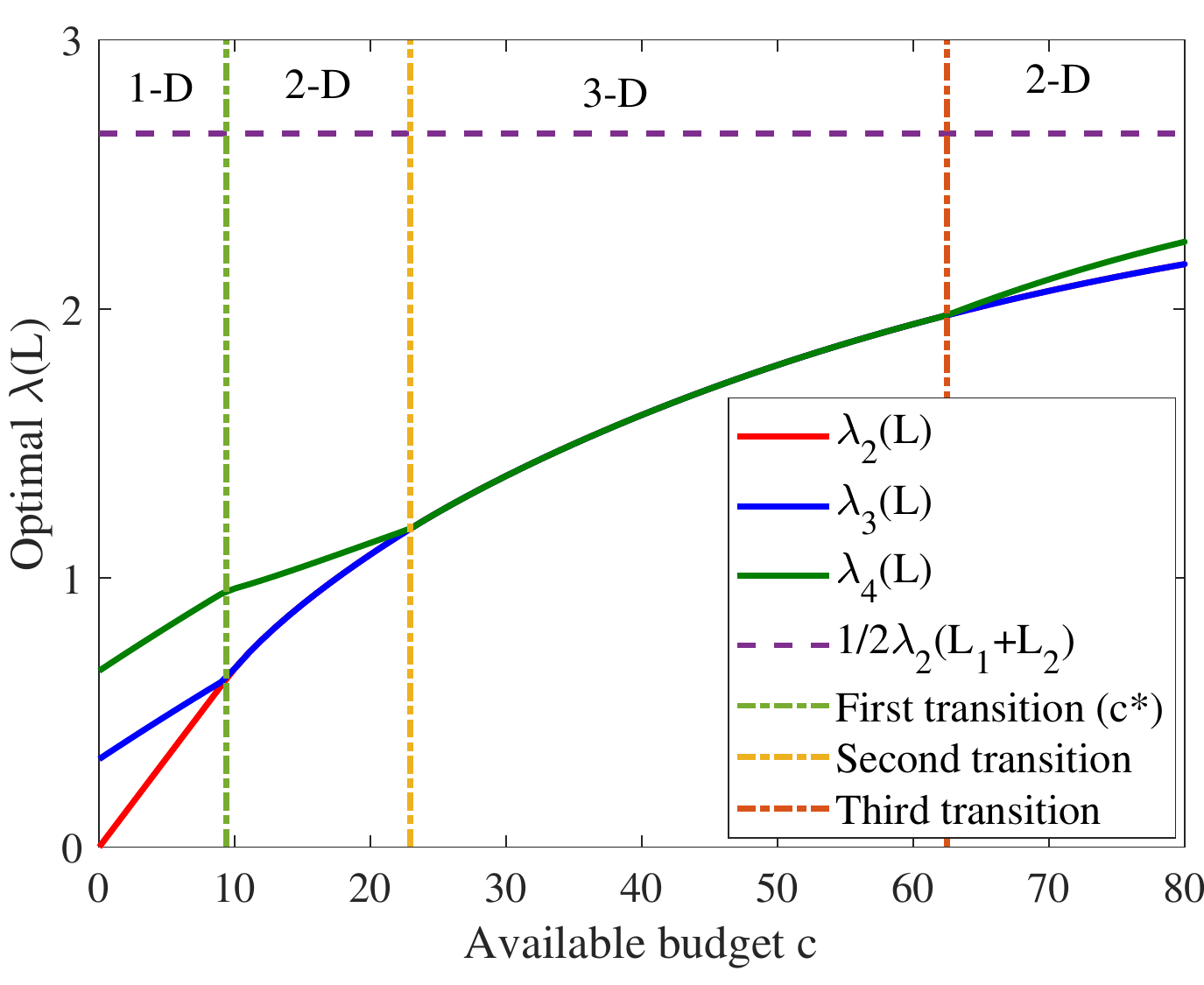}  
	}\\
	\subfloat[\label{fig:Emb3b}]{\includegraphics[clip,width=.3\columnwidth]{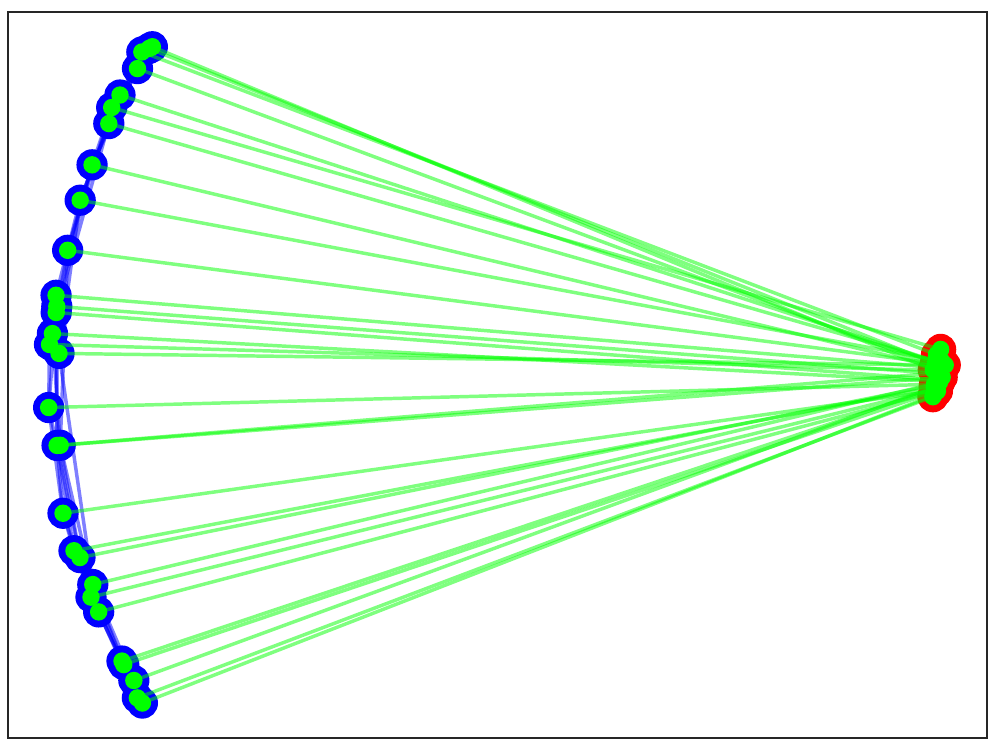}
	}
	\subfloat[\label{fig:Emb3c}]{\includegraphics[clip,width=.45\columnwidth]{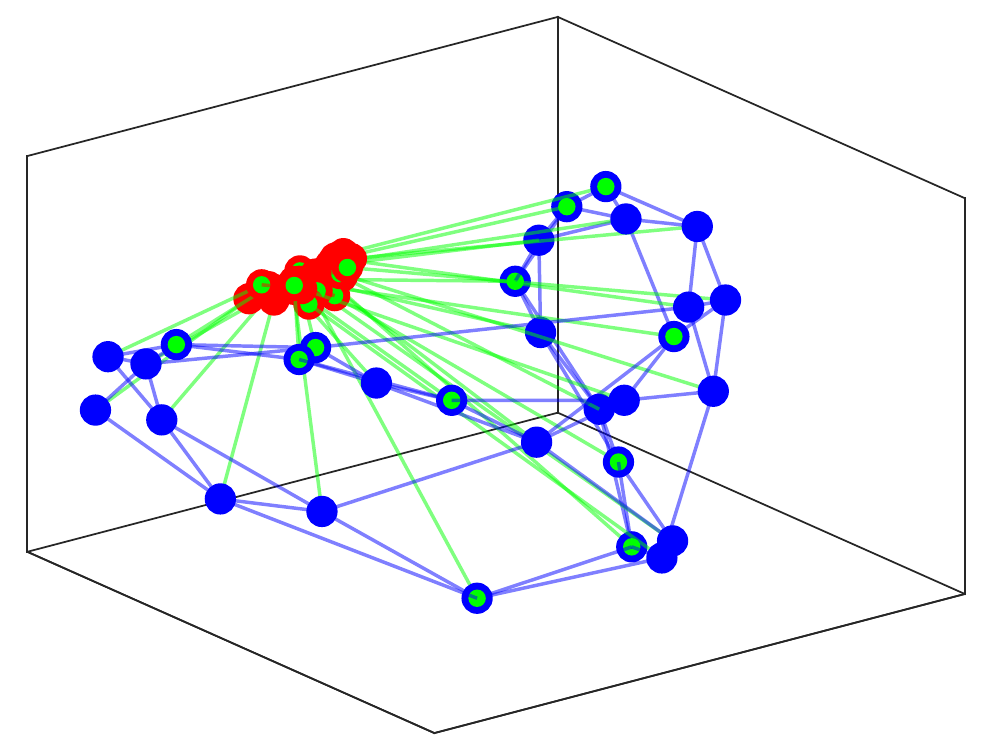}
	}
	\subfloat[\label{fig:Emb3d}]{\includegraphics[clip,width=.3\columnwidth]{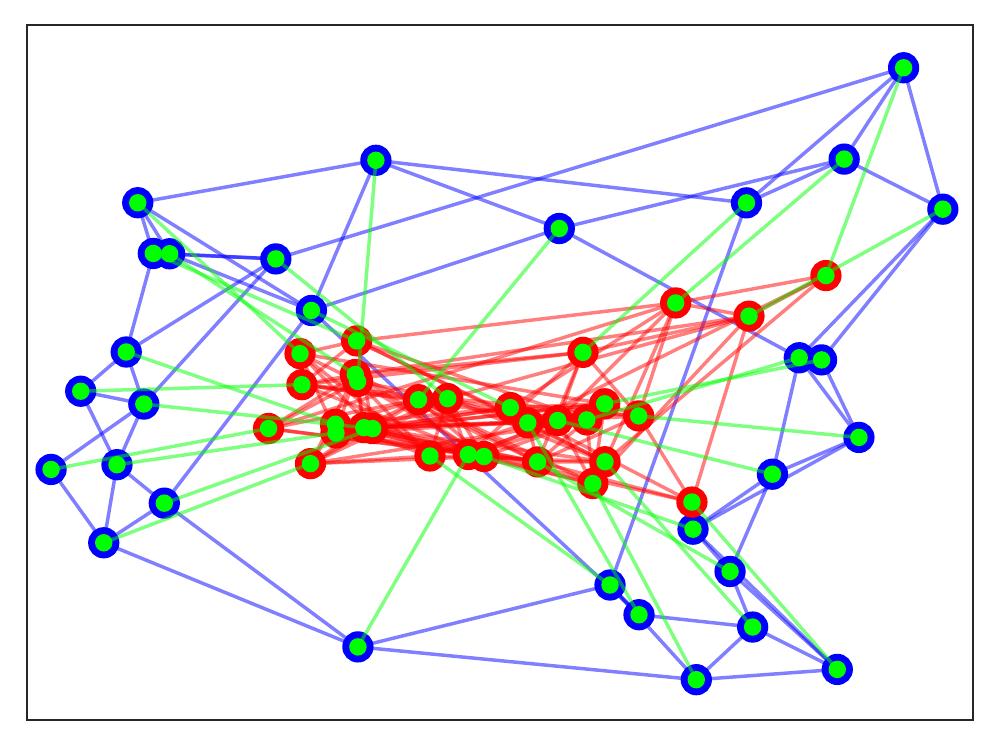}
	}\\
    \subfloat[]{\includegraphics[clip,width=.7\columnwidth]{Legend2.pdf}
    } 
	
	\caption{Embeddings for different budgets of two 30-node Watts-Strogatz networks with  algebraic connectivities $\lambda_2\left(L_1\right)=0.3270$ and $\lambda_2\left(L_2\right)=2.8733$ ($c^*=9.378$). (a) algebraic connectivity versus budget $c$, and embeddings for (b) $c=10$, (c) $c=25$, (d) $c=70$.}
	\label{fig:Emb3}
\end{figure}

The orientation of the embeddings discussed above bring forward an optimal criterion on the structure of the embeddings known as \textit{separator-shadow} established by  \citet{GoringShadowSeperator}. The result is that, the structural properties of optimal embeddings are governed by the separators of the graph. Let separator $S$ split the graph into at least two separate components between which no edges exist (the union of $S$ and the components constitute the whole graph). Then, all but one component
have the property that, the straight line segment between the origin and each node of the component intersects the convex hull of $S$--reminiscent of the forces that are aligned toward the origin. In other words, if we look at origin as a light source, nodes of all but one separated component is embedded in the shadow of the convex hull of the
separator. We investigate this in multiplex networks. First, let $G$ be connected, and $w_{ij}\geq0$ be a feasible solution of \eqref{eq:PrimalScaled}, i.e. it embodies the inequality constraints and let $G_w=\left(V,E_w:=\{ij\in E: w_{ij}>0\}\right)$. The embedding topology follows the seperator-shadow theorem stated in \citet{GoringShadowSeperator} (see SM \ref{Proof:Shadow-Separator}). 

An important feature of geometric dual formulation in this paper is that, by Proposition \ref{lem:ProjectionOfEmbedding}, different embedding patterns are informative of different optimal diffusion phases.  Particularly, Proposition \ref{lem:clump} shows that for $c\le c^*$, the embedding is one dimensional   and nodes in the same layer are clumped; consequently, the interlayer edges are not strong enough to unfold the embedding structures. In such conditions, diffusion within each layer is through intralinks, and interlinks do not contribute much to diffusion between nodes in the same layer. For $c>c^*$, we observe in Figure \ref{fig:Emb1} that stronger interlayer links start to unfold the layers and with increasing $c$, the distances between nodes in the same layer increase and thus interconnected nodes in different layers approach each other. This is indicative of transition from the intralayer phase of diffusion to the interlayer phase. Furthermore, in some circumstances,  for instance, when one individual network has  much larger algebraic connectivity than the other one,  increasing the budget past $c^*$ unfolds the layer with smaller algebraic connectivity. We observe an obvious example of this situation in  Figures \ref{fig:Emb3b} and \ref{fig:Emb3c}. This is due to higher robustness level of the layer with larger algebraic connectivity, and is indicative of interlayer optimal diffusion mode in the more weakly connected layer, versus intralayer optimal diffusion mode in the more connected one.

%
%
\section{Minimizing $\lambda_{n}$}\label{sec:lambdan}
In this section, we study the problem of minimizing the largest eigenvalue, $\lambda_n$, of the Laplacian in multiplex networks. For the total cost $c$ of interlayer links, the primal problem is defined as
\begin{equation}\label{Primal_n}
\begin{aligned}
& \underset{w_{ij}}{\text{minimize}}
& & \lambda_n \\
& \text{subject to}
& & \sum_{\lbrace i,j\rbrace\in E_3}w_{ij}E_{ij}+L_0 - \lambda_n I \preceq 0 \\
& & & \sum_{\lbrace i,j\rbrace\in E_3}w_{ij}\ge  c \\
& & & w_{ij}\geq 0 
\end{aligned}
\end{equation} 
To get some insight into the solution of the primal problem \eqref{Primal_n}, consider the characterization of the largest eigenvalue in terms of Rayleigh quotients, as
\begin{equation}\label{lambda_max}
\begin{aligned}
\lambda_n[L(w)]=\underset{\Vert v\Vert \neq 0}{\text{max}} \ \ \frac{v^TL(w)v}{\Vert v\Vert^2}
\end{aligned}
\end{equation} 
The optimal weight problem \eqref{Primal_n} is then  
\begin{equation}\label{F_c}
\begin{aligned}
\lambda_n^*(c) \ := \ \underset{\underset{w^T\boldsymbol{1}=c}{w\geq 0}}{\text{min}} \ \ \lambda_n[L(w)]
\end{aligned}
\end{equation}
for a given budget $c\geqslant 0$. Since $L$ is an affine function of $w$, and $\lambda_n$ is a convex function of $L$, it follows that \eqref{F_c} is a convex optimization problem. Indeed, \eqref{F_c} can be recast as the semidefinite programming \eqref{Primal_n} and may be solved using standard numerical methods. 

In \citet{Ribalta2013Spectral}, the weight distribution is assumed to be uniform. In that case, for small values of $c$, $\lambda_n$ grows approximately linearly with $c$, namely,  $\lambda_n=\lambda_{N}^0+2c/N$, where $\lambda_{N}^0:=\text{max}(\lambda_{max}(L_1),\lambda_{max}(L_2))$, and $N$ is the number of nodes in each layer. 

Here we will show that the optimal $\lambda_n^*$ in \eqref{Primal_n}, may actually be constant for small budgets. 
\begin{theorem}\label{thm_nodallines}
	Assume that $\lambda_{max}(L_1)>\lambda_{max}(L_2)$ and that $\lambda_N^1:=\lambda_{max}(L_1)$ is simple for $L_1$ with eigenvector $v_N^1$. Suppose that the nodal set $\{x\in V(G_1): v_N^1(x)=0\}$ is non-empty, and define $c_1^*>0$ to be the largest budget such that for $0\le c<c_1^*$, the optimal $\lambda_n^*(c)$ is simple. Then, 
	\begin{equation}\label{eq_constant}
	\lambda_n^*(c)\equiv \lambda_N^1\qquad \text{for $0\le c\le c_1^*$.}
	\end{equation}
	Moreover, in this case, a weight $w$ with $w^T\boldsymbol{1}=c$ is optimal for \eqref{Primal_n} if and only if
	\begin{equation}\label{opt_cond}
	Wv_N^1=0
	\end{equation}
	where $W=\diag(w)$ and $0$ is a zero vector.
\end{theorem}
\begin{proof}
	First note that, by Remark \eqref{remark:posPert}, the optimal largest eigenvalue $\lambda_n^*$ for \eqref{Primal_n} is a nondecreasing function of $c$. So 
	\begin{equation}\label{eq_lowerbound}
	\lambda_n^*(c)\ge \lambda_N^1\qquad\forall c\ge 0.
	\end{equation}
	By Weyl's Theorem (see \citet[Chapter 3]{roger1994topics}), since $L(w)$ is a small perturbation of $L(0)$, there is $c_1>0$, so that for $0\le c<c_1$, 
	\begin{enumerate}
		\item the corresponding $\lambda_n^*(c)$ is also simple; 
		\item and the second largest eigenvalue $\lambda_{n-1}^*(c)$ satisfies
		\begin{equation}\label{eq:upperbd}
		\lambda_{n-1}^*(c)<\lambda_N^1.
		\end{equation} 
	\end{enumerate}
	Note that $c_1\le c_1^*$, because condition 2. is in principle adding an extra restriction.

	Fix a budget $0\le c<c_1$, and let $w$ be a feasible weight with $w^T\boldsymbol{1}=c$. 
Assume first that $w$ satisfies \eqref{opt_cond}. Then
\begin{equation} \label{eq:LaplacMatrix2}
\begin{aligned}
\begin{bmatrix}
L_1+W& -W \\
-W & L_2+W
\end{bmatrix}
\begin{bmatrix}
v_N^1 \\
0
\end{bmatrix}
=\lambda_N^1\begin{bmatrix}
v_N^1 \\
0
\end{bmatrix}.
\end{aligned}
\end{equation}
Therefore, $\lambda_N^1$ is an eigenvalue of $L(w)$ with corresponding eigenvector $v:=\left( {v_N^1}^T,0\right) ^T$,
and by \eqref{eq:upperbd}, it must be the largest eigenvalue of $L(w)$. In particular, by \eqref{eq_lowerbound}, $w$ is optimal and $\lambda_n^*(c)=\lambda_N^1$. 

This shows that \eqref{opt_cond} is sufficient for optimality of $w$. In addition, since the nodal set is non-empty, there are many weights that satisfy \eqref{opt_cond}, indeed any weight supported on the nodal set will be optimal. So we have also established \eqref{eq_constant} for $0\le c\le c_1$, by continuity up to $c_1$. To conclude, note that for $c=c_1$, either $\lambda_n^*(c_1)$ is not simple, or $\lambda_{n-1}^*(c_1)=\lambda_N^1$, in which case $\lambda_{n-1}^*(c_1)=\lambda_n^*(c_1)$, and thus $\lambda_n^*(c_1)$ is again not simple. This shows that $c_1=c_1^*$.

The only direction left, is to show that \eqref{opt_cond} is also necessary for optimality in this case.

So, assume that $w$ is optimal. Using the layer structure, the vector $v$ in \eqref{lambda_max} can be written as $v=(v_1^T,v_2^T)^T$, where $v_i$ is the restriction of $v$ to layer $i$. Then, \eqref{lambda_max} for $\lambda_n[L(w)]=\lambda_N^1$ implies that
\begin{equation}\label{lam_max_comp}
\begin{aligned}
v_1^TL_1v_1+v_2^TL_2v_2+(v_1-v_2)^TW(v_1-v_2) \\
-\lambda_N^1(\Vert v_1\Vert^2+\Vert v_2\Vert^2)\leqslant0 
\end{aligned}
\end{equation}

 If we decompose $v_1$  into two orthogonal components as $v_1=\alpha v_N^1+u_1$,  for a scalar $\alpha$, where $u_1^Tv_N^1=0$, then $\|v_1\|^2=\alpha^2+\|u_1\|^2$  and 
 
 
 \begin{equation}\label{lam_max_decom}
 \begin{split}
 u_1^TL_1u_1&+v_2^TL_2v_2+\left( u_1-v_2\right) ^TW\left( u_1-v_2\right) 
 +\\
 &\alpha^2 (v_N^1)^TWv_N^1 +2\alpha \left( u_1-v_2\right)^TWv_N^1   
 \\
 &-\lambda_N^1\left( \Vert u_1\Vert^2+\Vert v_2\Vert^2\right)\leqslant0, 
\\
&\forall \alpha\in \mathbb{R}, u_1,v_2\in\mathbb{R}^N, \ \ u_1^Tv_N^1=0 
 \end{split}
 \end{equation}
 Since \eqref{lam_max_decom} holds for every $\alpha$, we must have $(v_N^1)^TWv_N^1=0$.  This concludes the `only if' direction.
\end{proof}

\begin{proposition}
At the threshold $c_1^*$, the matrix $Q+2W^\star$ has a zero eigenvalue, where $Q=\bar{L}-\tilde{L}\bar{L}^\dagger\tilde{L}$ and $\bar{L}=\frac{L_1+L_2}{2}-\lambda_N^1I$, $\tilde{L}=\frac{L_1-L_2}{2}$.
\end{proposition}
\begin{proof}
	Using the approach in \citet{sahneh2014exact}, we write the eigenvalue problem $Lv=\lambda v$ as
	\begin{equation}\label{ExtEigProb}
	\begin{aligned}
	\begin{bmatrix}
	L_1+W^* && -W^* \\
	-W^* && L_2+W^*
	\end{bmatrix}
	\begin{bmatrix}
	v_1 \\
	v_2
	\end{bmatrix}
	=\lambda 
	\begin{bmatrix}
	v_1 \\
	v_2
	\end{bmatrix}
	\end{aligned}
	\end{equation} 
	where $v=[v_1^T, v_2^T]^T\in\mathbb{R}^{2N}$ satisfies the following eigenvector normalization:
	\begin{equation}\label{NormalCond}
	v_1^Tv_1+v_2^Tv_2=2N
	\end{equation} 
	Differentiating \eqref{ExtEigProb} and \eqref{NormalCond} with respect to $c$ yields the governing equations for the eigenderivatives $\frac{dv_1}{dc}$, $\frac{dv_2}{dc}$, and $\frac{d\lambda}{dc}$ 
	\begin{equation}
	{\renewcommand{\arraystretch}{1.8}
	\begin{aligned}
	\begin{bmatrix}
	L_1+W^*-\lambda I && -W^* && -v_1 \\
	-W^* && L_2+W^*-\lambda I && -v_2 \\
	-v_1^T && -v_2^T && 0
	\end{bmatrix}
	\begin{bmatrix}
	\frac{dv_1}{dc} \\ \frac{dv_2}{dc} \\ \frac{d\lambda}{dc}
	\end{bmatrix} 
	=\\
	 \begin{bmatrix}
	-\frac{dW^*}{dc}\left(v_1-v_2\right) \\ \frac{dW^*}{dc}\left(v_1-v_2\right) \\ 0
	\end{bmatrix}
	\end{aligned}
    }
	\end{equation}
	For $0\le c<c_1^*$, $\lambda=\lambda_{N}^1$, $v_1=v_N^1$, $v_2=0$, therefore
	\begin{equation}
	{\renewcommand{\arraystretch}{1.8}
	\begin{aligned}
	\begin{bmatrix}
	L_1+W^*-\lambda_{N}^1 I && -W^* && -v_N^1 \\
	-W^* && L_2+W^*-\lambda_{N}^1 I && 0 \\
	-{v_N^1}^T && 0 && 0
	\end{bmatrix}
	\begin{bmatrix}
	\frac{dv_1}{dc} \\ \frac{dv_2}{dc} \\ \frac{d\lambda}{dc}
	\end{bmatrix} 
	=\\
	 \begin{bmatrix}
	-\frac{dW^*}{dc}v_N^1 \\ \frac{dW^*}{dc}v_N^1 \\ 0
	\end{bmatrix}
	\end{aligned}
}
	\end{equation}
	
	Consider 
	\[
	M=\begin{bmatrix}
	L_1+W^*-\lambda_{N}^1 I && -W^* && -v_N^1 \\
	-W^* && L_2+W^*-\lambda_{N}^1 I && 0 \\
	-{v_N^1}^T && 0 && 0
	\end{bmatrix}
	\]
	The matrix $L$ can have repeated eigenvalue only if $M$ is singular. Hence, at the threshold $c_1$ there exists $x\neq 0$ such that  $M(W^*)x=0$, thus
	\begin{equation}\label{eigen1}
	\begin{aligned}
	\begin{bmatrix}
	L_1-\lambda_N^1 I && 0 && -v_N^1 \\
	0 && L_2-\lambda_N^1 I && 0 \\
	-{v_N^1}^T && 0 && 0
	\end{bmatrix}x 
	=\\
	\begin{bmatrix}
	-W^* && W^* && 0 \\ W^* && -W^* && 0 \\ 0 && 0 && 0
	\end{bmatrix}x.
	\end{aligned}
	\end{equation}
	Using Theorem \ref{thm_nodallines} and \eqref{opt_cond}, the inner product of the first row on both sides of \eqref{eigen1} with $v_N^1$ shows that $x$ is of the form $x=\left[x_1^T \ \ x_2^T \ \ 0\right]^T$. Then, a linear transformation $y_1=x_1-x_2$, $y_2=x_1+x_2$, in \eqref{eigen1} implies that 
	\begin{equation}\label{eigen2}
	\begin{aligned}
	\bar{L}y_1+\tilde{L}y_2=-2W^* y_1 \\
	\tilde{L}y_1+\bar{L}y_2=0 \ \ \ \ \ \ \ \ \
	\end{aligned}
	\end{equation}
Eliminating $y_2$ in \eqref{eigen2} yields 
	\begin{equation}\label{eigen3}
	\begin{aligned}
	(\bar{L}-\tilde{L}\bar{L}^\dagger\tilde{L})y_1=-2W^* y_1
	\end{aligned}
	\end{equation}
	where $\bar{L}^\dagger$ is the Moore Penrose pseudo-inverse of $\bar{L}$. 
\end{proof}

\begin{figure}[!htb]
	\centering
		\includegraphics[clip,width=.9\columnwidth]{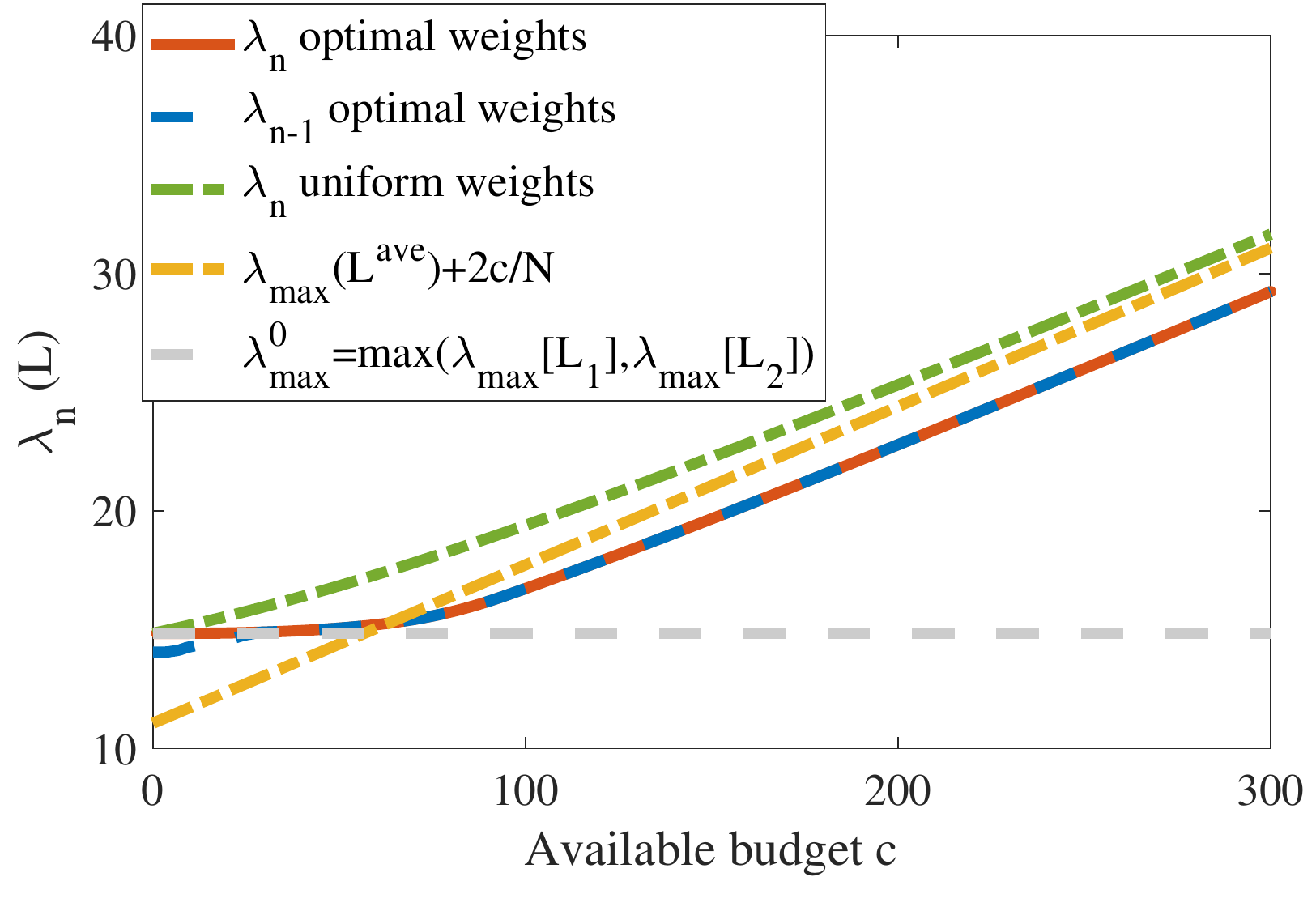} 
		
	\caption{The largest eigenvalue $\lambda_n$ in a multiplex of two random geometric networks with 30 nodes}
	\label{fig:lambda_n30}
\end{figure}
Figure \ref{fig:lambda_n30}  compares the optimal value of $\lambda_n(L)$ to the one obtained by the uniform distribution, as the budget $c$ varies for a multiplex with two different random geometric network layers with 30  nodes. The optimal distribution gives smaller $\lambda_n$ for all budgets. We observe that after coalescing with $\lambda_{n-1}$, i.e. after the threshaold $c_1^*$, the multiplicity of $\lambda_{n}^*$ never decreases back to one. This phenomenon seems to persist in all other examples we have computed. Moreover,  for $c>c_1^*$, it becomes impossible for a condition such as  \eqref{opt_cond} to hold, since there will be generally no diagonal matrix $W$ satisfying multiple conditions corresponding to different eigenvectors of $L$. Therefore, $\lambda_N^1$ will be no longer attainable, and the optimal solution $\lambda^*_n$ enters a nonlinear regime for $c>c_1^\star$. 

In Figure \ref{fig:OptWei_n30}, we illustrate the behavior of the optimal weight distribution for this example, when varying the budget. For $c<c_1^*$, the optimal weight distribution is highly nonuniform in that the total budget is assigned to one node, i.e. to node numbered 11, while the others experiencing zero weight. Inspecting the eigenvector $v_N^1$ corresponding to the largest eigenvalue $\lambda_N^1$ of the individual layers reveals that $v_N^1(11)=0.0013$, and this is the smallest absolute value among the entries of $v_N^1$. Therefore we see that even though \eqref{opt_cond} is not exactly satisfied, i.e. the nodal set of $v_N^1$ is empty,  $\lambda_n^*$ is still close to  $\lambda_N^1$.

\begin{figure}[!htb]
	\subfloat[]{\includegraphics[clip,width=.5\columnwidth]{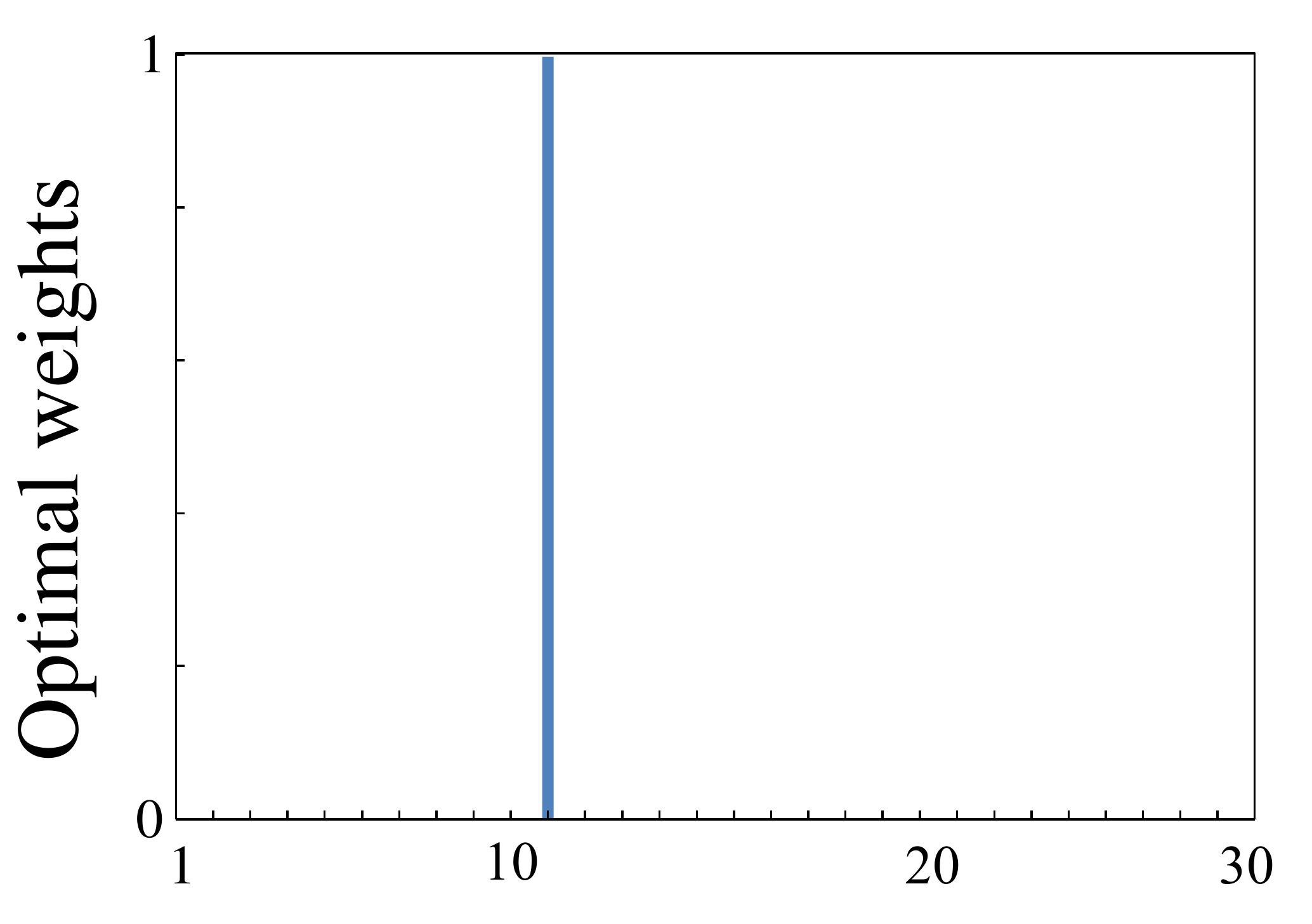}}
	\subfloat[]{\includegraphics[clip,width=.5\columnwidth]{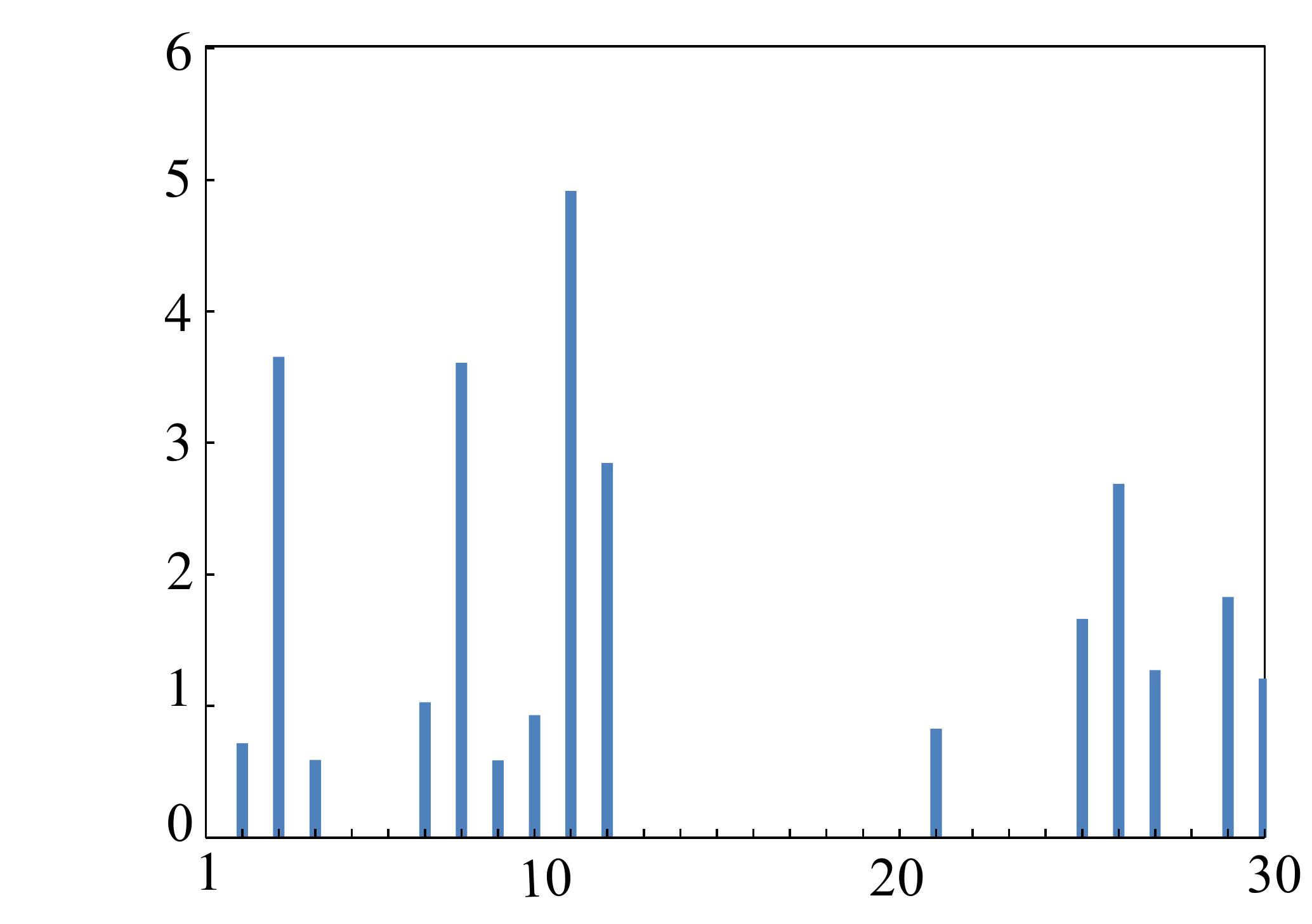}} \\
	\subfloat[]{\includegraphics[clip,width=.5\columnwidth]{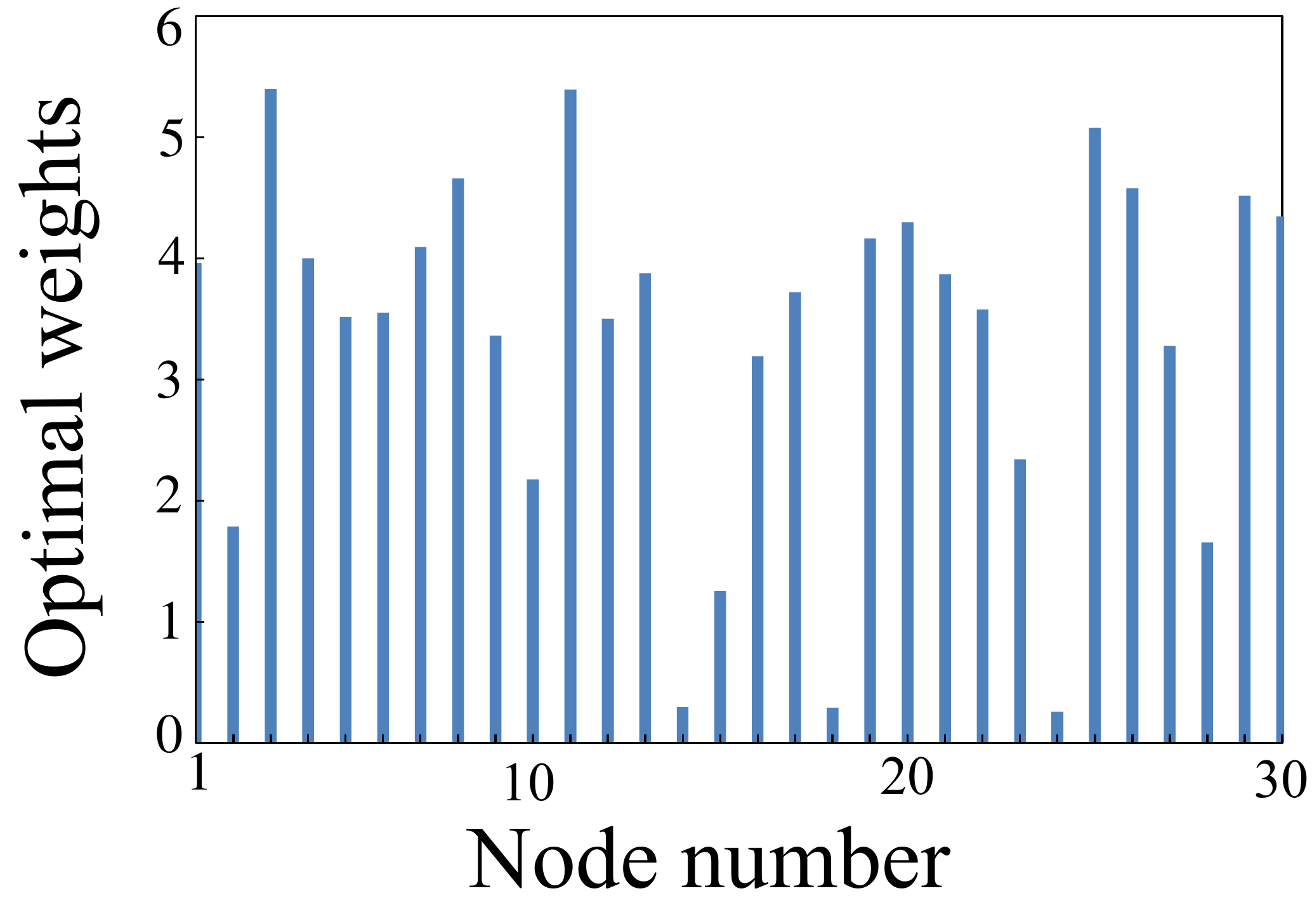}}
	\subfloat[]{\includegraphics[clip,width=.5\columnwidth]{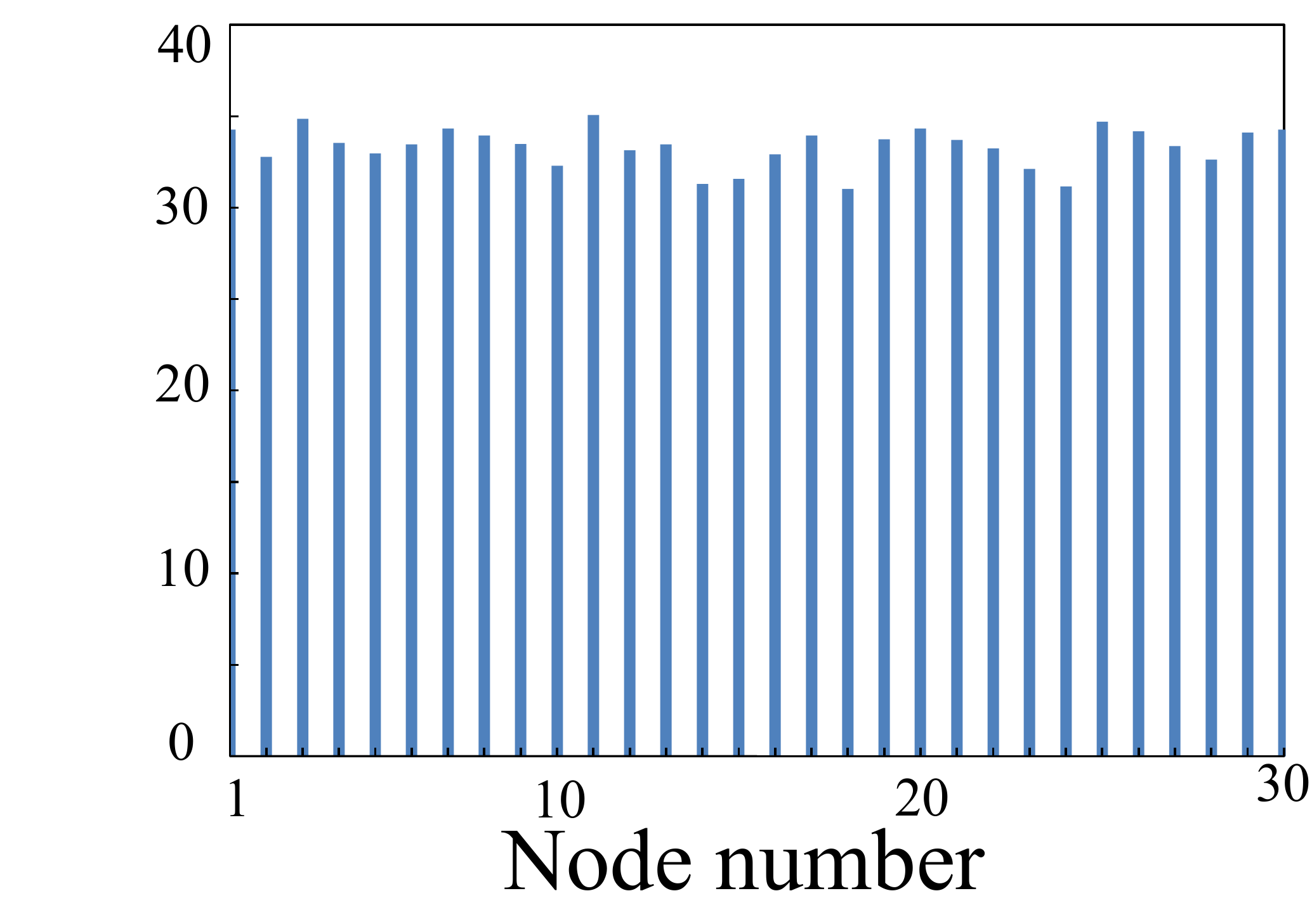}} \\
	\centering
	\subfloat[]{\includegraphics[clip,width=.5\columnwidth]{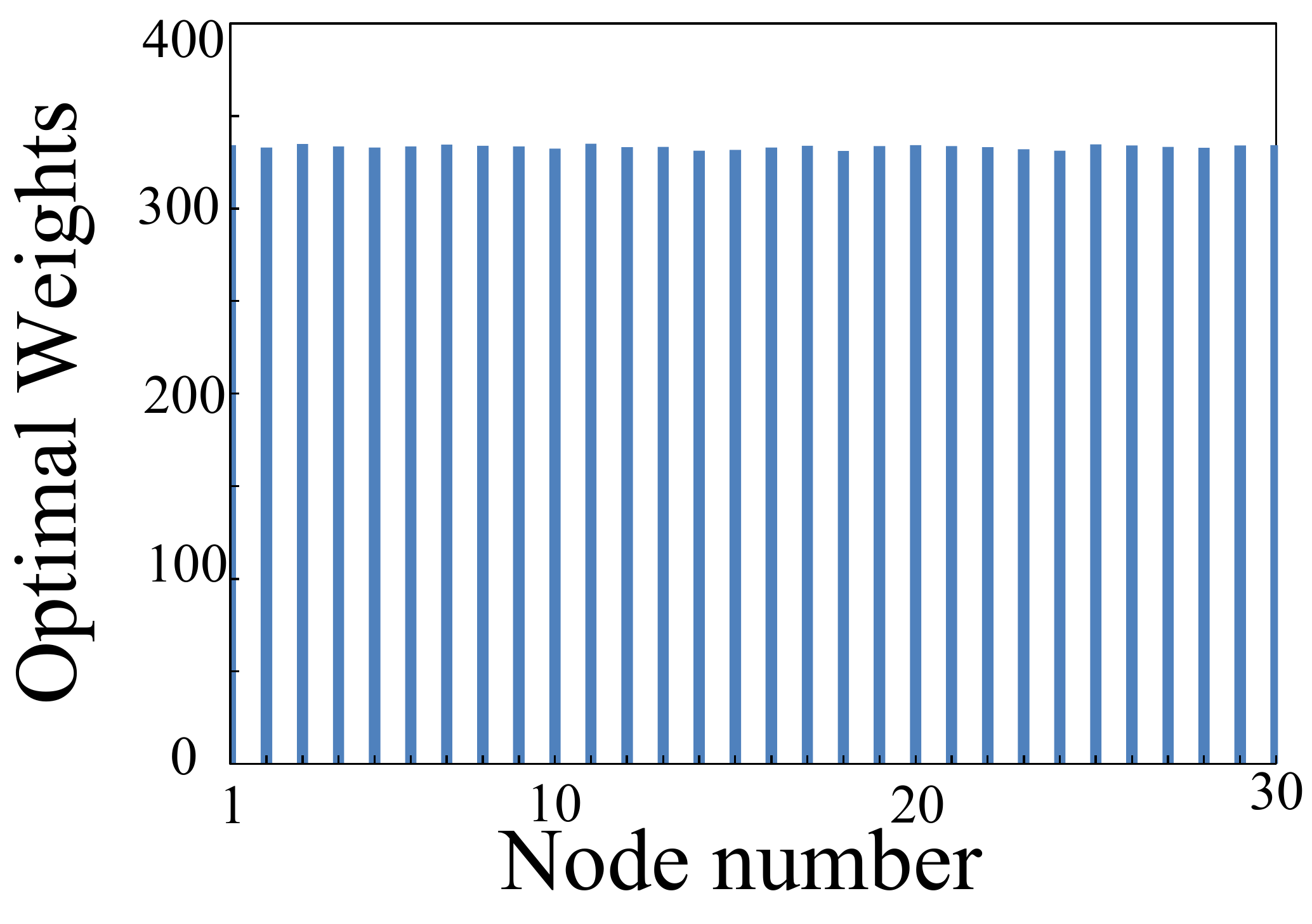}}
	\caption{Optimal weight distribution in the multiplex of two random geometric networks with 30 nodes for the budget (a) $c=1$, (b) $c=30$, (c) $c=100$, (d) $c=1000$, (e) $c=10000$.}
	\label{fig:OptWei_n30}
\end{figure}


Figure \ref{fig:OptWei_n30} 
 also indicates that the optimal weight distribution wants to become uniform with increasing $c$. This seems to be a general phenomenon and we provide an explanation next. First, the uniform weight distribution imposes an upper-bound on $\lambda_n^*$, i.e.,  $\lambda_n^*\leqslant\lambda_{max}(L_{ave})+2c/N$ for large $c$. Second, $\lambda= 2c/N$ is always a lowerbound obtained from substituting  $v_1=-v_2\equiv 1$ in \eqref{lambda_max}. Therefore, for large values of $c$, $\lambda_n^*$ is restricted as follows
\begin{equation}\label{F_c_rest}
\begin{aligned}
\frac{2c}{N}\leqslant \lambda_n^*\leqslant\lambda_{max}(L_{ave})+\frac{2c}{N}, \ \ \forall \ c \ \text{large.}
\end{aligned}
\end{equation}
Scaling the elements of the Laplacian matrix by $1/c$, scales the eigenvalues accordingly. Moreover, 
as $c$ goes to infinity, the scaled multiplex behaves more and more like a perfect matching bipartite graph, hence $\lambda_n$ tends to be twice the largest weight of the interlayer edges. Consequently, in the limit $\lambda_n^*$ approaches the uniform weight distribution and the lower bound in \eqref{F_c_rest}.

However, for large $c$ the largest eigenvalue $\lambda_n$ has multiplicity,
while the uniform weight distribution leads to single largest eigenvalue for large $c$ \citep{Ribalta2013Spectral}. Therefore, there is a gap, albeit vanishing, between the optimal and the uniform weight distribution for large budgets shown in Figure \ref{fig:OptWei_n30}(d-e) (for more results see SM \ref{app_radius}).  This becomes clearer in the next section when we examine the dual formulation.

\subsection{Dual formulation and embedding of minimizing $\lambda_{n}$}
The dual  problem of \eqref{Primal_n} is:
\begin{equation}\label{DualLam_n}
\begin{aligned}
& \underset{Y, \xi}{\text{maximize}}
& & c\xi + \langle Y, L_0\rangle \\
& \text{subject to}
& & \langle Y, E_{ij}\rangle\geq \xi\ \ \text{for}\ \lbrace i,j\rbrace\in E_3 \\
& & & \langle Y, I\rangle = 1\\
& & & Y\succeq 0,~\xi\in \mathbb{R}
\end{aligned}
\end{equation}
Propositions \ref{lem:FeasibleDualSolution} and \ref{lem:StrongDuality} also hold for \eqref{DualLam_n}, thus strong duality is satisfied by the primal and dual problems \eqref{Primal_n} and \eqref{DualLam_n}. 
By considering $Y$ as a Gram matrix $Y = V^TV$, where $V\in \mathbb{R}^{n\times n}$, then the problem \eqref{DualLam_n} is equivalent to the following embedding problem:
\begin{equation}\label{Embedding_n}
\begin{aligned}
& \underset{v_i, \xi}{\text{maximize}}
& & c\xi + \sum_{\lbrace i,j\rbrace\in E_1\cup E_2}\| v_i - v_j \|^2 \\
& \text{subject to}
& & \| v_i - v_j\|^2\geq \xi \ \ \forall \lbrace i,j\rbrace\in E_3\\
& & & \sum_{i\in N}\| v_i \|^2 = 1\\
& & & v_i\in \mathbb{R}^n,~\xi\in \mathbb{R}
\end{aligned}
\end{equation}
We interpret the column vectors $v_i^*$ of the matrix $V^*$ as coordinates for the position of each node in the corresponding embedding into $\mathbb{R}^n$.
 
The following result is similar to Proposition \ref{lem:ProjectionOfEmbedding}.
\begin{proposition}\label{lem:ProjLam_n}
	Projecting an optimal embedding $v_1^*,...,v_n^*$, solving \eqref{Embedding_n}, onto a one-dimensional subspace, yields an eigenvector for the largest eigenvalue $\lambda_n^*$.
\end{proposition}

%

%

	In the embedding problem \eqref{Embedding_n}, although not an explicit constraint, the barycenter is forced to stay at the origin. This follows from Proposition \ref{lem:ProjLam_n} and the fact that the eigenvector coresponding to $\lambda_n(L)$ is perpendicular to the constant ones vector $\boldsymbol{1}$. 
	
	Moreover, we have the following observation.
\begin{remark} \label{rem:EmbDim_n}
	Proposition \ref{lem:ProjLam_n} implies that the multiplicity of $\lambda_n(L)$ is an upper-bound on the dimension of the embedding problem \eqref{Embedding_n}.  
\end{remark}

Using Proposition \ref{lem:ProjLam_n}  and  Theorem \ref{thm_nodallines}, for $c\leqslant c_1^*$ the eigenvector corresponding to $\lambda_n^*$ is $v_n=\left( {v_N^1}^T,0\right) ^T$, and therefore the following result holds.

\begin{proposition}\label{Emb_n_Small}
	For budget values up to the threshold $c^\star$, as defined in Theorem \ref{thm_nodallines}, the optimal solution of the embedding problem \eqref{Embedding_n} is one-dimensional and, up to a rotation, is given by
	\begin{equation}\label{eq:SmalEmb}
	\begin{aligned}
	v_i^*=\begin{cases}
	\gamma v_N^1(i)e_1, & \text{if}\ i \in V_1 \\ 
	0, & \text{if}\ i \in V_2
	\end{cases},\ \ \ \ \forall c<c^\star
	\end{aligned}
	\end{equation} 
	where $\gamma$ is a constant and $e_1\in\mathbb{R}^n$ is the first standard basis vector.
\end{proposition}
An illustration of \eqref{eq:SmalEmb} is given in Figure \ref{fig:WS_ln_Embb}.

Similarly, for very large values of $c$, we have the following statement.

\begin{proposition}\label{Emb_n_Large}
	For very large values of $c$, if $v_1^*,...,v_n^*$ is the optimal solution of the embedding problem \eqref{Embedding_n}, then  $v_i^*\approx -v_{N+i}^*$, for $i=1,...,N$.
\end{proposition}

Figure \ref{fig:WS_ln_Emb}
  shows the optimal embedding for a multiplex with two Watts-Strogatz layers. According to Remark \ref{rem:EmbDim_n},  the embedding dimension remains bounded by the multiplicity of $\lambda_n^*$ in Figure \ref{fig:WS_ln_Emba}. The embedding for $c=1<c^\star$, in \ref{fig:WS_ln_Embb} shows a 1-dimensional embedding where the nodes of the layer with the largest $\lambda_N$, in this case $G_1$, are distributed along a line centered on the origin, while the layer $G_2$ is concentrated at the origin.  Consequently,  the optimal embeddings for small budgets $c<c^\star$ are based on \eqref{eq:SmalEmb} and are driven by $\lambda_N^1$ and its corresponding eigenvector $v_N^1$.
  
For budgets slightly or moderately above $c_1^*$, 
the embedding of the individual layers starts to expand in the plane, see Figures \ref{fig:WS_ln_Embc} and \ref{fig:WS_ln_Embd}. However, the optimal embedding is still influenced by $\lambda_N^1$ and $v_N^1$ so the nodes of $G_2$ with smaller largest eigenvalue unfold slowly around the origin, while the nodes of $G_1$ are still related to the components of $v_N^1$. 

For larger budgets $c$, the optimal embedding expands into 3-D space. Figure \ref{fig:WS_ln_Embe} shows the expansion of the nodes of $G_2$  in 3-D space. 
In Figure \ref{fig:WS_ln_Embf}, we display a 2-D projection of the embedding and see that the
nodes of the two layers are mixed and become more homogeneous, in the sense that each interlayer link in the embedding crosses the origin at the midpoint, thus confirming the claim in Proposition \ref{Emb_n_Large}.

\begin{figure}[!htb]
	\subfloat[\label{fig:WS_ln_Emba}]{\includegraphics[clip,width=.6\columnwidth]{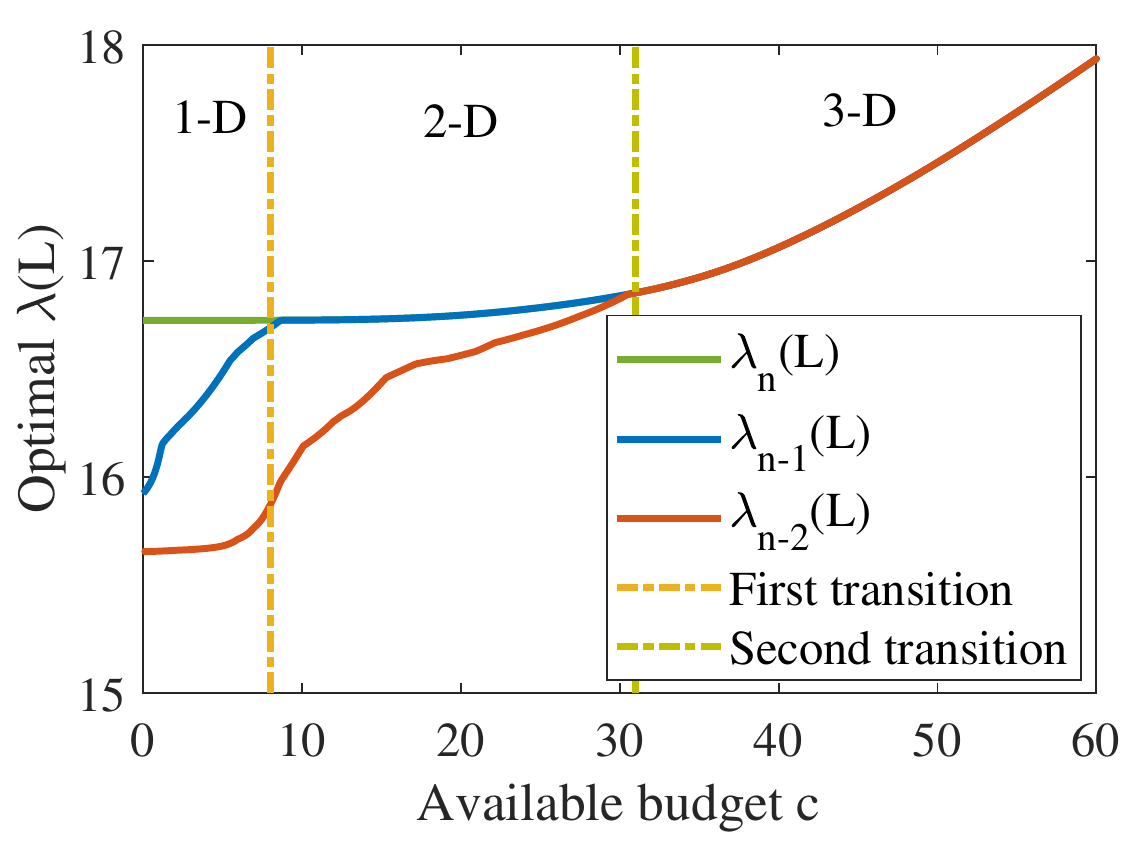}}
	
	\subfloat[\label{fig:WS_ln_Embb}]{\includegraphics[clip,width=.3\columnwidth]{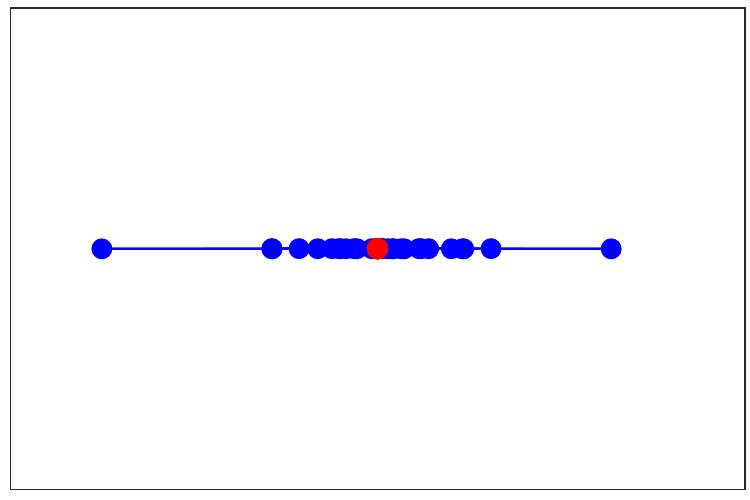}} \ \ \ 	
	\subfloat[\label{fig:WS_ln_Embc}]{\includegraphics[clip,width=.3\columnwidth]{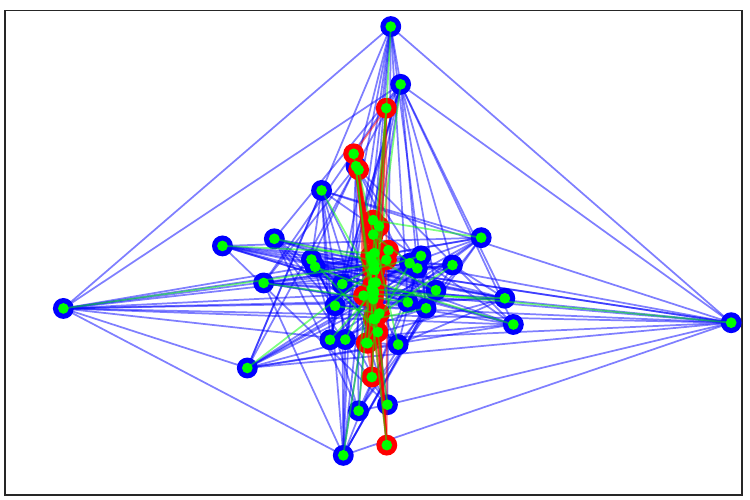}} \ \ \
	\subfloat[\label{fig:WS_ln_Embd}]{\includegraphics[clip,width=.3\columnwidth]{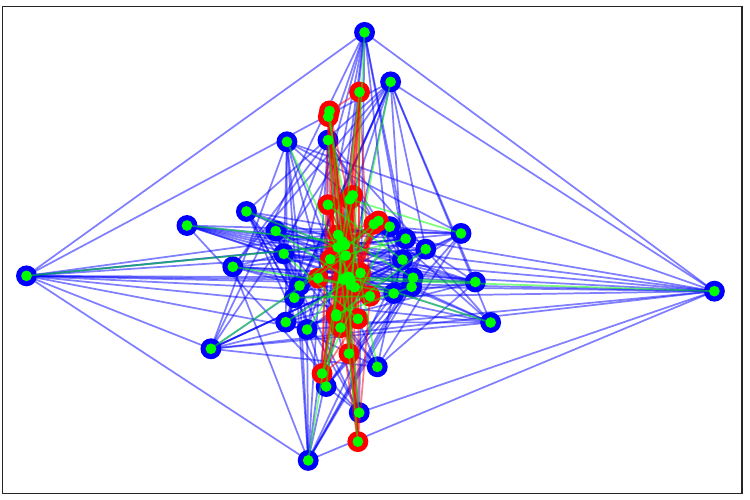}}
	
	\subfloat[\label{fig:WS_ln_Embe}]{\includegraphics[clip,width=.4\columnwidth]{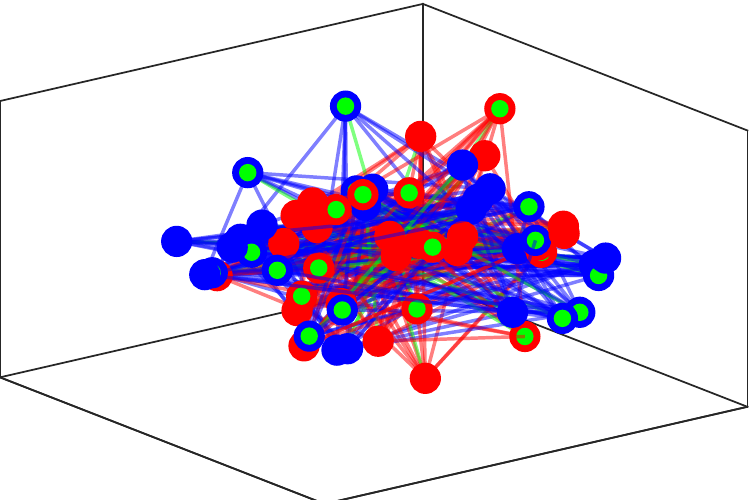}} \ \ \
	\subfloat[\label{fig:WS_ln_Embf}]{\includegraphics[clip,width=.3\columnwidth]{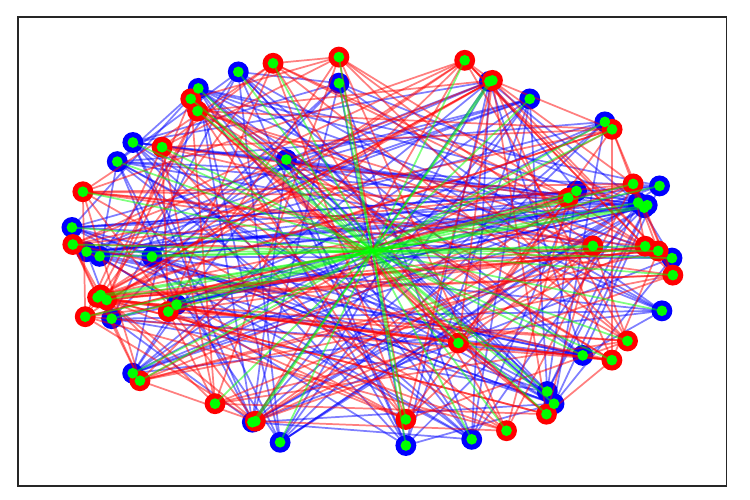}}
	
	\subfloat[]{\includegraphics[clip,width=.6\columnwidth]{Legend2.pdf}}
	
	\caption{Optimal embeddings corresponding to minimum $\lambda_n$ in a multiplex of two different Watts-Strogatz networks with 30 nodes. The largest eigenvalues of individual networks are $16.7251=\lambda_{\text{max}}(L_1)>\lambda_{\text{max}}(L_2)=14.4187$. (a) The three largest eigenvalues of supra-Laplacian versus available budget $c$. The first transition occurs at $c^\star=8.9$ where $\lambda_{n-1}(L)$ coalesces $\lambda_n(L)$, and the second transition occurs at $c^{\star\star}=30.5$ where $\lambda_{n-2}(L)$ coalesces $\lambda_{n-1}(L)$ and $\lambda_n(L)$. (b) 1-D embedding for $c=1$. (c) 2-D embedding for $c=20$. (d) 2-D embedding for $c=30$. (e) 3-D embedding for $c=100$. (f) Projection of embedding on 2-D for extremely large values of $c$ (weights are close to uniform, however, multiplicity is $>1$).}
	\label{fig:WS_ln_Emb} 
\end{figure}

\section{Minimizing spectral width $\lambda_n-\lambda_2$}\label{app:gap}
The problem of distributing the total budget $c$ on the inter-layer links to minimize the spectral width of the Laplacian matrix is defined as 
\begin{equation}\label{PrimalGap}
\begin{aligned}
& \underset{w_{ij}, \lambda_n, \lambda_2, \mu}{\text{minimize}}
& & \lambda_n-\lambda_2 \\
& \text{subject to}
& & \sum_{\lbrace i,j\rbrace\in E_3}w_{ij}L_{ij}+L_0 - \lambda_n I \preceq 0 \\
& & & \sum_{\lbrace i,j\rbrace\in E_3}w_{ij}L_{ij}+L_0+\mu ee^T - \lambda_2 I \succeq 0 \\
& & & \sum_{\lbrace i,j\rbrace\in E_3}w_{ij}= c \\
& & & w_{ij}\geq 0 
\end{aligned}
\end{equation}
where $L_0 = \sum_{\lbrace i,j\rbrace\in E_1\cup E_2}L_{ij}$. In terms of Rayleigh quotients,  \eqref{PrimalGap} becomes
\begin{equation}\label{GapCharac}
\begin{split}
H(c) & := \underset{\underset{w^T\boldsymbol{e}=c}{w\geq 0}}{\text{min}} \ \ \lambda_n[L(w)]-\lambda_2[L(w)] \\
\lambda_2[L(w)]&=\underset{\underset{\Vert u\Vert \neq 0}{u^T\boldsymbol{e}=0}}{\text{min}}, \ \ \frac{u^TL(w)u}{\Vert u\Vert^2}, \\
 \lambda_n[L(w)]&=\underset{\Vert v\Vert \neq 0}{\text{max}} \ \ \frac{v^TL(w)v}{\Vert v\Vert^2}
\end{split}
\end{equation} 

\begin{lemma}\label{lem:GapReach}
	Assume that $\lambda_N^1$ is the greatest largest Laplacian eigenvalue of the layers. Then, $H(c)$ is bounded below as follow
	\begin{equation}\label{GapBound}
	\begin{aligned}
	H(c)\geq\lambda_N^1-\frac{2c}{N}
	\end{aligned}
	\end{equation}
	 and  \eqref{GapBound} is a strict inequality for $c\neq 0$.
\end{lemma}

\begin{proof}
	We proved $\lambda_2\leq 2c/N$ in \citep{shakeri2015PRL} and showed in  \eqref{eq_lowerbound}  that $\lambda_n\geq\lambda_N^1$. Thus the lowerbound \eqref{GapBound} is trivial; furthermore, the equality $H(c)=\lambda_N^1-2c/N$ is possible only if $\lambda_n=\lambda_N^1$ and $\lambda_2=2c/N$. Therefore, the weight distribution must  satisfy the weight distribution conditions corresponding to the problems of maximizing $\lambda_2$ and minimizing $\lambda_n$ below the thresholds $c^*$ and $c^*_1$ respectively. However, a uniform weight distribution can not satisfy \eqref{opt_cond} for $v_N^1\neq0$ and $c\neq0$.

\end{proof}

	Lemma \ref{lem:GapReach} implies that, the lower bound \eqref{GapBound} is not reachable for $c>0$. However, since the lower bound is reachable for the extreme case of $c=0$, it follows from continuity, that the solution of \eqref{PrimalGap} approaches  the lower bound \eqref{GapBound}  for small budgets. That is
	\begin{equation}\label{GapAsymptote}
	\begin{aligned}
	H(c)\rightarrow\lambda_N^1-\frac{2c}{N} \ \ \text{as c} \rightarrow0
	\end{aligned}
	\end{equation}

\begin{remark}\label{rema:GapLarge}
	Since the algebraic connectivity $\lambda_2$ remains bounded for large budgets, according to \eqref{ConnecBound}, while the largest eigenvalue $\lambda_n$ grows unbounded after the threshold $c^\star$, it follows minimizing the gap $\lambda_n-\lambda_2$ is equivalent to minimizing $\lambda_n$ in such condition. Therefore, for large $c$, the solution to \eqref{PrimalGap}  approaches a uniform weight distribution. Moreover, according to \eqref{ConnecBound} and \eqref{F_c_rest}, the minimum gap is bounded as follows
	\begin{equation}\label{HBound}
	\begin{aligned}
	\frac{2c}{N}-\lambda_2\left(L^{\text{ave}}\right)\leqslant H(c)\leqslant\lambda_{\text{max}}(L^{ave})-
	\lambda_2\left(L^{\text{ave}}\right)+\frac{2c}{N}
	\end{aligned}
	\end{equation}  
\end{remark}

In Figure \ref{fig:GapPrimal}, we show the spectral width in a multiplex of two random geometric networks with 30 nodes. We see that the minimized gap is smaller than that of the uniform weight distribution in all budget regimes. The optimal gap can be investigated in three regimes of small, moderate, and large $c$. For small $c$, $c<c^*$, the gap approaches the linear asymptote \eqref{GapAsymptote} so that it begins from $\lambda_N^1$ at $c=0$ and decreases with slope $2c/N$. Then, it remains (almost) constant for moderate $c$, $c^*<c<c^\star$, and it increases (approximately) with slope $2c/N$ for large $c$, $c>c^\star$.  \\
\begin{figure}[!htb]
	\centering
	\includegraphics[clip,width=.75\columnwidth]{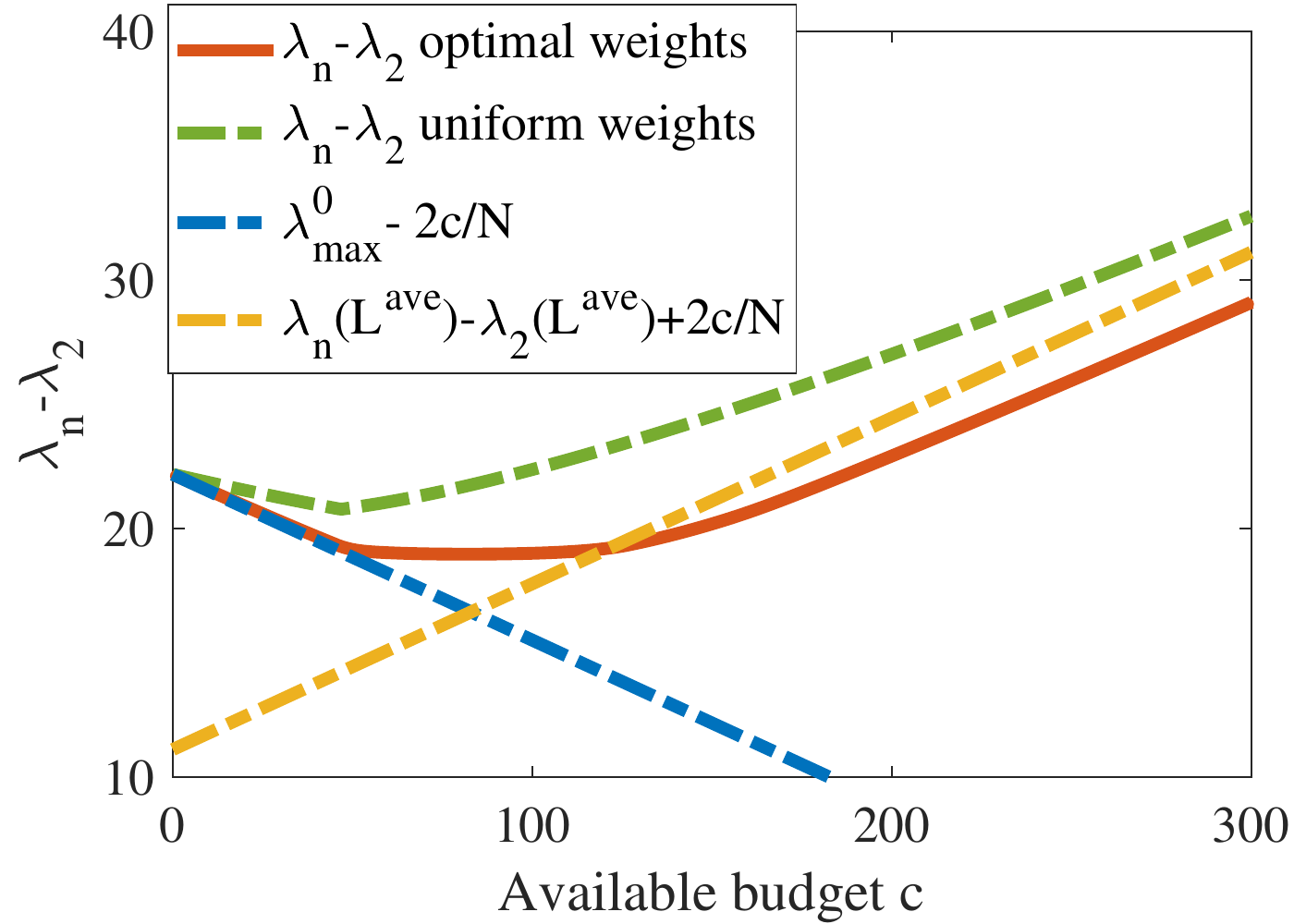}
	\caption{The optimal spectral width $\lambda_n-\lambda_2$ in a multiplex of two random geometric networks with 30 nodes and its related bounds.}
	\label{fig:GapPrimal}
\end{figure}

In Figure \ref{fig:GapAll}, we compare different gaps associated with different optimization problems for a multiplex of two random geometric networks with 30 nodes. According to Remark \ref{rema:GapLarge}, for large $c$, the gap obtained by minimizing $\lambda_n-\lambda_2$ approaches the gap obtained for the minimizer when minimizing only $\lambda_n$. For moderate $c$, this relation is inverted. Finally, for small $c$, the gap obtained for the weights that  maximize the algebraic connectivity is the same as the one for uniform weights.   while it grows with steep slope after the corresponding threshold $c^*$. The largest gap is thereby associated with the problem of maximizing the algebraic connectivity. 
\begin{figure}[!htb]
	\centering
	\includegraphics[clip,width=.75\columnwidth]{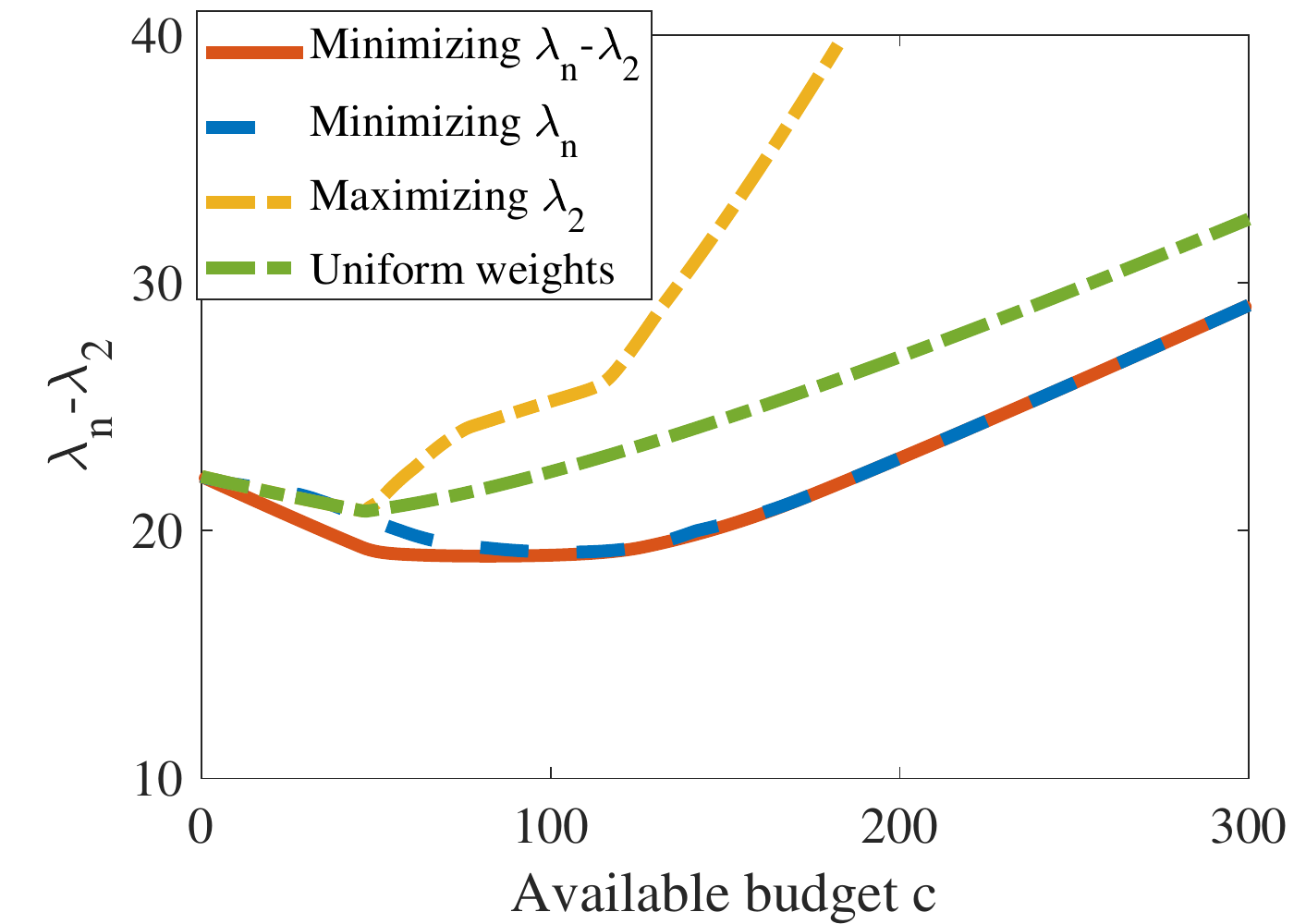}
	\caption{The spectral width $\lambda_n-\lambda_2$ associated with different optimization problems for a multiplex of two random geometric networks with 30 nodes: red is the optimal spectral width, dashed green is the spectral width for the weight distribution maximizing $\lambda_2$, dashed blue is the spectral width for the weight distribution minimizing $\lambda_n$.}
	\label{fig:GapAll}
\end{figure}

The above observations are supported by studying the optimal weight distributions obtained from different problems  (Figure \ref{fig:weight}). For small budgets, the weights that are  minimizing the gap are more similar to the uniform weights for maximizing the algebraic connectivity. However, minimizing the largest eigenvalue corresponds to completely nonuniform weights. For large budgets, the weight distribution for minimizing the gap is similar to weights for minimizing the largest eigenvalue (for more on this see SM \ref{app_gap_dim}).  
\begin{figure}[!htb]
		\centering
	\subfloat[\label{fig:w_gap}]{\includegraphics[clip,width=.5\columnwidth]{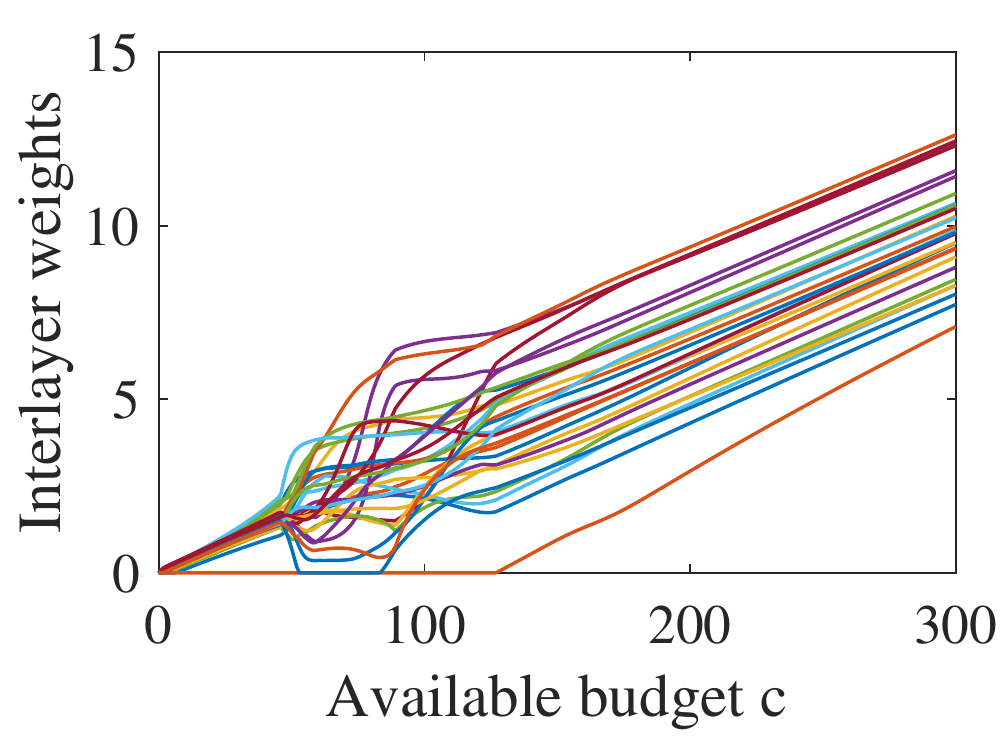}}
	\subfloat[\label{fig:w_lam_n}]{\includegraphics[clip,width=.5\columnwidth]{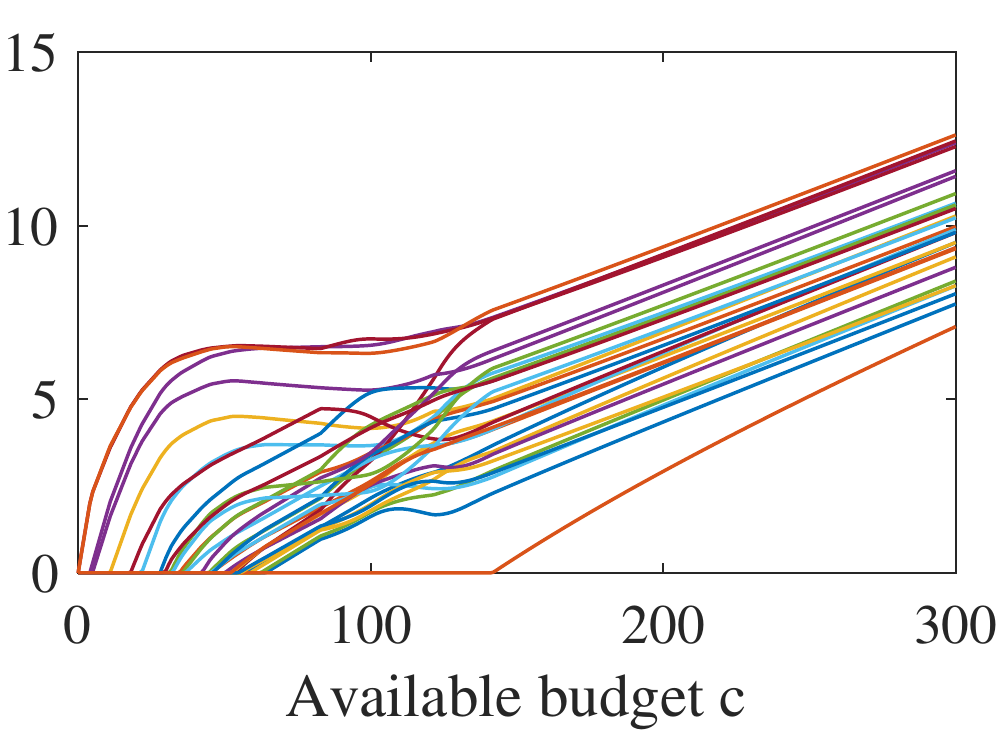}} \\
	\subfloat[\label{fig:w_lam_2}]{\includegraphics[clip,width=.5\columnwidth]{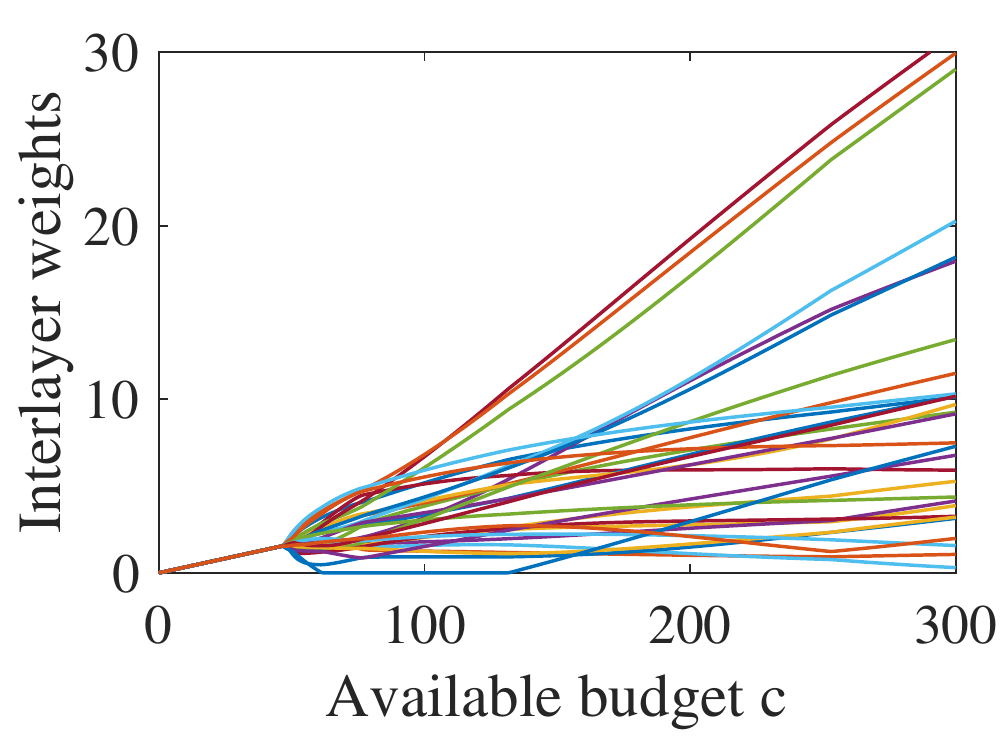}}
	\caption{Optimum weight distributions in a multiplex of two random geometric networks with 30 nodes associated with (a) Minimizing the gap. (b) Minimizing the largest eigenvalue. (c) Maximizing the algebraic connectivity.}
	\label{fig:weight}
\end{figure}

\subsection{The embedding problem associated to minimizing the spectral width }


The dual problem of \eqref{PrimalGap} is 
\begin{equation}\label{DualGap}
\begin{aligned}
& \underset{X,\ Y,\ \xi}{\text{maximize}}
& & c\xi+\langle Y,L_0\rangle -\langle X,L_0\rangle  \\
& \text{subject to}
& &  \langle X, I\rangle = 1\\
& & & \langle Y, I\rangle = 1\\
& & & \langle X, ee^T\rangle = 0\\
& & & \langle Y,E_{ij}\rangle-\langle X,E_{ij}\rangle -\xi \geq 0\\
& & & \xi\in \mathbb{R} ,\ X,\ Y\succeq 0 
\end{aligned}
\end{equation}

	It is an easy exercise to show that the primal problem \eqref{PrimalGap} and dual problem \eqref{DualGap} have feasible solutions, and thus common finite optimal value. Moreover,  strong duality holds and the optimal value is attainable.  

By Gram representations $X = U^TU$ and $Y = V^TV$, we obtain the embedding problem  as the following non-convex problem:
\begin{equation}\label{EmbeddingGap}
\begin{aligned}
& \underset{u_i,\ v_i,\ \xi}{\text{maximize}}
& & c\xi-\sum_{\lbrace i,j\rbrace\in E_1\cup E_2}\lVert u_i-u_j\rVert^2 +  \\
& & & \sum_{\lbrace i,j\rbrace\in E_1\cup E_2}\lVert v_i-v_j\rVert^2  \\
& \text{subject to}
& &  \sum_{i\in N}\lVert u_i\rVert^2 = 1\\
& & & \sum_{i\in N}\lVert v_i\rVert^2 = 1\\
& & &  \sum_{i\in N} u_i = 0\\
& & &  \lVert u_i-u_j\rVert^2-\lVert v_i-v_j\rVert^2+\xi \leq 0, \ \forall \lbrace i,j\rbrace\in E_3\\
& & & \xi\in \mathbb{R} ,\ u_i,\ v_i\in\mathbb{R}^n \ (i\in N)
\end{aligned}
\end{equation}

The following proposition guarantees this combined embedding is directly related to the single embeddings \eqref{Lambda2Embedding} and \eqref{Embedding_n}. Namely, the embedding problem \eqref{EmbeddingGap} has the same projection properties as in Propositions \ref{lem:ProjectionOfEmbedding} and \ref{lem:ProjLam_n}. 

\begin{proposition}\label{lem:ProjGap}
	The projections of the optimal $u$ and $v$ vectors onto one-dimensional subspaces yield eigenvectors to the second eigenvalue $\lambda_2(L)$ and  the largest eigenvalue $\lambda_n(L)$ respectively \footnote{Proof is similar to SM \ref{proof:ProjectionOfEmbedding}, and follows from complementary slackness. }.
\end{proposition}

The main consequence of Proposition \ref{lem:ProjGap} is that the multiplicities of $\lambda_2$ and $\lambda_n$ give upper bounds on the dimensionality of the $u$ and $v$ embeddings of \eqref{EmbeddingGap}, respectively. However, the dimensions of the combined embeddings do not necessarily match the individual ones, see SM \ref{app_gap_dim}.

\section{Conclusions}\label{sec:conclusion}
In this paper, we investigate the allocation of a budget $c$ on the interlayer links to optimize certain functions defined on the eigenvalues  of the Laplacian matrix for multiplex networks. The primal problem, in terms of Rayleigh quotients, and the dual problem, in terms of embeddings,  are useful to glean the optimal structure of the multiplex.
Maximizing the algebraic connectivity, $\lambda_2$, requires a uniform weight distribution up  to a threshold budget. While, for larger budgets, the optimal weights are generally non-uniform. 
The largest eigenvalue $\lambda_n$ remains constant if we invest the budget on the edges that correspond to the nodal lines of the layer with larger Laplacian spectral radius. Moreover, as we increase the budget, the optimal weights tend to become highly homogeneous, almost uniform.
Using these findings from the separate problems, maximizing $\lambda_2$ and minimizing $\lambda_n$, we analyze  the interlayer weights that minimize the spectral width $\lambda_n-\lambda_2$. There we see that for small budgets the problem behaves similarly to maximizing the algebraic connectivity and, for large budgets, minimizing the spectral radius.

\section*{Acknowledgment}
P. Poggi-Corradini is supported by NSF grant n. 1515810.
\bibliographystyle{chicago}
\bibliography{refs}

\clearpage
\newpage

\onecolumngrid
\appendix
\section*{Supplemental materials}\label{proofs} 

\subsection{Dual formulation of \eqref{eq:PrimalScaled} }\label{app_lagrangian}
The problem in \eqref{eq:PrimalScaled} is a standard convex semi-definite program (SDP) \citep[Chapter 4]{boyd2004convex}.
The Lagrangian will be
\begin{equation}
\begin{split}
L =& \sum_{i,j\in E_3}\hat{w}_{ij}+ 
\langle X,I- C\sum_{i,j\in E_3}\hat{w}_{ij}L_{ij} - (\sum_{ij\in E_3}\hat{w}_{ij})L_0-\mu e e^T\rangle-\langle Z,\hat{W}\rangle\\
=&\sum_{\lbrace i,j\rbrace\in E_3}\hat{w}_{ij}\left(1-\langle X, CL_{ij}+L_0\rangle-z_{ij}\right)+\langle X, I\rangle 
- \mu\langle X, ee^T\rangle
\end{split}
\end{equation}
and the dual is
\begin{equation}\label{DualFrame0}
\begin{aligned}
& \underset{X}{\text{maximize}}
& & \langle X, I\rangle \\
& \text{subject to}
& & \langle X, CL_{ij}+L_0\rangle\leq 1\ \ \text{for}\ \lbrace i,j\rbrace\in E_3 \\
& & & \langle X, ee^T\rangle= 0\\
& & & X\succeq 0
\end{aligned}
\end{equation}

The two SDP problems in \eqref{eq:PrimalScaled} and \eqref{DualFrame0}, form a dual pair that satisfy the weak and strong duality, i.e., for any feasible $\hat{\bw}$ and $X$:
\begin{equation}
\begin{split}
\sum_{ij\in E_3}\hat{w}_{ij}-\langle X,I\rangle &= 
\sum_{ij\in E_3}\hat{w}_{ij}-
\langle X,I\rangle +\sum_{ij\in E_3}\hat{w}_{ij}\langle X, CL_{ij}+
L_0\rangle \\
&- \sum_{ij\in E_3}\hat{w}_{ij}\langle X, CL_{ij}+L_0\rangle +\hat{\mu}\langle X, ee^T\rangle \\
&= 
\sum_{ij\in E_3}\hat{w}_{ij}\left(I-\langle X, CL_{ij}+L_0\rangle\right) \\
&+
\langle C\sum_{i,j\in E_3}\hat{w}_{ij}L_{ij}+(\sum_{ij\in E_3}\hat{w}_{ij})L_0-I+\hat{\mu}ee^T\rangle \langle X\rangle \geq 0.
\end{split}
\end{equation}

\subsection{Proof of Proposition \ref{lem:FeasibleDualSolution}}\label{proof:FeasibleDualSolution}
\begin{proof}
	Let
	\begin{equation}\label{EqApp}
	\begin{aligned}
	X=[u_1,\dots,u_n]^T[u_1,\dots,u_n]\\
	\end{aligned}
	\end{equation}
	where
	\begin{equation}\label{EqApp}
	\begin{aligned}
	u_i=\begin{cases} 
	\alpha_i\boldsymbol{h}, &  i\in{V_1}  \\
	\beta_i\boldsymbol{h},  &  i\in{V_2}
	
	\end{cases}
	\end{aligned}
	\end{equation}
	with $ \boldsymbol{h} \in{\mathbb{R}^n}, \lVert\boldsymbol{h}\rVert=1 $. If the constants $ \alpha_i $  and $ \beta_i $ exist such that
	\begin{equation}\label{eqapp}
	\begin{aligned}
	\begin{cases}
	\langle X,I \rangle = \sum_{i\in{V_1}} \alpha_i^2 + \sum_{i\in{V_2}}\ \beta_i^2=1  \\
	\langle X,\boldsymbol{e}\boldsymbol{e}^T \rangle =\left(\sum_{i\in{V_1}} \alpha_i\ + \sum_{i\in{V_2}} \beta_i \right)^2=0
	\end{cases}
	\end{aligned}
	\end{equation}
	then we can choose $ \xi\leq\min\left\{-\left(\alpha_i-\beta_i\right)^2\left|ij\in E_3, i\in V_1, j\in V_2 \right. \right\} $ and $ X $ is feasible for the dual problem. Since $ X $ is a Gram matrix, it is positive semidefinite. It is not a difficult task to show that Equation \eqref{eqapp} is solvable. Indeed one solution is simply $ \alpha_i=-\beta_i=\dfrac{1}{\sqrt{n}} $.
\end{proof} 
\subsection{Proof of Proposition \ref{lem:StrongDuality}}\label{proof:StrongDuality}
\begin{proof}
	Since the primal problem \eqref{lambda_2Primal} is a convex optimization \citep{shakeri2015PRL}, to show strong duality it is sufficient to show the Slater's constraint qualification is satisfied. First, it is observed that $ \left(\hat{\lambda}<0,\hat{\mu}\geq0,\hat{\omega}>0\right) $ is a strictly feasible primal solution whenever $ \sum_{ij\in E_3}\omega_{ij}=c $. The feasible primal solution  $ \left(\hat{\lambda}<0,\hat{\mu}\geq0,\hat{\omega}>0\right) $ and the existence of dual solution by Proposition \ref{lem:FeasibleDualSolution} gives rise to strong duality by the Slater's qualification condition. Since none of feasible sets are empty, the optimal value is finite. Therefore the dual attains its optimal solution. Moreover, the primal constraint gives $ \omega_{ij}\leq c $. Thus the weights remain in a compact subsets.  
\end{proof}
\subsection{Proof of Proposition \ref{lem:ProjectionOfEmbedding}}\label{proof:ProjectionOfEmbedding}
\begin{proof}
	Let $ U=\left[u_1\dots u_n\right] $ be an optimal embedding of \eqref{Lambda2Embedding} and $ \omega_{ij} $ be the corresponding optimal weights leading to the algebraic connectivity $ \lambda_2 $ in \eqref{lambda_2Primal}. Then $ X=U^TU $ is an optimal solution of \eqref{lambda_2Dual}. Since by KKT complementarity condition for the first inequality in \eqref{lambda_2Primal} we have $\left<X,\sum_{ij\in E_3}w_{ij}L_{ij}+L_0+\mu \boldsymbol{e} \boldsymbol{e}^T - \lambda_2 I\right>=0$, and by the second equality in \eqref{lambda_2Dual} $ \left<X,\boldsymbol{e}  \boldsymbol{e}^T \right>=0 $, it follows
	\begin{equation}
	\begin{aligned}
	0=\left<X,\sum_{\lbrace i,j\rbrace\in E_3} \omega_{ij}L_{ij}+L_0-\lambda_2I\right> \\
	=\left<U^TU,\sum_{\lbrace i,j\rbrace\in E_3} \omega_{ij}L_{ij}+L_0-\lambda_2I\right> \\
	=\left<I,U\left(\sum_{\lbrace i,j\rbrace\in E_3} \omega_{ij}L_{ij}+L_0-\lambda_2I\right)U^T\right>
	\end{aligned}
	\end{equation}
	where we use the condition $\text{trace}\left(U^TUA\right)=\text{trace}\left(UAU^T\right)$. The above condition shows that $\text{trace}\left[U\left(\sum_{\lbrace i,j\rbrace\in E_3} \omega_{ij}L_{ij}+L_0-\lambda_2I\right)U^T\right]=0$, so that \\ $\sum_{i=1}^ny_i^T\left(\sum_{\lbrace i,j\rbrace\in E_3} \omega_{ij}L_{ij}+L_0-\lambda_2I\right)y_i=0$ with $y_i$ the i-th row of $U$. Since $\lambda_2$ is the smallest nonzero eigenvalue, we know that each element $y_i^T\left(\sum_{\lbrace i,j\rbrace\in E_3} \omega_{ij}L_{ij}+L_0-\lambda_2I\right)y_i\geq0$. Therefore, $y_i^T\left(\sum_{\lbrace i,j\rbrace\in E_3} \omega_{ij}L_{ij}+L_0-\lambda_2I\right)y_i=0$. This indicates that each row of $U$, or column of $ U^T $, and accordingly the vector $ \nu=U^Tp $ with $ p $ arbitrary vector, is in the eigenspace of $ L=\sum_{\lbrace i,j\rbrace\in E_3} \omega_{ij}L_{ij}+L_0 $ to $ \lambda_2 $.
\end{proof}

\subsection{Proof of Lemma \ref{lem:clump}}\label{Proof:clump}
For $c<c^*$, the second eigenvalue of the supra-Laplacian is simple, and the corresponding eigenvector is a one-dimensional subspace given by the Fiedler vector $\boldsymbol{v} = \frac{1}{\sqrt{n}}\left[\boldsymbol{e};-\boldsymbol{e}\right]$ \citep{shakeri2015PRL}. Then, since by Proposition \ref{lem:ProjectionOfEmbedding} the projection of optimal embedding onto any arbitrary one-dimensional subspace is parallel to $\boldsymbol{v}$, i.e. $\exists \hspace{1mm} a>0 \hspace{1mm} \text{s.t.} \hspace{1mm} \left[\hat{u_1},\dots,\hat{u_n}\right]^T\boldsymbol{p}=a\boldsymbol{v} \hspace{3mm}\forall\hspace{1mm}\boldsymbol{p}\in\mathbb{R}^n$, it can result that the optimal embedding takes the form \eqref{eq:UniformEmbed}.

\subsection{Separator-shadow theorem}\label{Proof:Shadow-Separator}

\begin{theorem}\label{thm_Separator-shadow}
	Let the weights $w_{ij}\geq0$ and points $\hat{u}_i\in\mathbb{R}^n,\ i\in V,$ be optimal solutions of \eqref{eq:PrimalScaled} and \eqref{eq:EmbeddingScaled}, respectively. Let $S$ be a separator of $G_w$ partitioning the graph as $V=S\cup C_1\cup C_2$ with no edge between $C_1$ and $C_2$. Then, for at least one $j\in\{1,2\}$
	\begin{align*}
	\text{conv}\{0,\hat{u}_i\}\cap\{\hat{u}_s:s\in S\}\neq\emptyset \ \ \ \ \forall i\in C_j.
	\end{align*}
	That means, there is at least one $j\in\{1,2\}$ such that, the straight line between the origin and each $\hat{u}_i,i\in C_j,$ intersects the convex hull of the points in $S$.
\end{theorem}
First, the following Lemma \ref{lem:FiedlerPartition} due to Fiedler (see Theorem 3.3 in  \citep{Fiedler1975Partition}) plays a crucial role in characterizing the graph structure based on a Fiedler vector, i.e. an eigenvector to the second smallest eigenvalue. First, denoting a Fiedler vector by $y=y(i)=[y_i], \  \forall i \in V$, the coordinates of $y$ can be assigned to the vertices of $G$. This is called a \textit{characteristic valuation}.
\begin{lemma}\label{lem:FiedlerPartition}
	Let $G=\left( V,E\right)$ be a finite connected graph with positive weights $w_{ij}$. Let $y_i$ be a characteristic valuation of $G$, and for any $\alpha\geq0$, define $V_\alpha = \{i\in V|y_i+\alpha\geq0\}$. Then the subgraph $G(\alpha)$ induced by $G$ on $V_\alpha$ is connected.
\end{lemma}
\begin{figure}
	\centering
	\includegraphics[clip,width=.4\columnwidth]{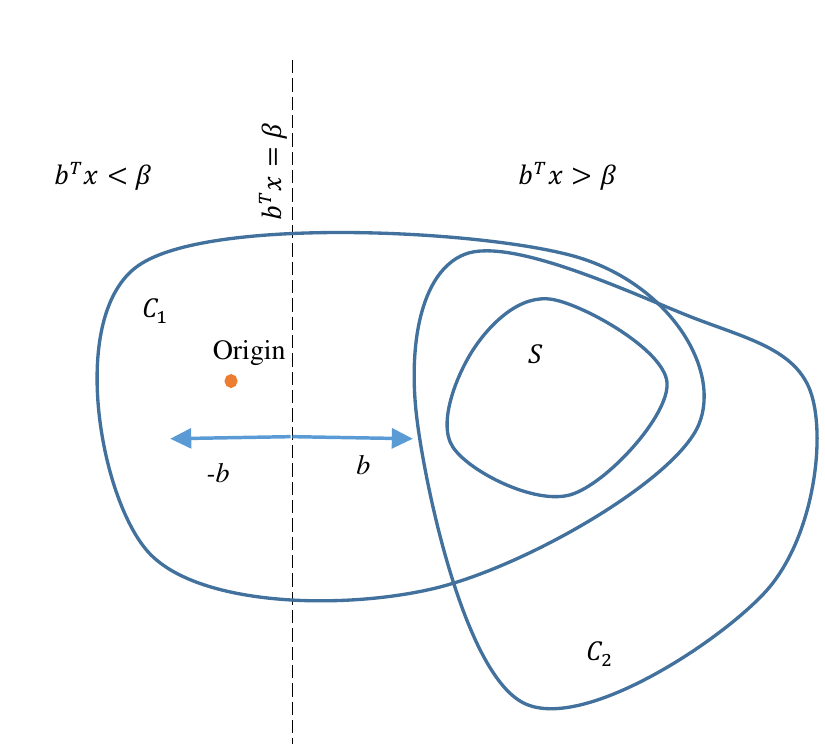}%
	\caption{Different graph partitions in Lemma \ref{lem:Hyperplane}}
	\label{fig:hyperplane}
\end{figure}
Now, having already shown Proposition \ref{lem:ProjectionOfEmbedding}, we can conclude the following Lemma \ref{lem:Hyperplane} for multiplex networks.  
\begin{lemma}\label{lem:Hyperplane}
	Let the weights $w_{ij}\geq0$ and points $\hat{u}_i\in\mathbb{R}^n,\ i\in V,$ be optimal solutions of \eqref{eq:PrimalScaled} and \eqref{eq:EmbeddingScaled}, respectively. Let $S$ be a separator of $G_w$ partitioning the graph as $V=S\cup C_1\cup C_2$ with no edge between $C_1$ and $C_2$. Suppose there is a normalized $b\in\mathbb{R}^n$ and $\beta>0$ defining a hyperplane $b^Tx=\beta$ within the subspace $\text{span}\{\hat{u}_i,i\in V\}$ that separates all points of $S$ from at least one point $\hat{i}\in C_1$, namely
	\begin{align*}
	b^T\hat{u}_{\hat{i}}\leq\beta<b^T\hat{u}_i, \ \forall i\in S.
	\end{align*}
	Then
	\begin{align*}
	b^T\hat{u}_i>\beta, \ \forall i \in C_2.
	\end{align*}
\end{lemma}
\begin{proof}
	By Proposition \ref{lem:ProjectionOfEmbedding} we know that the vector $-\left[b^T\hat{u}_1,\dots,b^T\hat{u}_n\right]^T$ is an eigenvector to $\lambda_2\left(L\left(G_w\right)\right)$. Then, setting $\alpha=\beta$ in Lemma \ref{lem:FiedlerPartition}, it follows that the subgraph $G\left(\beta\right)=\left(V_\beta,E_w\right)$ with $V_\beta=\{i\in V:b^T\hat{u}_i\leq\beta\}$ is connected. Therefore, since there is no edge between $C_1$ and $C_2$, the set $V_\beta$ can not simultaneously contain nodes from both $C_1$ and $C_2$, and since it already contains a node form $C_1$, i.e. $\hat{i}\in V_\beta$, it follows there is no point of $C_2$ in $V_\beta$.
\end{proof}
Lemma \ref{lem:Hyperplane} means that, of $C_1$ and $C_2$, only at most one part can be separated by the hyperplane $b^Tx=\beta$ from $S$ (see Figure \ref{fig:hyperplane}). This is the key to reach separator-shadow theorem.
\begin{figure}
	\centering
	\includegraphics[clip,width=.4\columnwidth]{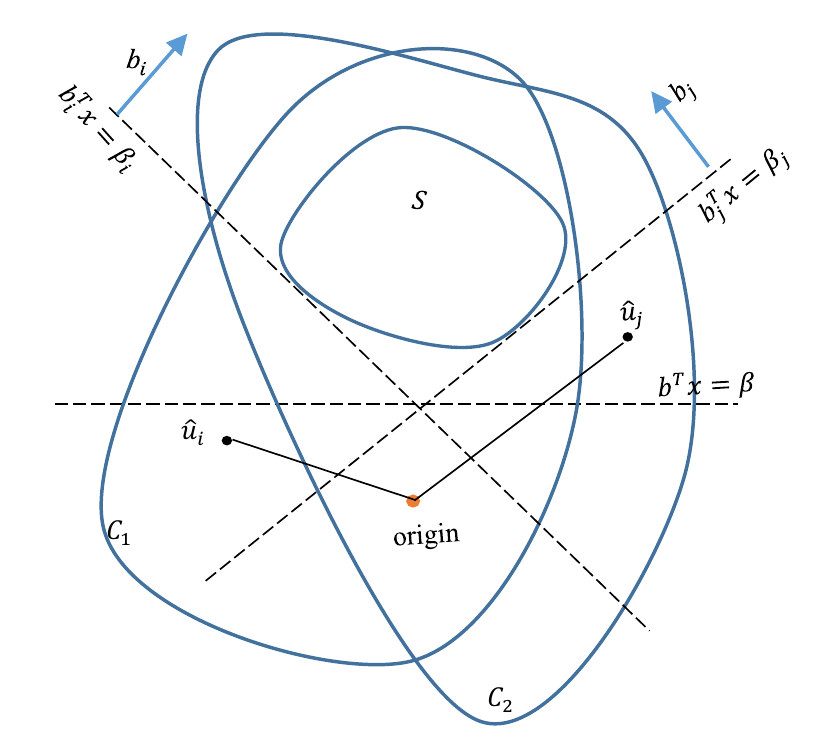}%
	\caption{Graph embedding contradicting separator-shadow criterion}
	\label{fig:separator-shadow}
\end{figure}
\\

\begin{proof}[Proof of Theorem \ref{thm_Separator-shadow}]. It is supposed that the origin is not contained in the convex hull of the points in $S$, and none of $S$, $C_1$, and $C_2$ are empty; otherwise the theorem holds trivially. \\ 

The proof is by a contradiction argument. If the theorem does not hold, then there are points $\hat{u}_i$ and $\hat{u}_j$ with $i\in C_1$ and $j\in C_2$ such that $\text{conv}\left[0,\hat{u}_i\right]\cap \mathcal{S}=\text{conv}\left[0,\hat{u}_j\right]\cap \mathcal{S}=\emptyset$, with $\mathcal{S}=\text{conv}\{\hat{u}_s:s\in S\}$ (Figure \ref{fig:separator-shadow}). Then, convexity shows for any $k\in\{i,j\}$ there is a hyperplane in the space spanned by the solution of \eqref{eq:EmbeddingScaled} that separates the line segment $\text{conv}\left[0,\hat{u}_k\right]$ from $\mathcal{S}$, i.e. there are $b_k\in\text{span}\{\hat{u}_r:r\in V\}$ and $\beta_k>0$, such that $b_k^Tx\geq\beta_k \ \forall x\in\mathcal{S}$ and $b_k^Tx<\beta_k \ \forall x\in\text{conv}\left[0,\hat{u}_k\right]$. \\ 

Next it is seen that there can be found a convex combination of the two hyperplanes by an $\alpha\in\left[0,1\right]$  such that, for $b\left(\alpha\right)=\left(1-\alpha\right)b_i+\alpha b_j$ and $\beta\left(\alpha\right)=\left(1-\alpha\right)\beta_i+\alpha\beta_j$, the half-space $\mathcal{H}=\{x:b\left(\alpha\right)^Tx<\beta\left(\alpha\right)\}$ contains points of both $C_1$ and $C_2$. This is because $\mathcal{H}$, that includes origin for all $\alpha\in\left[0,1\right]$, contains $\hat{u}_i\in C_1$ for $\alpha=0$ and $\hat{u}_j\in C_2$ for $\alpha=1$. Now, if for contradiction, when continuously varying $\alpha$ from $0$ to $1$, there is no $\alpha$ for which $\mathcal{H}$ contains points from both $C_1$ and $C_2$, there should exist an $\alpha=\bar{\alpha}$ where $\mathcal{H}$ leaves $C_1$ before it meets $C_2$. In such condition, the hyperplane $b\left(\bar{\alpha}\right)^Tx=\beta\left(\bar{\alpha}\right)$ would separate the origin from $\text{conv}\{\hat{u}_r:r\in V\}$, thus violating the equality constraint in \eqref{eq:EmbeddingScaled}. \\ 

The fact that $\mathcal{H}$ contains points from both $C_1$ and $C_2$ contradicts Lemma \ref{lem:Hyperplane}, and thus the points $\hat{u}_i$ and $\hat{u}_j$ assumed initially can not exist.
\end{proof}

The situation can be checked for the separator involving the node set $\{16,11,12,20,8,38\}$ in the multiplex network of Figure \ref{fig:SepShadEmb} .
\begin{figure}[!htb]
	\subfloat[Initial multiplex graph]
	{
		\includegraphics[clip,width=.5\columnwidth]{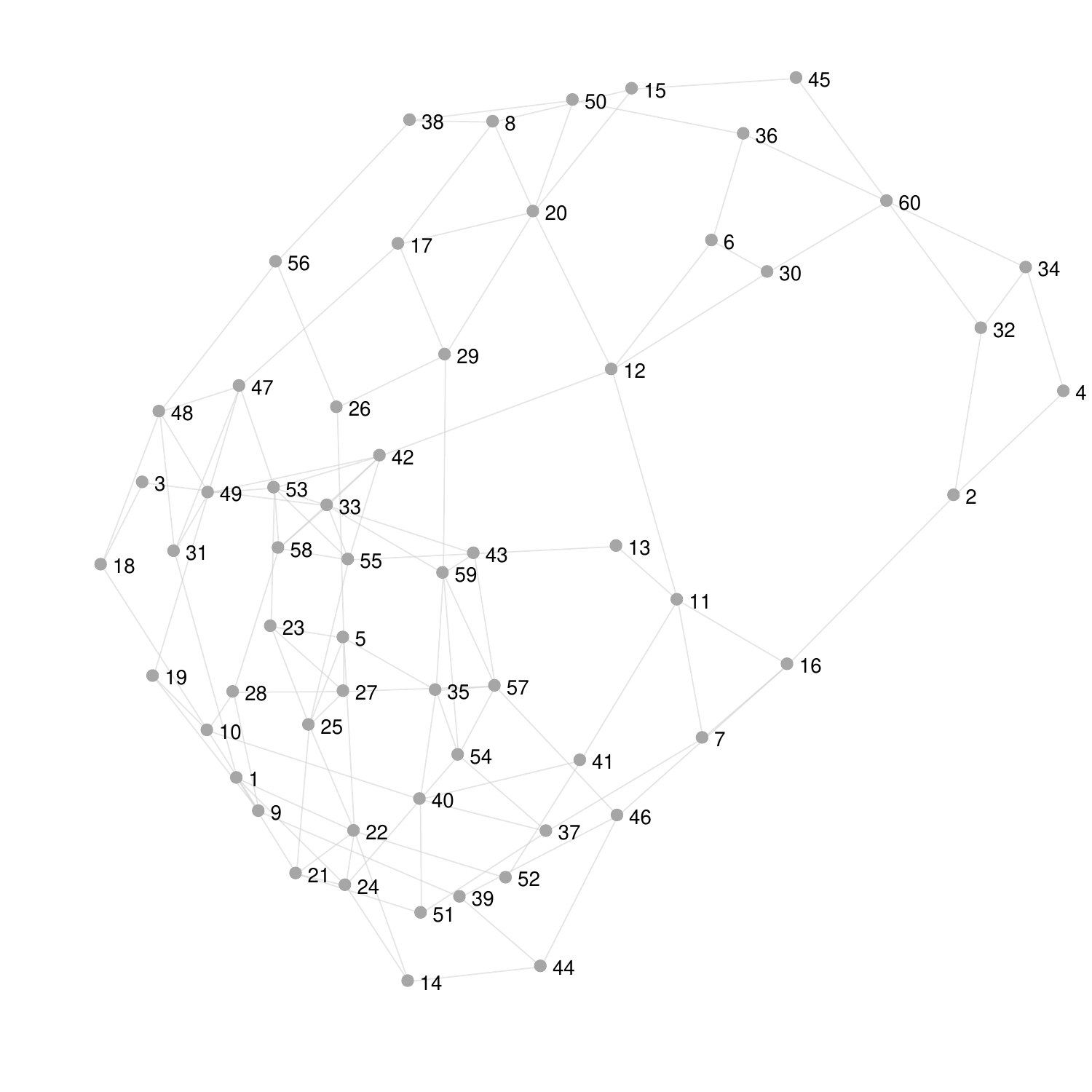}%
	}\\
	\subfloat[Embedded graph]
	{
		\includegraphics[clip,width=.5\columnwidth]{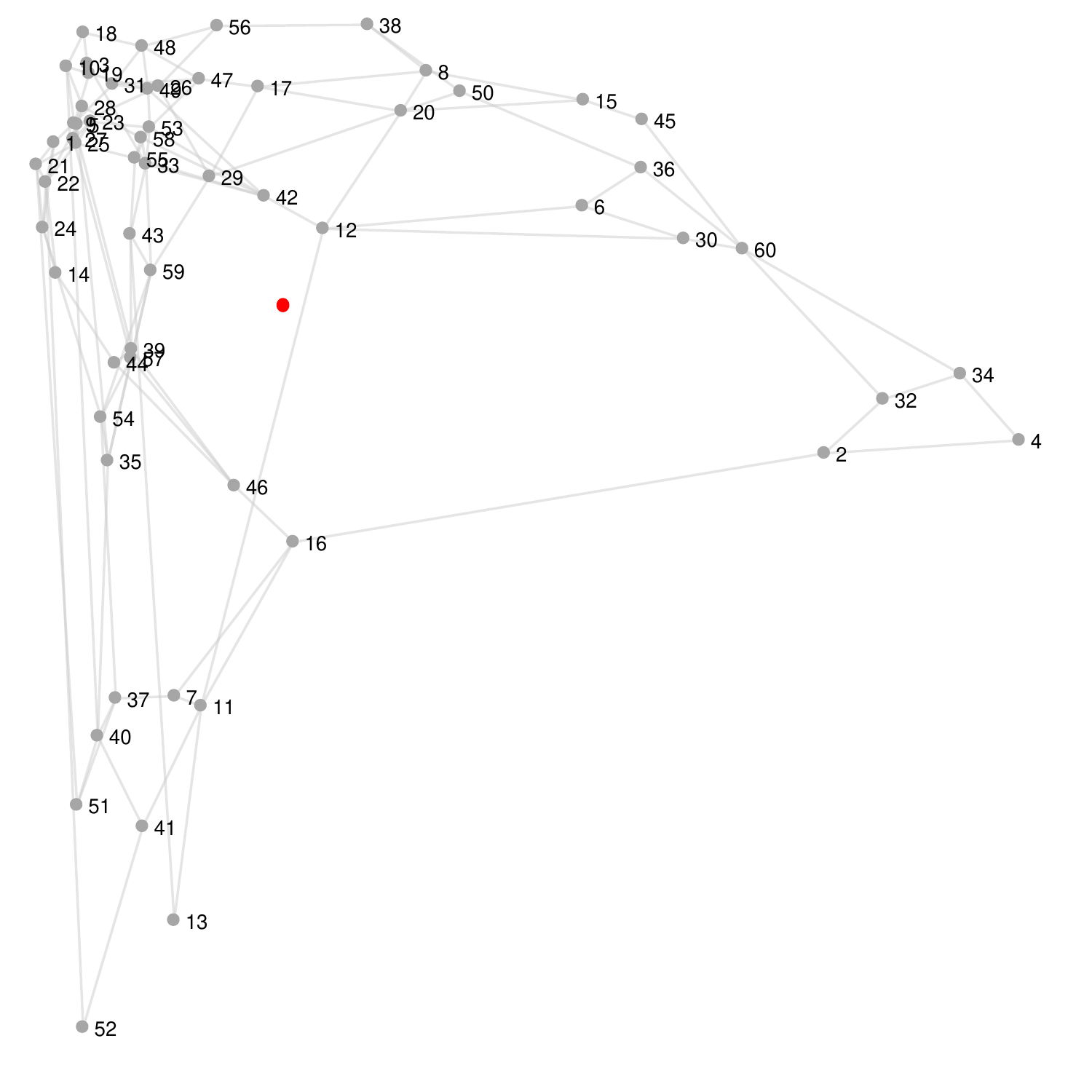}%
	}
	\caption{2-D embedding of two 30-node random geometric graphs: numbers 1-30 are the node set for $G_1$ and numbers 31-60 indicate the nodes of graph $G_2$. The red circle shows origin. The separated node set $\{2,4,6,15,30,32,34,36,40,45,50,60\}$ in the initial graph is embedded at the shadow of the separator constituted by the node set $\{16,11,12,20,8,38\}$.}
	\label{fig:SepShadEmb}
\end{figure}

\subsection{More results on maximizing the spectral radious of the Laplacian matrix}\label{app_radius}
In Figures \ref{fig:lam_n_all} and \ref{fig:lam_2_all}, we plot, respectively, the largest eigenvalue and the second eigenvalue associated with different optimization problems. In Figure \ref{fig:lam_n_all}, while maximizing the algebraic connectivity $\lambda_2$ generates the largest eigenvalue particularly after the threshold $c^*$, the problems of minimizing $\lambda_n$ and minimizing $\lambda_n-\lambda_2$ lead to similar results for the largest eigenvalue with some difference in small $c$. On the other hand, Figure \ref{fig:lam_2_all} illustrates that the algebraic connectivity resulting from minimizing the gap is approximately equal to that of maximizing $\lambda_2$ for small budgets and that of minimizing $\lambda_n$ for large budgets. \\

For moderate budgets, the algebraic connectivity of minimizing the gap is smaller than the value obtained by uniform weight distribution. Moreover, the problem of minimizing $\lambda_n$ generates the minimum algebraic connectivity at all. A further result is that, from perspective of the algebraic connectivity in Figure \ref{fig:lam_2_all}, the problem of minimizing the gap works based on (is approximately equivalent to) maximizing the algebraic connectivity for small budgets smaller than the threshold $c^*$, while it operates through (is approximately equivalent to) minimizing the largest eigenvalue $\lambda_n$ for large $c$, particularly after the threshold $c^\star$. For the largest eigenvalue in Figure \ref{fig:lam_n_all}, on the other hand, minimizing the gap operates similar to minimizing the largest eigenvalues for all budgets. The full spectra of supra-Laplacians associated with different primal problems are seen in Figure \ref{fig:spectrum}.   \\

%

\begin{figure*}
	\subfloat[\label{fig:lam_n_all}]{\includegraphics[clip,width=.4\columnwidth]{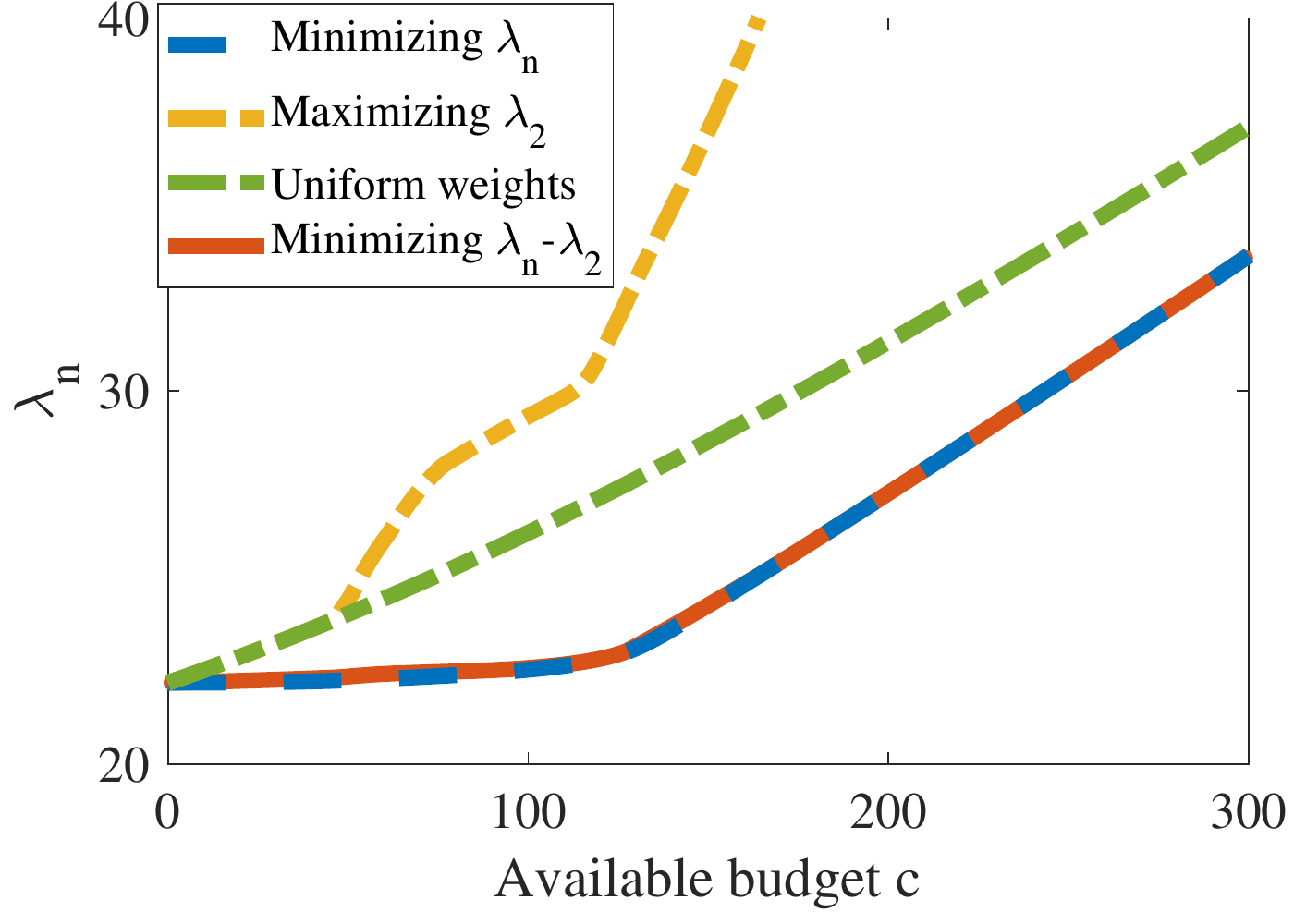}} \ \ \ \ \ \
	\subfloat[\label{fig:lam_2_all}]{\includegraphics[clip,width=.4\columnwidth]{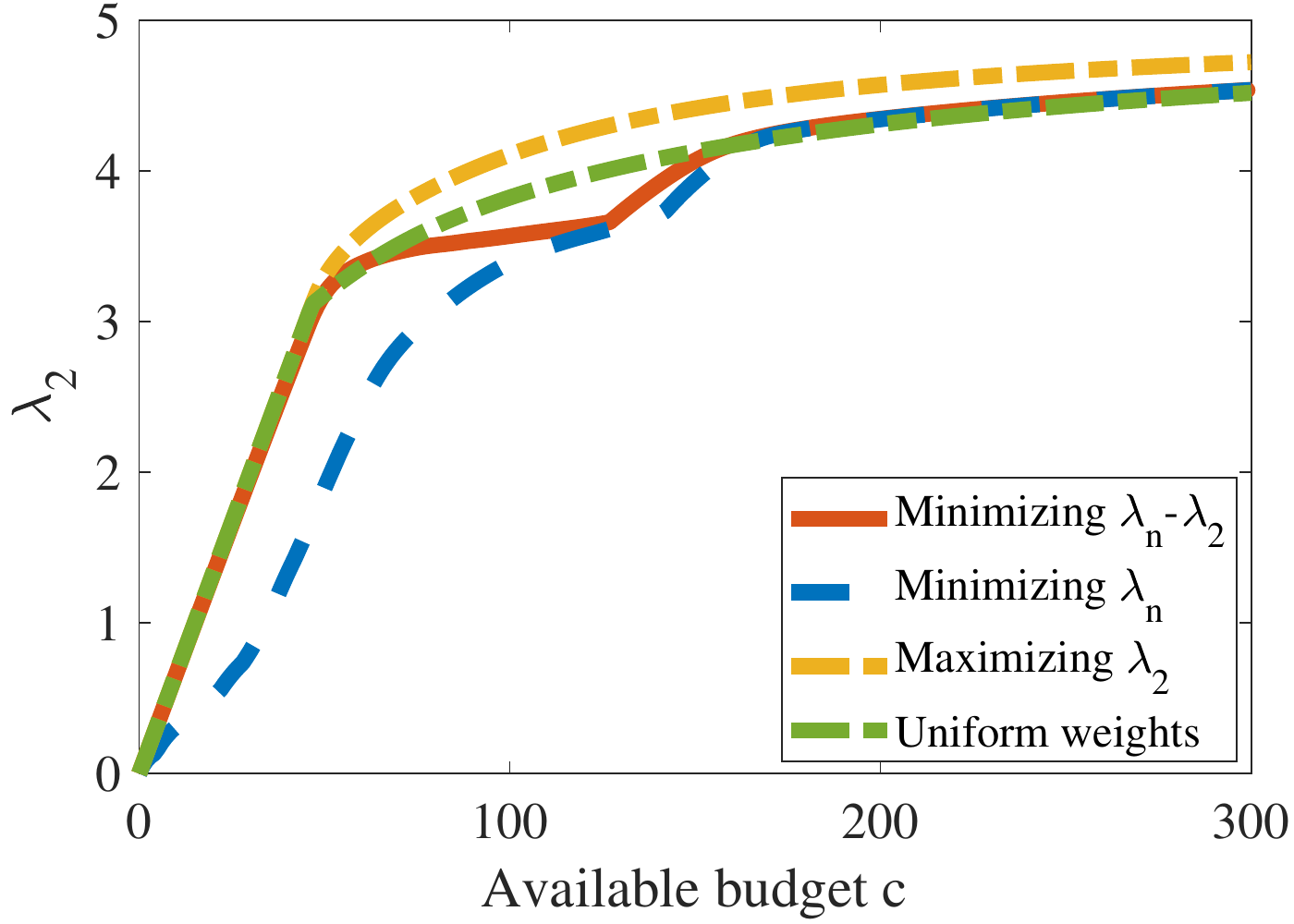}}
	\caption{(a) The largest eigenvalue $\lambda_n$ associated with different optimization problems for a multiplex of two random geometric networks with 30 nodes. (e) The algebraic connectivity $\lambda_2$ associated with different optimization problems for a multiplex of two random geometric networks with 30 nodes.}
\end{figure*}

\begin{figure*}
	\centering
	\subfloat[\label{fig:gap_spect}]{\includegraphics[clip,width=.4\columnwidth]{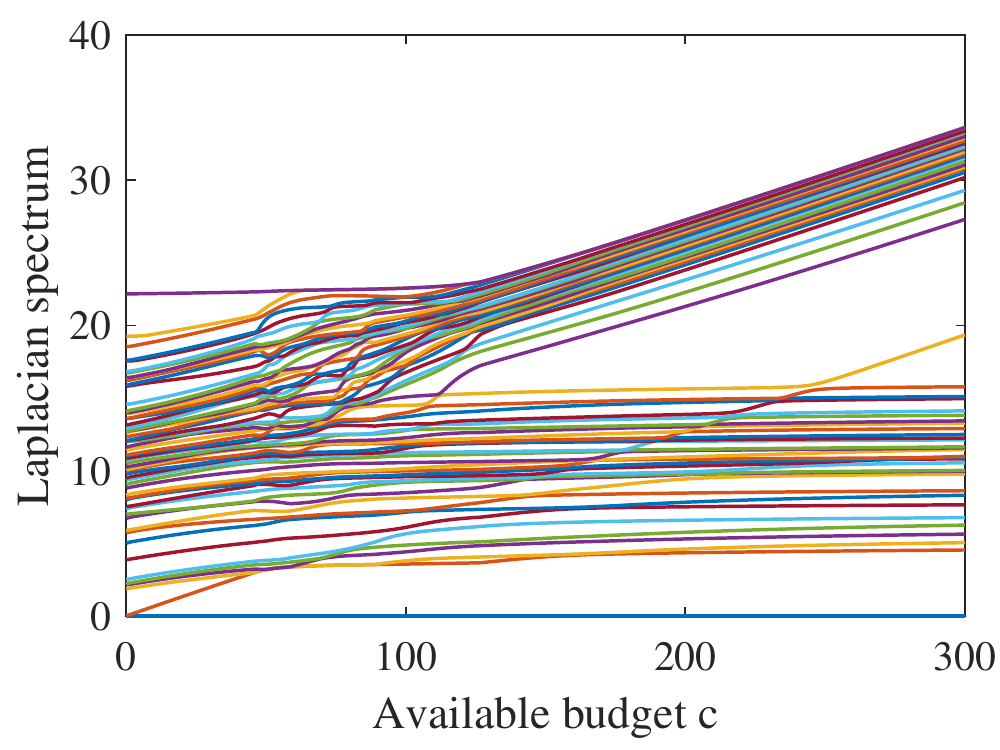}} \ \ \ \ \ \
	\subfloat[\label{fig:lam_n_spect}]{\includegraphics[clip,width=.4\columnwidth]{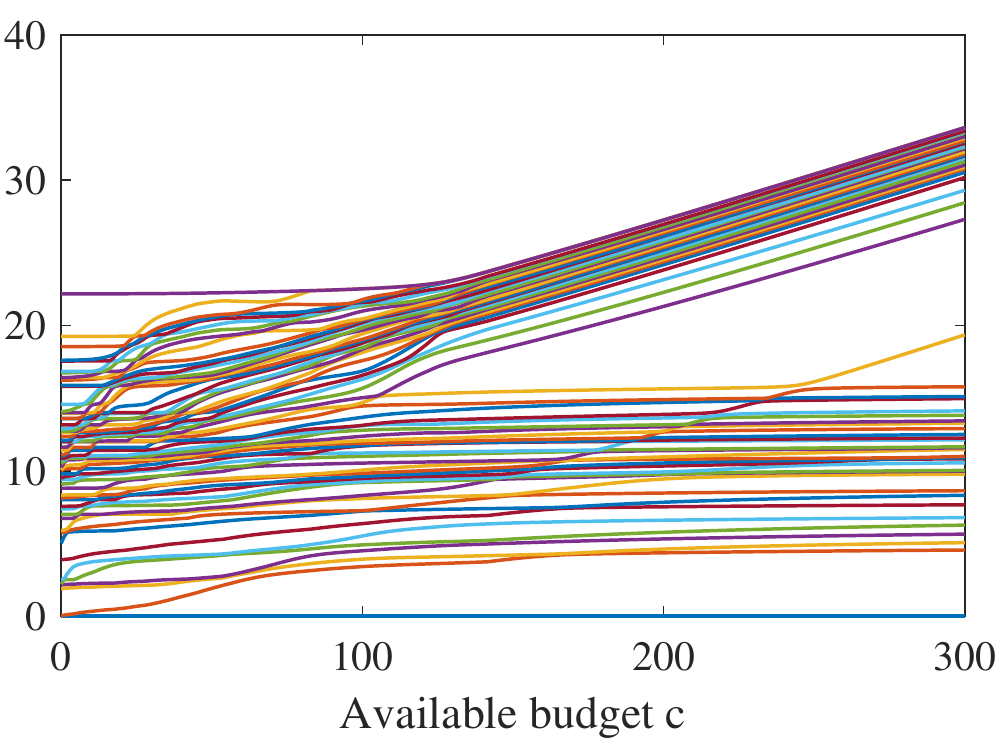}} \\
	\subfloat[\label{fig:lam_2_spect}]{\includegraphics[clip,width=.4\columnwidth]{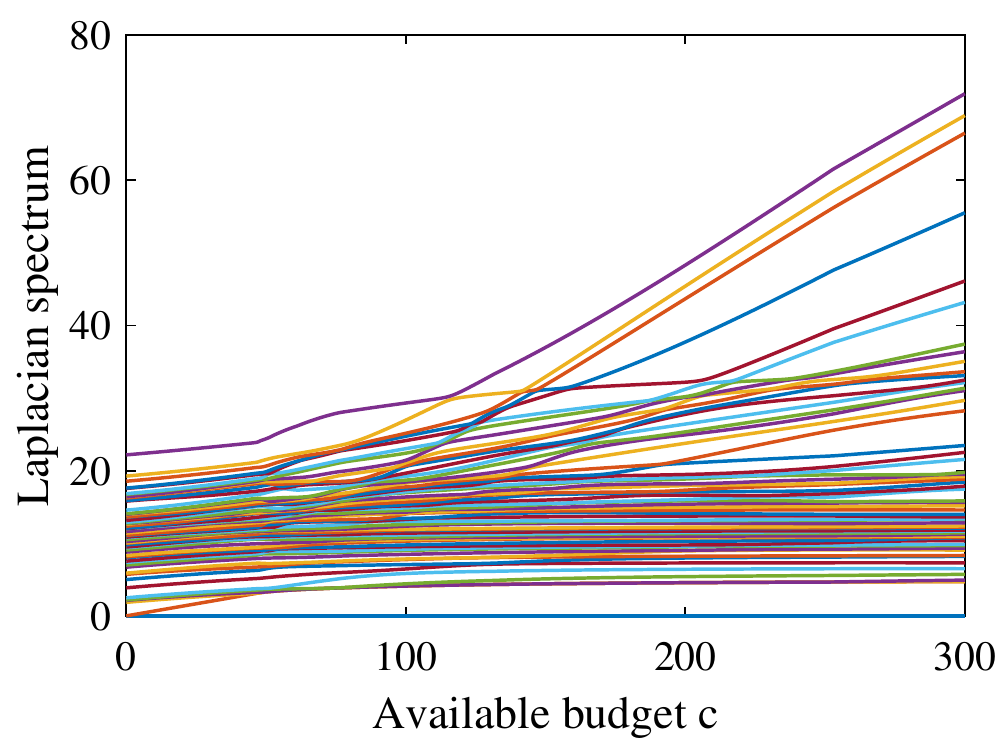}}
	\caption{The spectrum of supra-Laplacian in a multiplex of two random geometric networks with 30 nodes associated with (a) Minimizing the gap. (b) Minimizing the largest eigenvalue. (c) Maximizing the algebraic connectivity.}
	\label{fig:spectrum}
\end{figure*}
%
%

\subsection{Dimensionality of the combined embedding associated to minimizing the spectral width}\label{app_gap_dim}
%

Figure \ref{fig:GapEmbdLam} illustrates the dimensions of $u$- and $v$-embeddings in different regimes of the budget for a multiplex of two Watts-Strogatz networks, each with 30 nodes. 
\begin{figure*}[!htb]
	\centering
	\includegraphics[clip,width=.4\columnwidth]{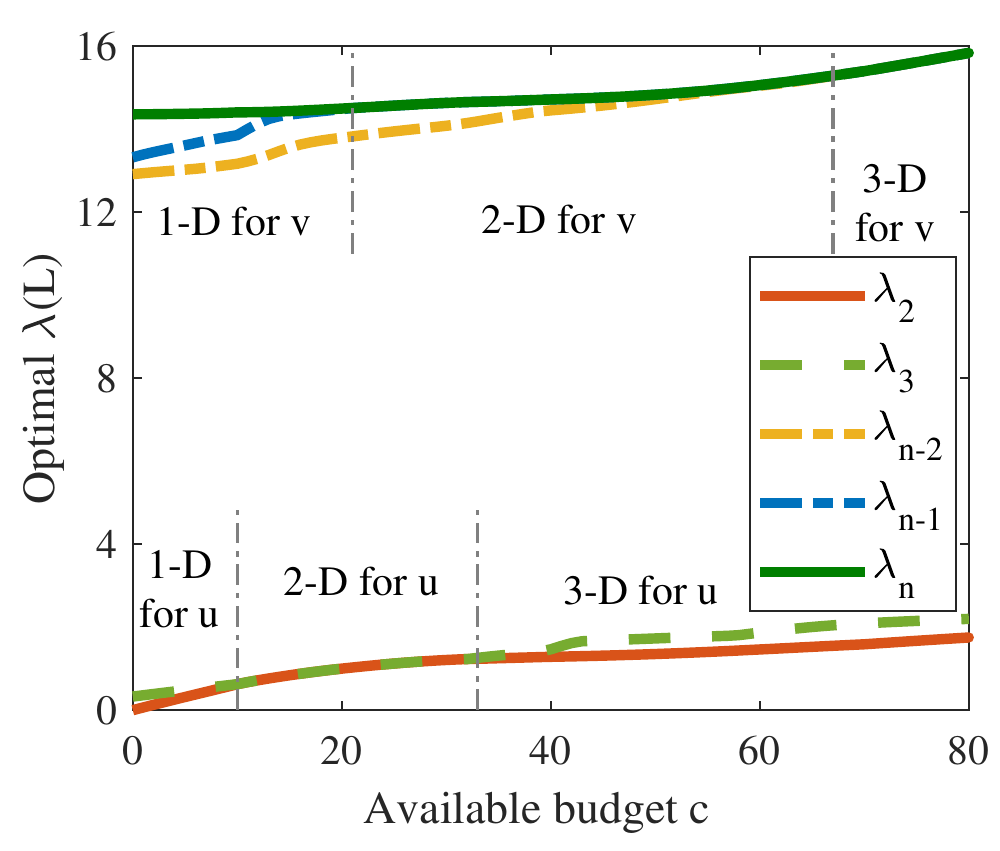}  
	\caption{Smallest and largest eigenvalues of $L$ for two 30-node Watts-Strogatz networks, obtained by minimizing the spectral width}
	\label{fig:GapEmbdLam}
\end{figure*}

Moreover, we show the  corresponding embeddings for  $\lambda_2$ and  $\lambda_n$ in minimum width problem in Figure \ref{fig:GapEmbc5}(a-b) for a small budget $c=5$ together with their counterparts in maximum $\lambda_2$ and minimum $\lambda_n$ problems  in Figure \ref{fig:GapEmbc5}(c-d). Whilst $u$- and $v$-embeddings are one dimensional they are perturbation of the corresponding individual problems, with more perturbation in $u$-embedding, and more similarity in $v$-embedding. This can be explained by the less rigid conditions in \eqref{opt_cond} compared to the optimal condition for maximizing algebraic connectivity. 

\begin{figure*}[!htb]
	\subfloat[\label{fig:GapEmbUc5}]{\includegraphics[clip,width=.3\columnwidth]{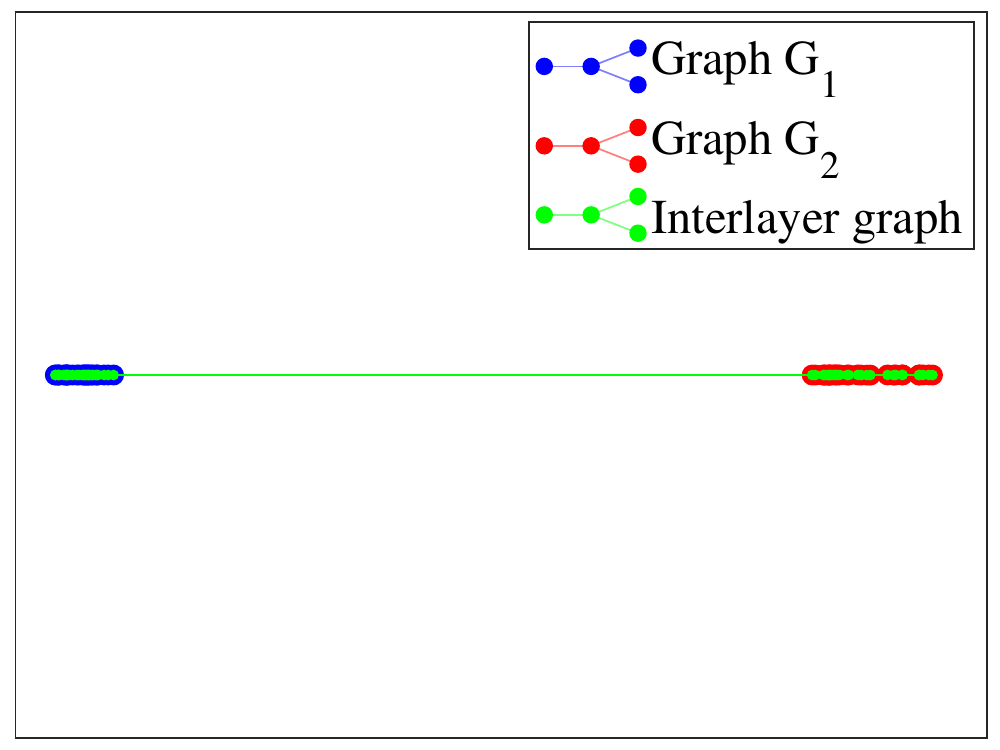}} \ \ \ \ \ \ \
	\subfloat[\label{fig:GapEmbVc5}]{\includegraphics[clip,width=.3\columnwidth]{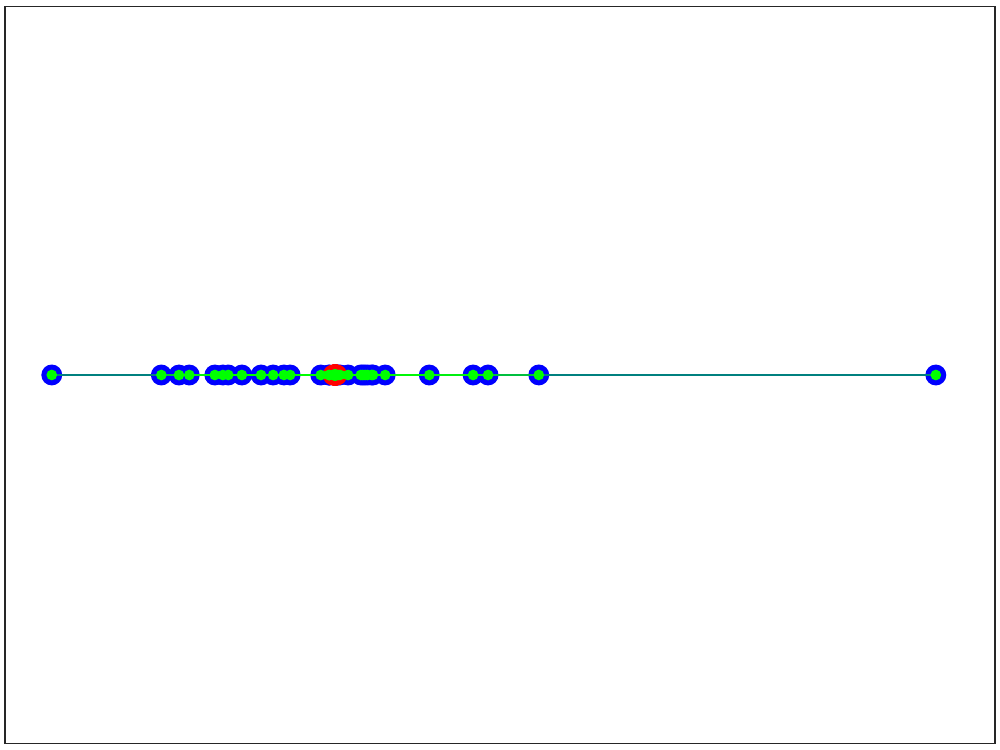}}
	
	\subfloat[\label{fig:GapEmbLam2c5}]{\includegraphics[clip,width=.3\columnwidth]{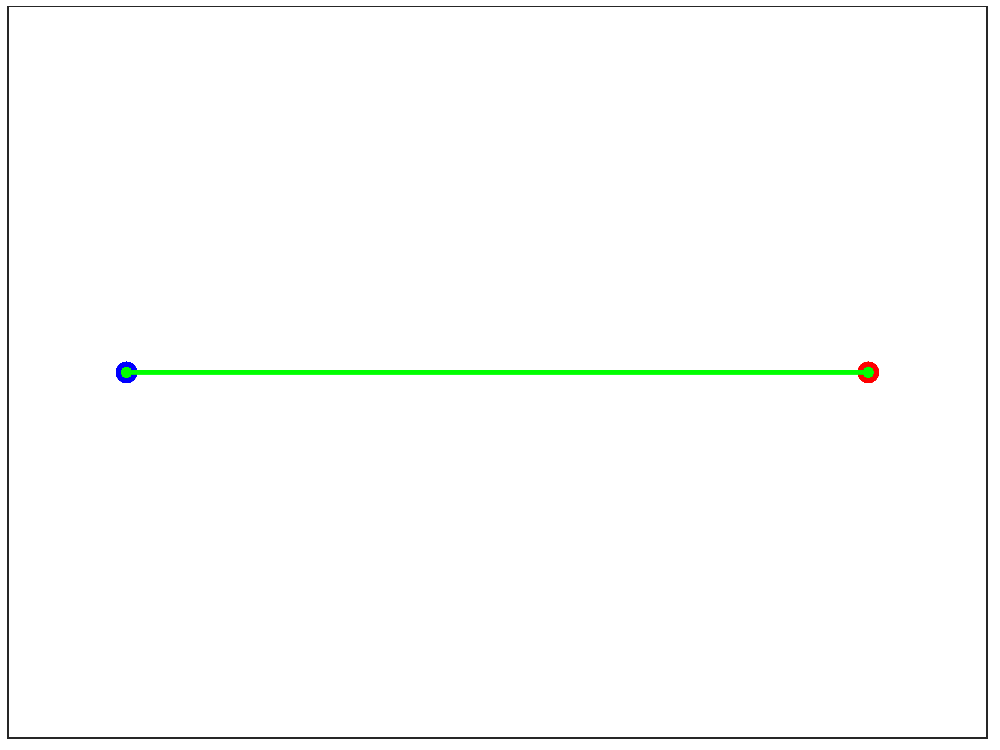}} \ \ \ \ \ \ \
	\subfloat[\label{fig:GapEmbLamnc5}]{\includegraphics[clip,width=.3\columnwidth]{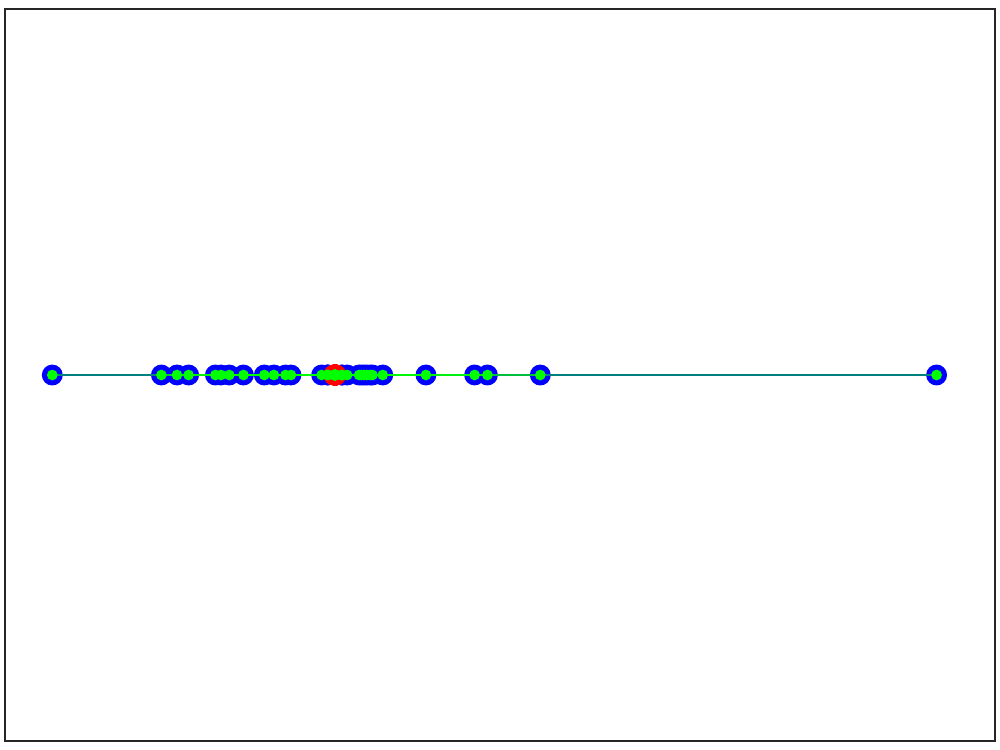}}
	
	\caption{(a) $u$- and (b) $v$-embedding corresponding to minimum spectral width, and embedding corresponding to (c) maximum $\lambda_2$ and (d) minimum $\lambda_n$ for a multiplex of two different Watts-Strogatz networks with 30 nodes for $c=5$.}
	\label{fig:GapEmbc5} 
\end{figure*}

Figure \ref{fig:GapEmbUc20} and \ref{fig:GapEmbLam2c5} illustrates the embeddings for $c=20$. The $u$-embedding and the corresponding embedding of maximizing $\lambda_2$ are both two dimensional. However, Figure \ref{fig:GapEmbc20} shows a 1-D embedding for $v$-embedding versus two dimensional embedding for minimizing $\lambda_n$. 
\begin{figure*}[!htb]
	\subfloat[\label{fig:GapEmbUc20}]{\includegraphics[clip,width=.3\columnwidth]{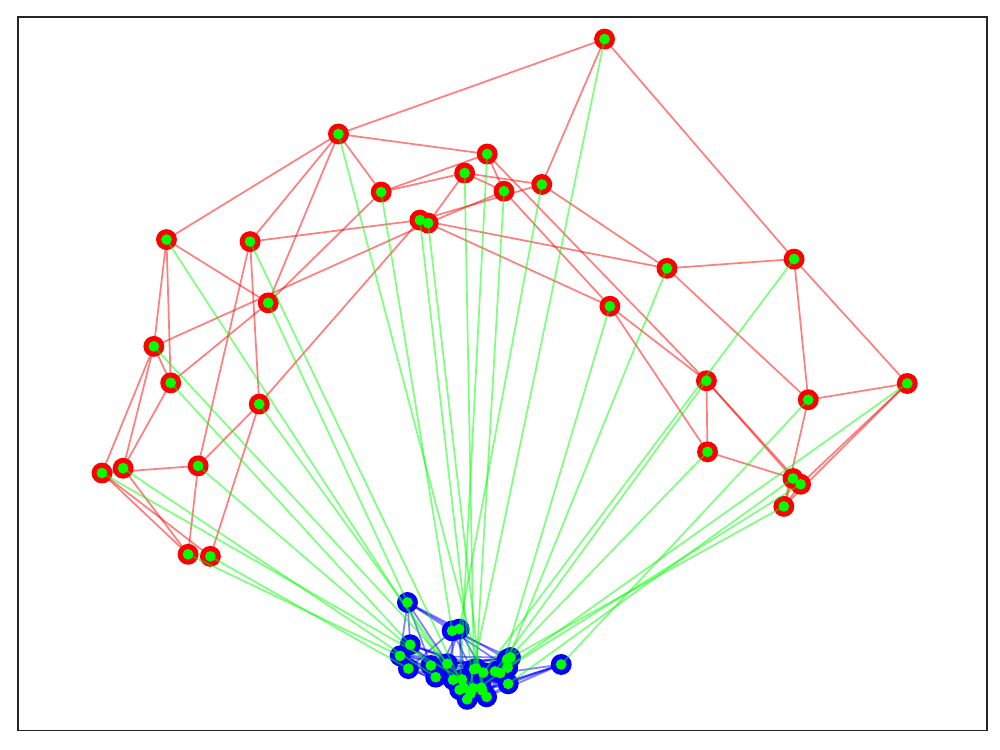}} \ \ \ \ \ \ \
	\subfloat[\label{fig:GapEmbVc20}]{\includegraphics[clip,width=.3\columnwidth]{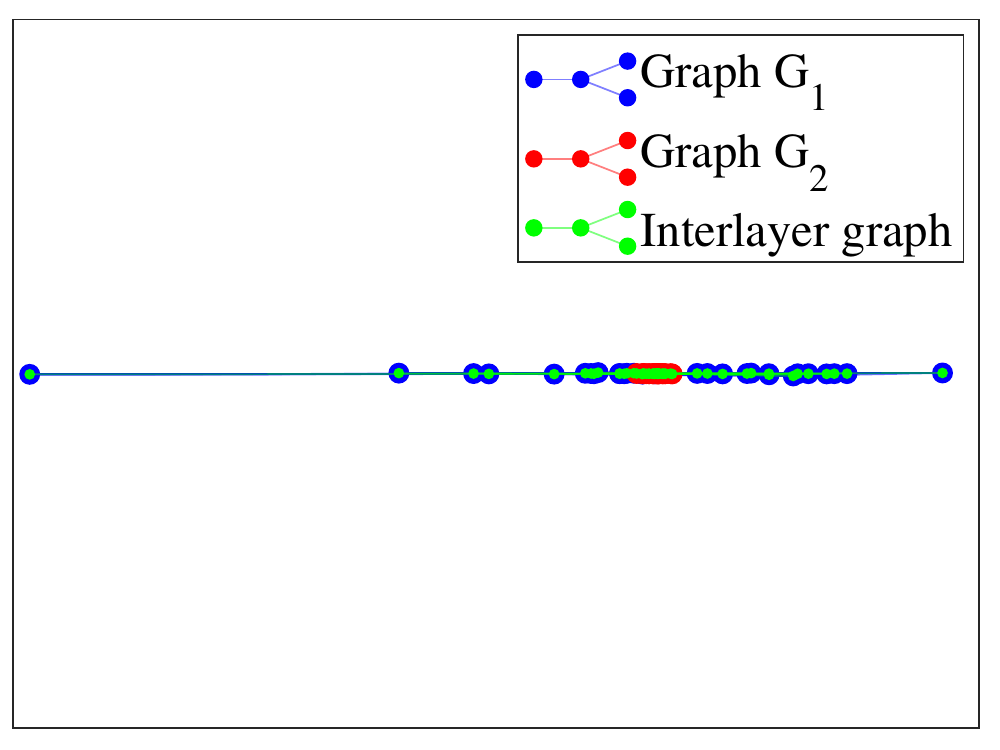}}
	
	\subfloat[\label{fig:GapEmbLam2c20}]{\includegraphics[clip,width=.3\columnwidth]{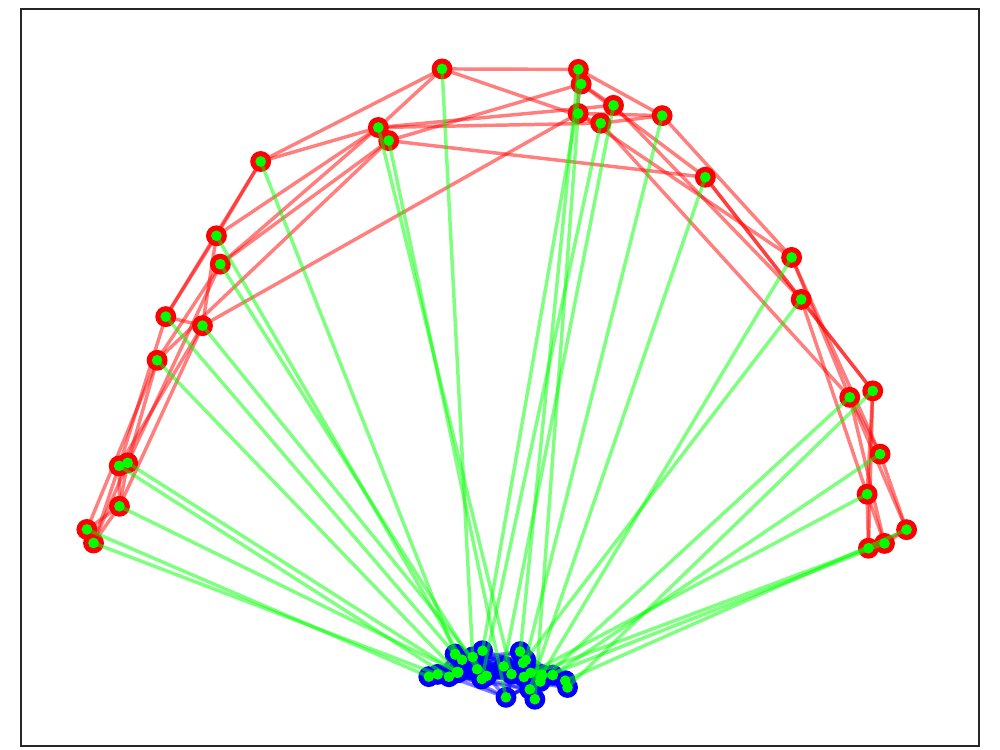}} \ \ \ \ \ \ \
	\subfloat[\label{fig:GapEmbLamnc20}]{\includegraphics[clip,width=.3\columnwidth]{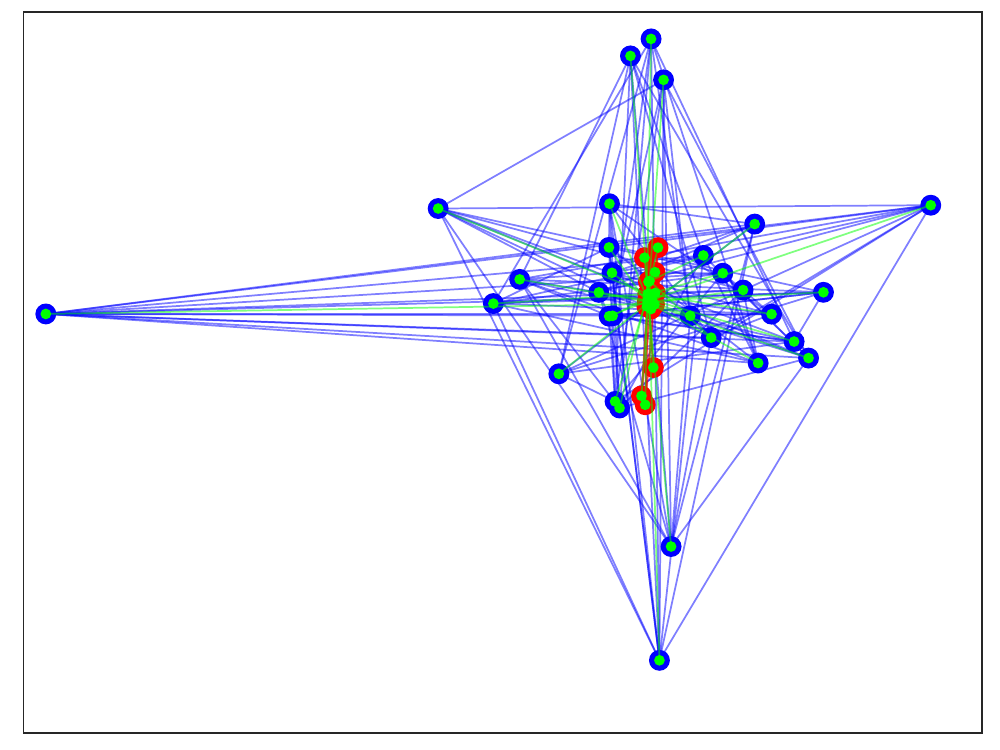}}
	
	\caption{(a) $u$- and (b) $v$-embedding corresponding to minimum spectral width, and embedding corresponding to (c) maximum $\lambda_2$ and (d) minimum $\lambda_n$ for a multiplex of two different Watts-Strogatz networks with 30 nodes for $c=20$.}
	\label{fig:GapEmbc20} 
\end{figure*}
%
%
Figure \ref{fig:GapEmbc50} shows the embeddings for $c=50$ with one and two dimensional $u$- and $v$-embeddings are comparable with three dimensional embeddings of individual problems.  

\begin{figure}[!htb]
	\subfloat[\label{fig:GapEmbUc50}]{\includegraphics[clip,width=.3\columnwidth]{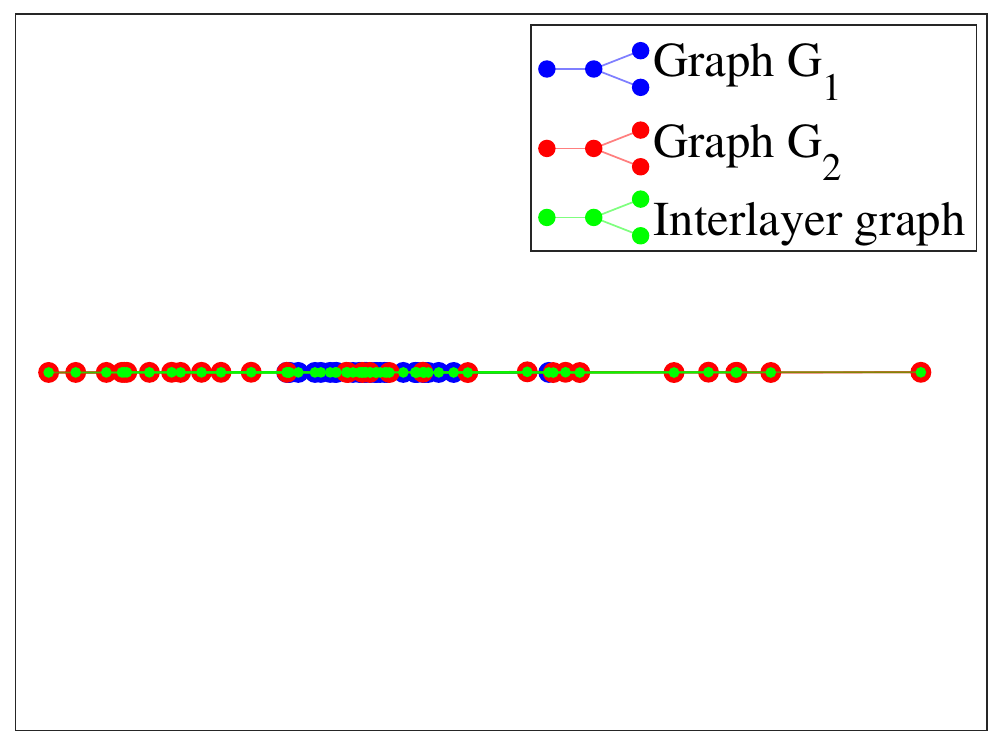}} \ \ \ \ \ \ \ 
	\subfloat[\label{fig:GapEmbVc50}]{\includegraphics[clip,width=.3\columnwidth]{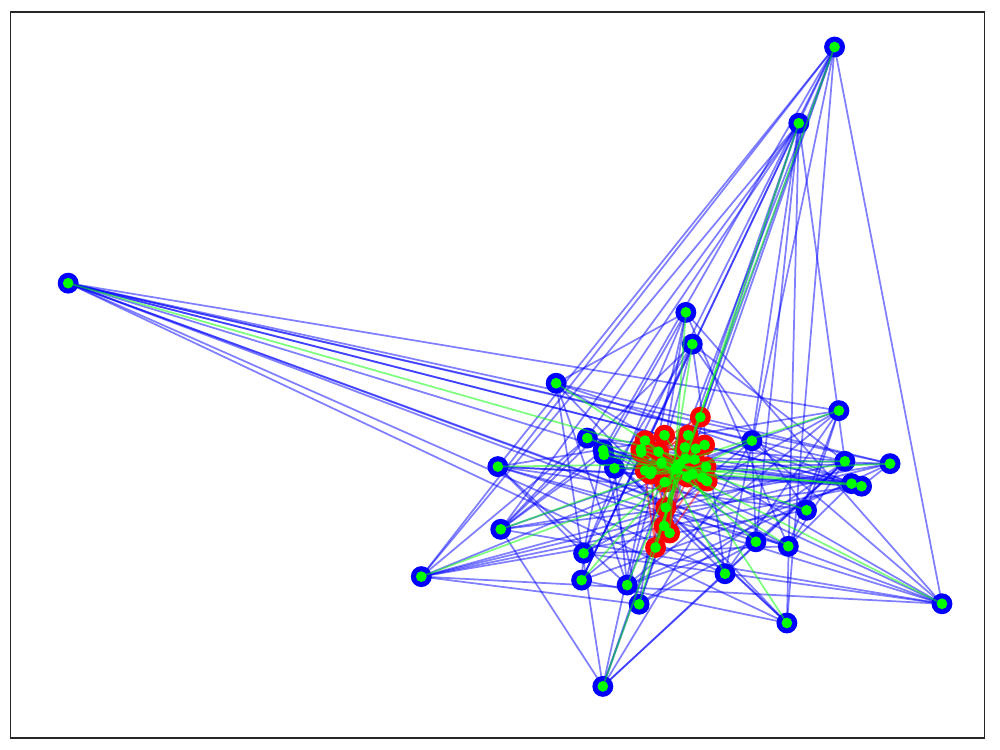}}
	
	\subfloat[\label{fig:GapEmbLam2c50}]{\includegraphics[clip,width=.3\columnwidth]{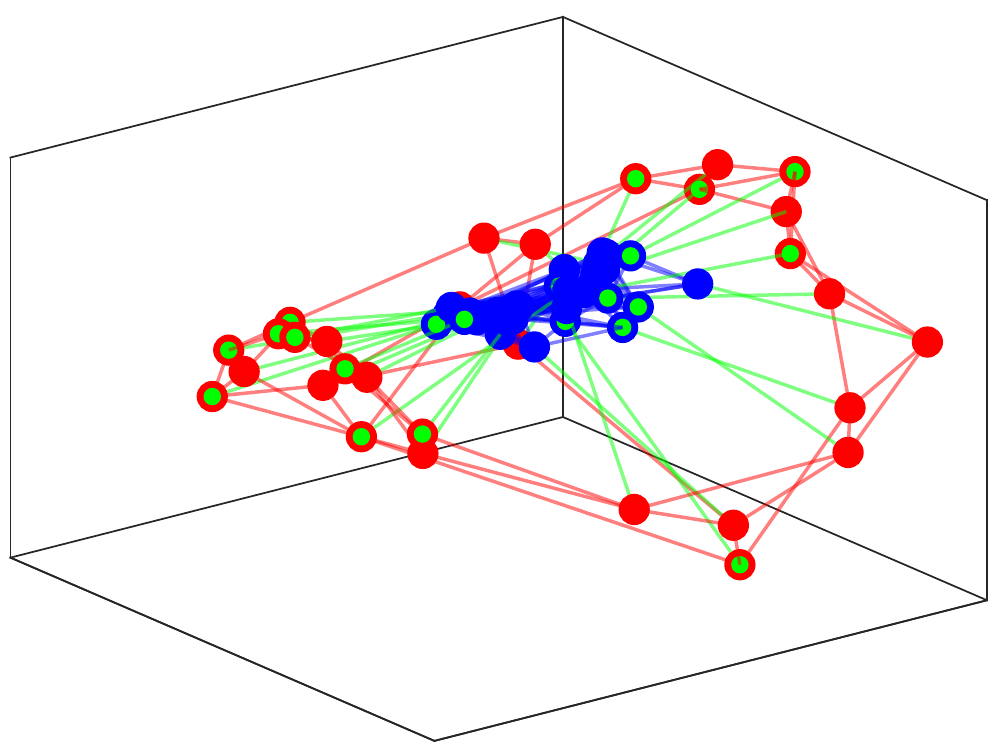}} \ \ \ \ \ \ \ 
	\subfloat[\label{fig:GapEmbLamnc50}]{\includegraphics[clip,width=.3\columnwidth]{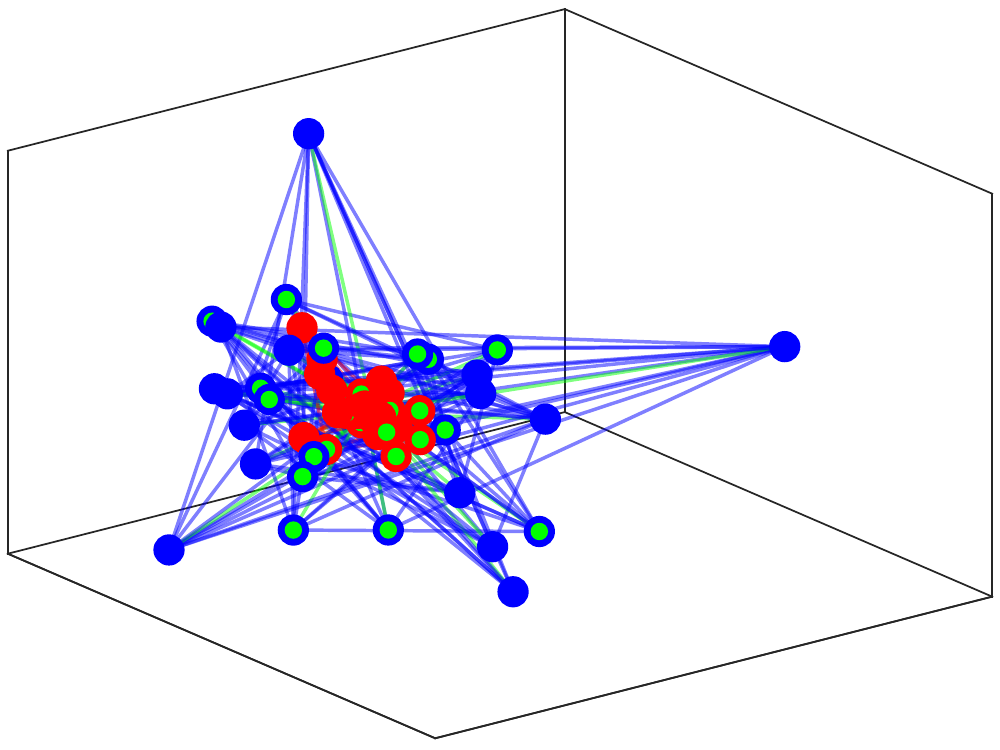}}
	
	\caption{(a) $u$- and (b) $v$-embedding corresponding to minimum spectral width, and embedding corresponding to (c) maximum $\lambda_2$ and (d) minimum $\lambda_n$ for a multiplex of two different Watts-Strogatz networks with 30 nodes for $c=50$.}
	\label{fig:GapEmbc50} 
\end{figure}

%
\subsection{Individual network models}\label{section:NetModels}
In the numerical simulations discussed throughout the paper, we use four different graph models for individual networks: \\

\textbf{Barab\'asi-Albert scale-free network (BA).} New nodes are attached to a specified number of already existing nodes in a preferential attachment fashion. For nodes number large enough, this method ensures the emergence of power-law behavior observed in many real-world networks \citep{barabasi1999emergence}.\\

\textbf{Erd\H{o}s-R\'enyi (ER).} Starting with a complete graph for the given number of nodes, the edges are randomly deleted according to a specified probability \citep{bollobas1998random}. \\

\textbf{Geometric network (Geo).} A set of nodes, picked randomly in a specified interval, are connected by an edge if the Euclidean distance is up to a definite value \citep{penrose2003random}. \\ 

\textbf{Watts-Strogatz (WS).} First, all nodes are connected to their immediate neighbors according to a fixed degree specified. Then, all existing links are rewired with a given probability, which produces graphs with low average hop count yet high clustering coefficient, which mimics the small-world property found in real-world networks \citep{watts1998collective}. \\

\subsection{Correlation with centrality measures}\label{sec:Correlation}
We have observed that maximum algebraic connectivity is achieved with uniform wights when $c\leq c^*$. For larger budgets the weight distribution is generally not uniform. To illustrate that optimal weights are not so simple as to be exclusively determined based on common centrality measures, we examine the weight distribution for maximizing algebraic connectivity corresponding to Figure \ref{fig:Emb3}. We examined several common centrality measures such as Degree, Eigenvector, and Page Rank, as well as Fiedler vector components, to evaluate nodes importance. Here, only results of degree centrality and Fiedler vector are shown in Figure \ref{fig:CentrCorltn}, since no additional results were observed by other measures.  Figures \ref{fig:Centr1} and \ref{fig:Centr2} report no correlation of optimal weights with degree centrality of individual network components. There is also no correlation in Figure \ref{fig:Fiedler2} between optimal weights and Fiedler vector components of the graph with larger connectivity; here $G_2$. However, some positive correlation is observed between optimal weights and Fiedler vector components of the graph with smaller connectivity; here $G_1$, i.e. the subgraph first unfolds in Figure \ref{fig:Emb3}. Therefore, for budget values beyond the threshold $c^*$, optimal weights prove positive correlation with Fiedler vector components of the subgraph with smaller algebraic connectivity so that, generally, the higher the absolute value of Fiedler vector component corresponding to a node in this subgraph, the larger the weight assigned to it; or, the smallest weights are assigned to those nodes with smallest Fiedler vector component and the largest weights are assigned to those nodes with largest Fiedler vector component . The layer with larger connectivity hardly plays role in this regard. We remind that this condition is manifested for budget values just above $c^*$ and when the two subgraphs hold very different algebraic connectivity values. Thus, with distancing from these two conditions correlation between optimal weights and Fiedler components of the subgraph with smaller algebraic connectivity dies away. Finally, we show in Figure \ref{fig:CentrCorltnLrg} the results corresponding to the larger network investigated in Figure \ref{fig:EmbLrg}. Again, only some positive correlation can be observed with Fiedler vector of subgraph $G_1$ with smaller algebraic connectivity. \\

\begin{figure}[!htb]
	\subfloat[\label{fig:Centr1}]{\includegraphics[clip,width=.25\columnwidth]{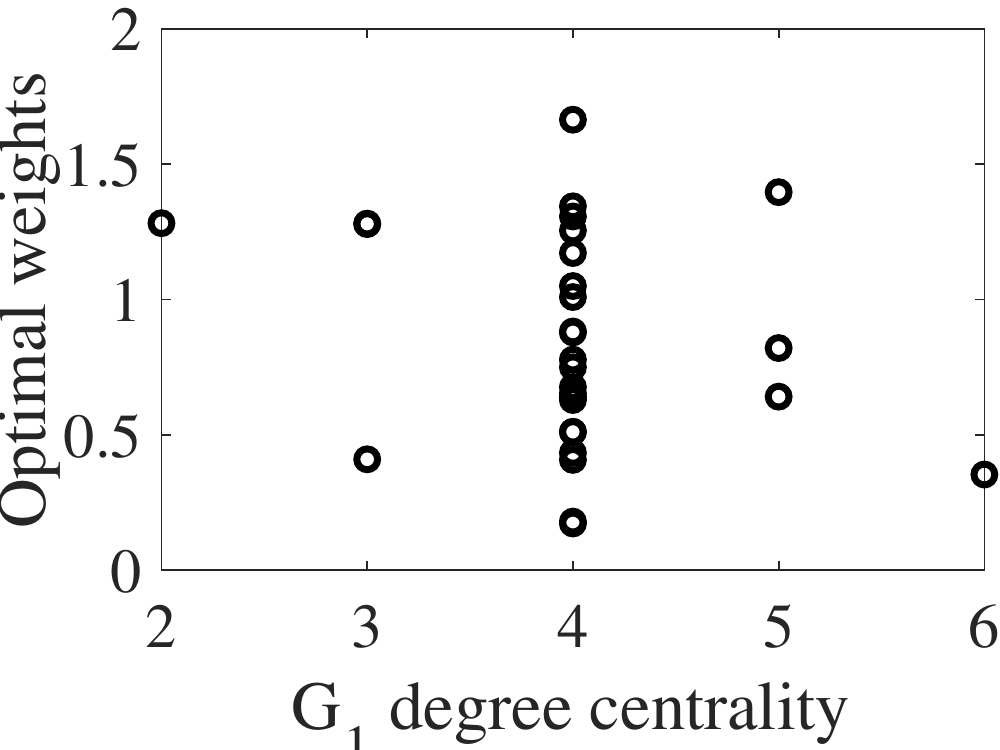}}	
	\subfloat[\label{fig:Centr2}]{\includegraphics[clip,width=.25\columnwidth]{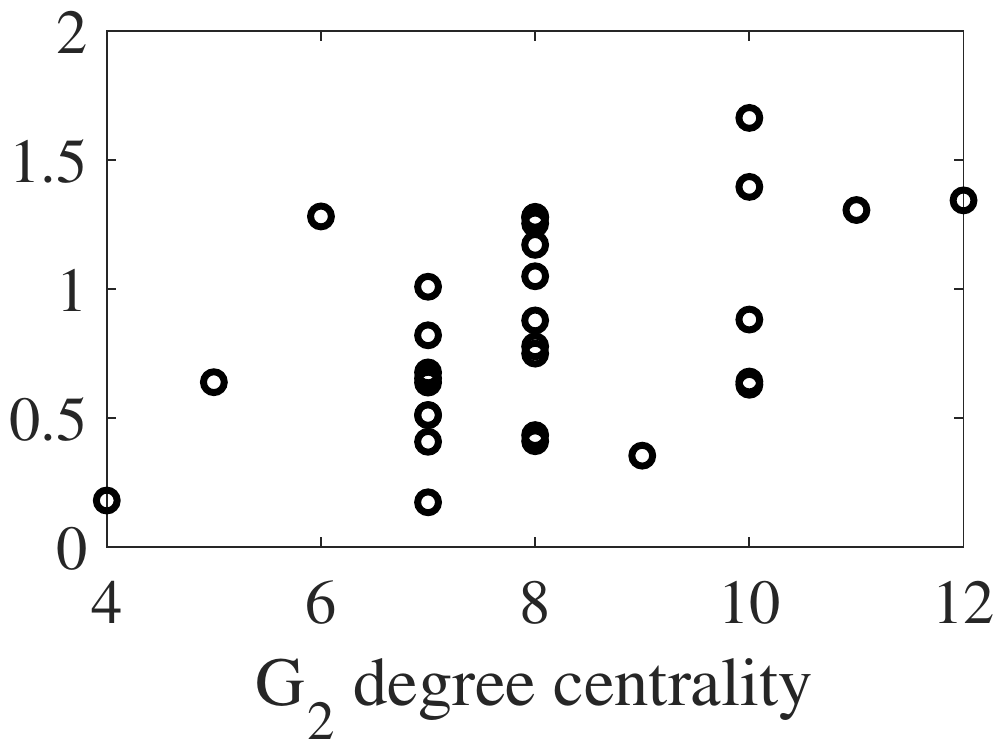}} 
	\subfloat[\label{fig:Fiedler1}]{\includegraphics[clip,width=.25\columnwidth]{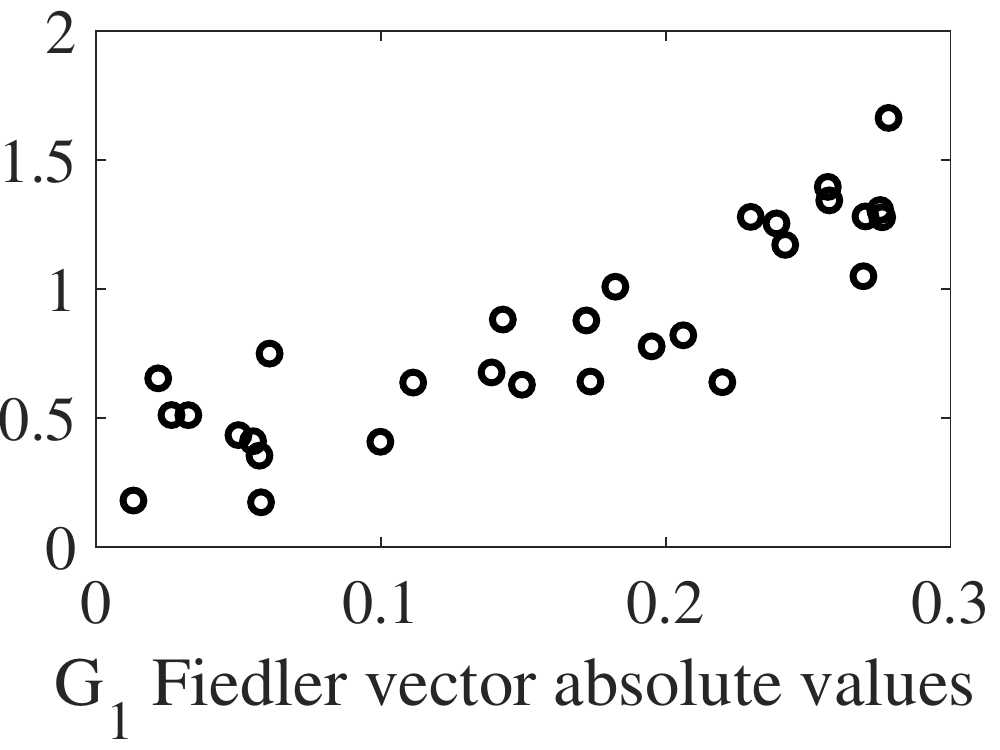}} 	\subfloat[\label{fig:Fiedler2}]{\includegraphics[clip,width=.25\columnwidth]{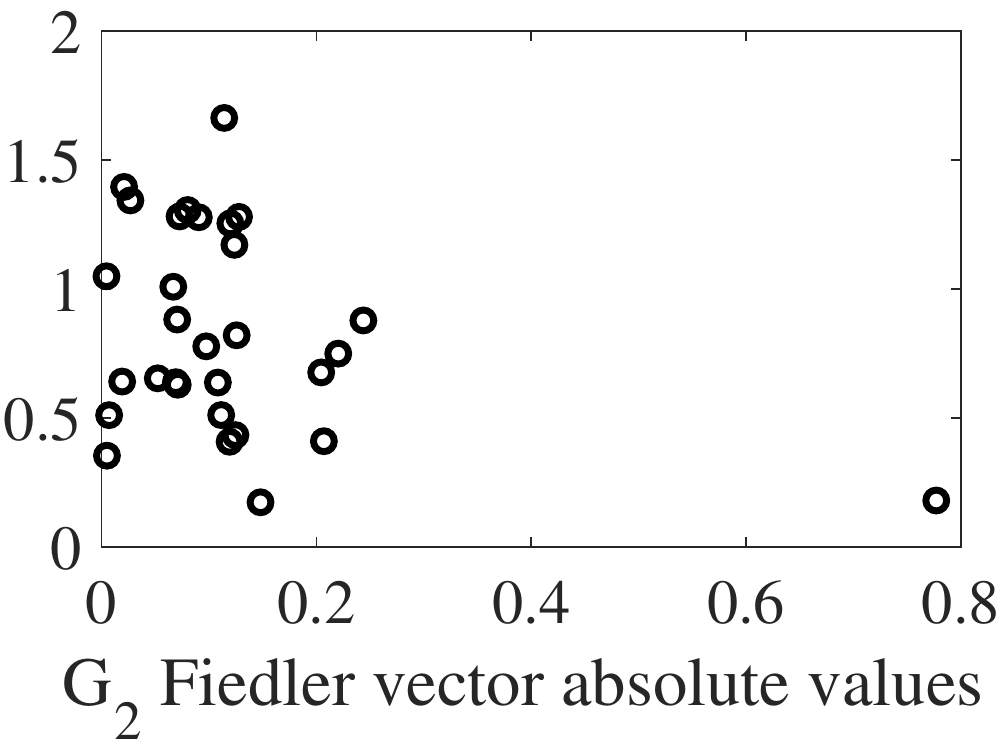}}
	
	\caption{Optimal weights corresponding to Figure \ref{fig:Emb3} when $c=10$ for different values of (a) degree centrality in $G_1$, (b) degree centrality in $G_2$, (c) absolute values of Fiedler vector components in $G_1$, and (d) absolute values of Fiedler vector components in $G_2$.}
	\label{fig:CentrCorltn} 
\end{figure}
\begin{figure}[!htb]
	\subfloat[\label{fig:Centr1Lrg}]{\includegraphics[clip,width=.25\columnwidth]{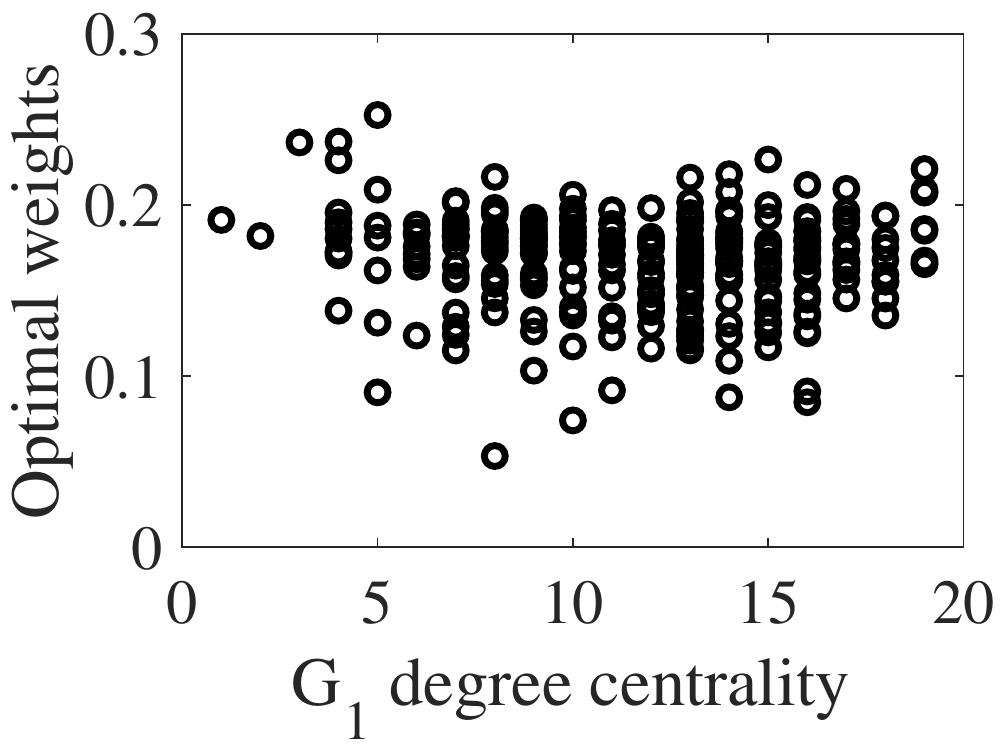}}	
	\subfloat[\label{fig:Centr2Lrg}]{\includegraphics[clip,width=.25\columnwidth]{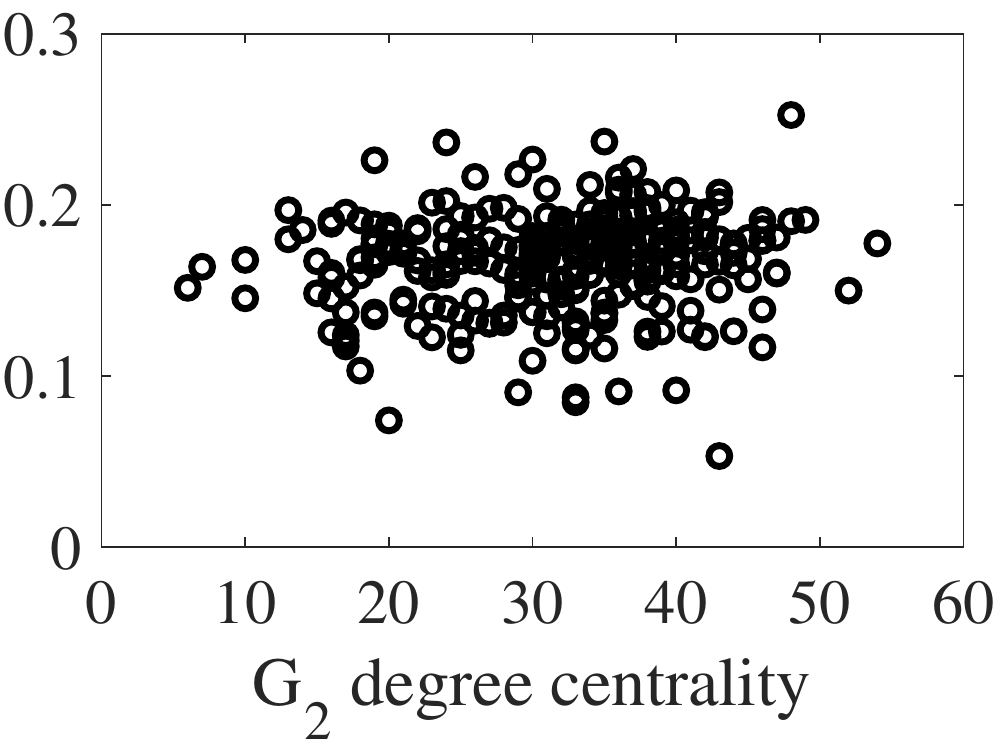}}
	\subfloat[\label{fig:Fiedler1Lrg}]{\includegraphics[clip,width=.25\columnwidth]{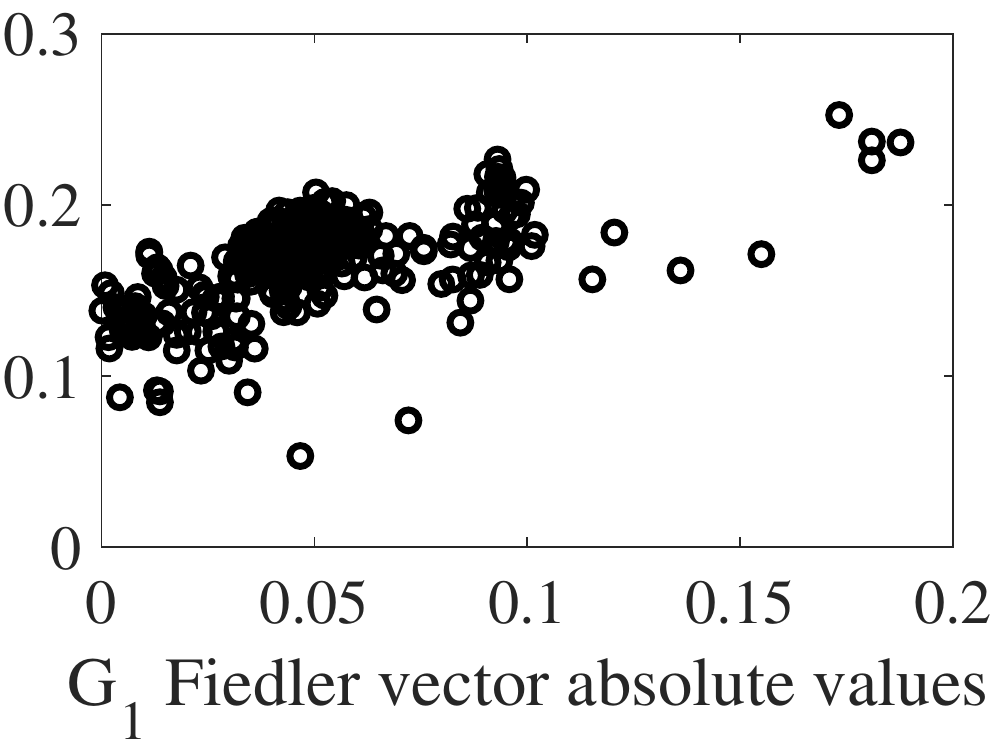}}	\subfloat[\label{fig:Fiedler2Lrg}]{\includegraphics[clip,width=.25\columnwidth]{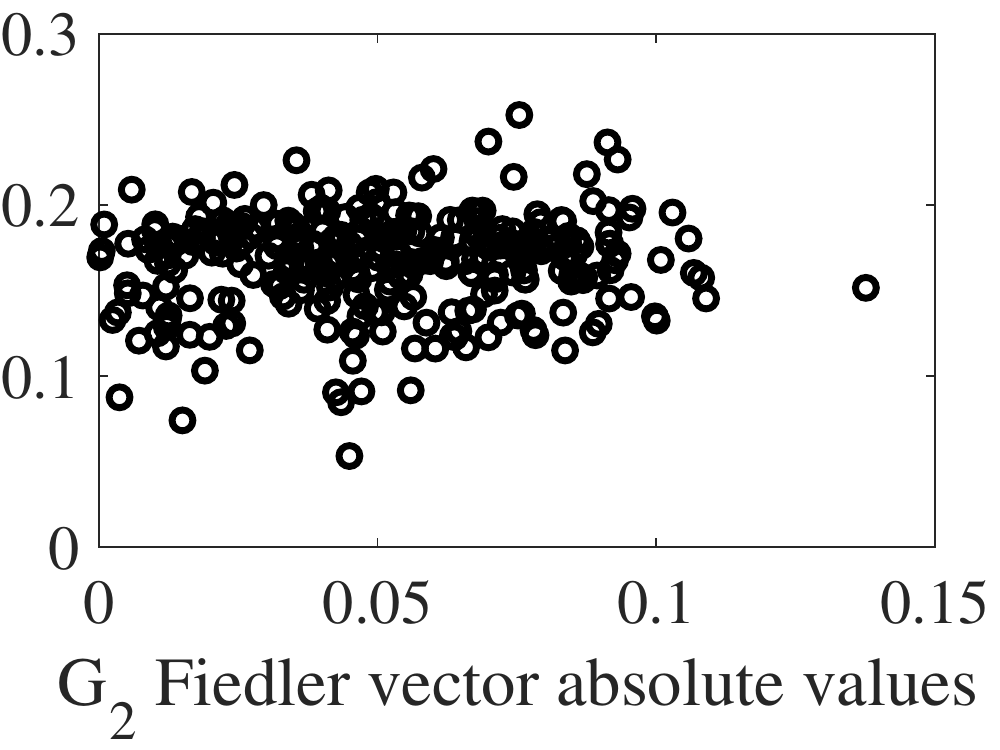}}
	
	\caption{Optimal weights corresponding to Figure \ref{fig:EmbLrg} when $c=50$ for different values of (a) degree centrality in $G_1$, (b) degree centrality in $G_2$, (c) absolute values of Fiedler vector components in $G_1$, and (d) absolute values of Fiedler vector components in $G_2$.}
	\label{fig:CentrCorltnLrg} 
\end{figure}

\end{document}